\documentclass[11pt]{book}

\newcommand{\mytitle}{
  A Framework for Specifying, Prototyping, and
  Reasoning about Computational Systems
}

\usepackage{simplethesis}
\usepackage{amsmath,amssymb,amsthm}
\usepackage{proof}
\usepackage{xspace}
\usepackage{tikz}
\usepackage{url}

\newcommand{\ignore}[1]{}

\newcommand{\lvl}{{\rm lvl}}
\newcommand{\supp}{{\rm supp}}
\newcommand{\htf}{{\rm ht}}

\newcommand{\hsl}[1]{\hbox{\sl #1}}
\newcommand{\abs}[2]{\hbox{\sl abs} \; #1 \; #2}
\newcommand{\abst}[3]{\hbox{\sl abs} \; #1 \; #2 \; #3}
\newcommand{\var}[1]{\hbox{\sl var} \; #1}
\newcommand{\app}[2]{\hbox{\sl app} \; #1 \; #2}
\newcommand{\arr}[2]{\hbox{\sl arr} \; #1 \; #2}
\newcommand{\arrow}[2]{\hbox{\sl arrow} \; #1 \; #2}
\newcommand{\all}[2]{\hbox{\sl all} \; #1 \; #2}
\newcommand{\sub}[2]{\hbox{\sl sub} \; #1 \; #2}
\newcommand{\bound}[2]{\hbox{\sl bound} \; #1 \; #2}
\newcommand{\assm}[2]{\hbox{\sl assm} \; #1 \; #2}
\newcommand{\ftop}{\hbox{\sl top}}
\newcommand{\wfty}[1]{\hbox{\sl wfty} \; #1}
\newcommand{\lambdat}[3]{\lambda #1\!:\!#2 . \ #3}
\newcommand{\of}[2]{\hbox{\sl of} \; #1 \; #2}
\newcommand{\typeof}[3]{\hbox{\sl typeof} \; #1 \; #2 \; #3}
\newcommand{\eval}[2]{\hbox{\sl eval} \; #1 \; #2}
\newcommand{\member}[2]{\hbox{\sl member} \; #1 \; #2}

\newcommand{\seq}[3]{\hbox{\sl seq}_{#1} \; #2 \; #3}
\newcommand{\prog}[2]{\hbox{\sl prog} \; #1 \; #2}
\newcommand{\ctx}[1]{\hbox{\sl ctx} \; #1}
\newcommand{\nat}[1]{\hbox{\sl nat} \; #1}
\newcommand{\ack}[3]{\hbox{\sl ack} \; #1 \; #2 \; #3}
\newcommand{\lt}[2]{\hbox{\sl lt} \; #1 \; #2}
\newcommand{\even}[1]{\hbox{\sl even} \; #1}
\newcommand{\odd}[1]{\hbox{\sl odd} \; #1}
\newcommand{\name}[1]{\hbox{\sl name} \; #1}
\newcommand{\fresh}[2]{\hbox{\sl fresh} \; #1 \; #2}

\newcommand{\dom}{\hbox{\sl dom}}
\newcommand{\tup}[1]{\langle #1\rangle}
\newcommand{\tridot}{\!\Vdash\!}
\newcommand{\tridotn}[1]{\!\Vdash_{\!#1}\!}
\newcommand{\new}{\reflectbox{\ensuremath{\mathsf{N}}}\xspace}

\newcommand{\spec}[3]{\hbox{\sl spec} \; #1 \; #2 \; #3}
\newcommand{\monoTy}[1]{\hbox{\sl monoTy} \; #1}
\newcommand{\polyTy}[1]{\hbox{\sl polyTy} \; #1}

\newcommand{\restrict}{\uparrow}
\newcommand{\CSNAS}{\hbox{\sl CSNAS}}
\newcommand{\CSU}{\hbox{\sl CSU}}

\newcommand{\fsub}{\mbox{\tt <\!:}}
\newcommand{\term}[1]{\hbox{\sl term} \; #1}
\newcommand{\pathp}[2]{\hbox{\sl path} \; #1 \; #2}
\newcommand{\uabs}[1]{\hbox{\sl abs} \; #1}
\newcommand{\pleft}[1]{\hbox{\sl left} \; #1}
\newcommand{\pright}[1]{\hbox{\sl right} \; #1}
\newcommand{\bnd}[1]{\hbox{\sl bnd} \; #1}
\newcommand{\ctxs}[2]{\hbox{\sl ctxs} \; #1 \; #2}
\newcommand{\add}[3]{\hbox{\sl add} \; #1 \; #2 \; #3}
\newcommand{\hodb}[3]{\hbox{\sl ho2db} \; #1 \; #2 \; #3}
\newcommand{\depth}[2]{\hbox{\sl depth} \; #1 \; #2}

\newcommand{\dapp}[2]{\hbox{\sl dapp} \; #1 \; #2}
\newcommand{\dvar}[3]{\hbox{\sl dvar} \; #1 \; #2 \; #3}
\newcommand{\mle}[2]{\hbox{\sl le} \; #1 \; #2}
\newcommand{\dctx}[2]{\hbox{\sl dctx} \; #1 \; #2}
\newcommand{\type}[1]{\hbox{\sl type} \; #1}
\newcommand{\reduce}[2]{\hbox{\sl reduce} \; #1 \; #2}
\newcommand{\subst}[3]{\hbox{\sl subst} \; #1 \; #2 \; #3}
\newcommand{\step}[2]{\hbox{\sl step} \; #1 \; #2}
\newcommand{\sn}[1]{\hbox{\sl sn} \; #1}
\newcommand{\neutral}[1]{\hbox{\sl neutral} \; #1}

\newcommand{\simp}[2]{\hbox{\sl sim} \; #1 \; #2}

\newcommand{\FOL}{FO\lambda}
\newcommand{\N}{{\rm I} \! {\rm N}}
\newcommand{\FOLDN}{\ensuremath{\FOL^{\Delta\N}}\xspace}
\newcommand{\foldnb}{\ensuremath{FO\lambda^{\Delta\nabla}}\xspace}
\newcommand{\G}{$\mathcal{G}$\xspace}
\newcommand{\logic}{\G}
\newcommand{\LG}{$LG^\omega$\xspace}
\newcommand{\ie}{{\em i.e.}}
\newcommand{\eg}{{\em e.g.}}

\newcommand{\hh}{$hH^2$\xspace}

\newcommand{\TRUE}{\mbox{TRUE}}
\newcommand{\OR}{\mbox{OR}}
\newcommand{\AND}{\mbox{AND}}
\newcommand{\AUGMENT}{\mbox{AUGMENT}}
\newcommand{\GENERIC}{\mbox{GENERIC}}
\newcommand{\INSTANCE}{\mbox{INSTANCE}}
\newcommand{\BACKCHAIN}{\mbox{BACKCHAIN}}

\newcommand{\lra}{\longrightarrow}

\newcommand{\defL}{\hbox{\sl def}\mathcal{L}}
\newcommand{\defR}{\hbox{\sl def}\mathcal{R}}
\newcommand{\cL}{\hbox{\sl c}\mathcal{L}}
\newcommand{\unrhdL}{\unrhd\mathcal{L}}
\newcommand{\unrhdR}{\unrhd\mathcal{R}}
\newcommand{\cut}{\hbox{\sl cut}}
\newcommand{\botL}{\bot\mathcal{L}}
\newcommand{\topR}{\top\mathcal{R}}
\newcommand{\lorL}{\lor\mathcal{L}}
\newcommand{\lorR}{\lor\mathcal{R}}
\newcommand{\landL}{\land\mathcal{L}}
\newcommand{\landR}{\land\mathcal{R}}
\newcommand{\supsetL}{\supset\!\mathcal{L}}
\newcommand{\supsetR}{\supset\!\mathcal{R}}
\newcommand{\forallL}{\forall\mathcal{L}}
\newcommand{\forallR}{\forall\mathcal{R}}
\newcommand{\nablaL}{\nabla\mathcal{L}}
\newcommand{\nablaR}{\nabla\mathcal{R}}
\newcommand{\existsL}{\exists\mathcal{L}}
\newcommand{\existsR}{\exists\mathcal{R}}

\newcommand{\IL}{\mathcal{IL}}
\newcommand{\CIR}{\mathcal{CIR}}

\newcommand{\mueq}{\stackrel{\mu}{=}}
\newcommand{\nueq}{\stackrel{\nu}{=}}

\newcommand{\cas}[1]{[\![ #1 ]\!]}
\newcommand{\enc}[1]{\ulcorner #1 \urcorner}

\newtheorem{definition}{Definition}[section]
\newtheorem{theorem}[definition]{Theorem}
\newtheorem{lemma}[definition]{Lemma}
\newtheorem{corollary}[definition]{Corollary}
\newtheorem{notation}[definition]{Notation}

\begin{document}

\Author{Andrew Jude Gacek}
\Title{\mytitle}
\Month{September}
\Year{2009}
\Advisor{Gopalan Nadathur}
\Degree{DOCTOR OF PHILOSOPHY}
\degreelevel{doctoral}

\prelimpages


\titlepage

\copyrightpage

\setcounter{page}{1}

\acknowledgments{Many people have supported me during the development of this thesis
and I owe them all a debt of gratitude.

Firstly, I would like to thank my advisor Gopalan Nadathur for his
patience and guidance which have played a significant part in my
development as a researcher. His willingness to share his opinions on
everything from academic life to playing squash has helped me to
develop a perspective and to have fun while doing this. I look forward
to continuing my interactions with him far into the future.

I am grateful to Dale Miller for sharing with me an excitement for
research and an appreciation of the uncertainty that precedes
understanding. I have never met anybody else who so enjoys when things
seem amiss, because he knows that a new perspective will eventually
emerge and bring clarity.

This thesis has been heavily influenced by the time I have spent
working with Alwen Tiu, David Baelde, Zach Snow, and Xiaochu Qi.
Understanding their work has given me a deeper understanding of my own
research and its role in the bigger picture. I am thankful for the
time I have had with each and every one of them.

I have been inspired in my studies by my friends Mike Whalen and Jared
Davis. Their intelligence, drive, and curiosity are remarkable and have
challenged me to work harder so that I may hope to be considered their
equals.

I also want to thank my committee members Eric Van Wyk, Mats Heimdahl,
and Wayne Richter for their time and for their guidance in my research
career.

Finally, I am thankful to the many people who have supported me long
before this thesis began. I want especially to thank my wife, Ann,
for her patience, understanding, and love, and my
parents for their never-ending encouragement and support. To the rest
of my family and friends: I thank you all!

\vfill
\vfill
\vfill
\vfill
\vfill
\vfill
\vfill
\vfill
\vfill
\vfill
\vfill
\vfill
\vfill
\vfill
\vfill
\vfill
\vfill
\vfill
\vfill
\vfill
\vfill
\vfill
\vfill
\vfill
\vfill
\vfill

Work on this thesis has been partially funded by the NSF Grants
CCR-0429572 and CCF-0917140. Support has also been received from a
research contract from Boston Scientific and from funds provided by
the Institute of Technology and the Department of Computer Science and
Engineering at the University of Minnesota. Opinions, findings, and
conclusions or recommendations expressed in this thesis should be
understood as mine. In particular, they do not necessarily reflect the
views of the National Science Foundation.


}

\abstract{A major motivation for formal systems such as programming languages
and logics is that they support the ability to perform computations in
a safe, secure, and understandable way. A considerable amount of
effort has consequently been devoted to developing tools and
techniques for structuring and analyzing such systems. It is natural
to imagine that research in this setting might draw benefits from its
own labor. In particular, one might expect the study of formal systems
to be conducted with the help of languages and logics designed for
such study. There are, however, significant problems that must be
solved before such a possibility can be made a practical reality. One
such problem arises from the fact that formal systems often have to
treat objects such as formulas, proofs, programs, and types that have
an inherent binding structure. In this context, it is necessary to
provide a flexible and logically precise treatment of related notions
such as the equality of objects under the renaming of bound variables
and substitution that respects the scopes of binders; there is
considerable evidence that if such issues are not dealt with in an
intrinsic and systematic way, then they can overwhelm any relevant
reasoning tasks. For a logic to be useful in this setting, it must
also support rich capabilities such as those for inductive reasoning
over computations that are described by recursion over syntax.

This thesis concerns the development of a framework that facilitates
the design and analysis of formal systems. Specifically, this
framework is intended to provide 1) a specification language which
supports the concise and direct description of a system based on its
informal presentation, 2) a mechanism for animating the specification
language so that descriptions written in it can quickly and
effectively be turned into prototypes of the systems they are about,
and 3) a logic for proving properties of descriptions provided in the
specification language and thereby of the systems they encode. A
defining characteristic of the proposed framework is that it is based
on two separate but closely intertwined logics. One of these is a
specification logic that facilitates the description of computational
structure while the other is a logic that exploits the special
characteristics of the specification logic to support reasoning about
the computational behavior of systems that are described using it.
Both logics embody a natural treatment of binding structure by using
the $\lambda$-calculus as a means for representing objects and by
incorporating special mechanisms for working with such structure. By
using this technique, they lift the treatment of binding from the
object language into the domain of the relevant meta logic, thereby
allowing the specification or analysis components to focus on the more
essential logical aspects of the systems that are encoded.

One focus of this thesis is on developing a rich and expressive
reasoning logic that is of use within the described framework. This
work exploits a previously developed capability of definitions for
embedding recursive specifications into the reasoning logic; this
notion of definitions is complemented by a device for a case-analysis
style reasoning over the descriptions they encode. Use is also made of
a special kind of judgment called a generic judgment for reflecting
object language binding into the meta logic and thereby for reasoning
about such structure. Existing methods have, however, had a
shortcoming in how they combine these two devices. Generic judgments
lead to the introduction of syntactic objects called nominal constants
into formulas and terms. The manner in which such objects are
introduced often ensures that they satisfy certain properties which
are necessary to take note of in the reasoning process. Unfortunately,
this has heretofore not been possible to do. To overcome this problem,
we introduce a special binary relation between terms called {\it
  nominal abstraction} and show this can be combined with definitions
to encode the desired properties. The treatment of definitions is
further enriched by endowing them with the capability of being
interpreted inductively or co-inductively. The resulting logic is
shown to be consistent and examples are presented to demonstrate its
richness and usefulness in reasoning tasks.

This thesis is also concerned with the practical application of the
logical machinery it develops. Specifically, it describes an
interactive, tactic-style theorem prover called Abella that realizes
the reasoning logic. Abella embodies the use of lemmas in proofs and
also provides intuitively well-motivated tactics for inductive and
co-inductive reasoning. The idea of reasoning using two-levels of
logic is exploited in this context. This form of reasoning, pioneered
by McDowell and Miller, embeds the specification logic explicitly into
the reasoning logic and then reasons about particular specifications
through this embedding. The usefulness of this approach is
demonstrated by showing that general properties can be proved about
the specification logic and then used as lemmas to simplify the
overall reasoning process. We use these ideas together with Abella to
develop several interesting and challenging proofs. The examples
considered include ones in the recently proposed POPLmark challenge
and a formalization of Girard's proof of strong normalization for the
simply-typed $\lambda$-calculus. We also explore the notion of
adequacy that relates theorems proved using Abella to the properties
of the object systems that are ultimately of primary interest.


}

\setcounter{tocdepth}{2}
\tableofcontents

\listoffigures

\textpages


\chapter{Introduction}
\label{ch:introduction}

In this thesis we are interested in developing a framework for
mechanizing the specification and prototyping of formal systems and
also the process of reasoning about the properties of such systems
based on their specifications. The formal systems that are of interest
to us are ones that concern computation: for instance, they might
characterize evaluation and typing in a programming language,
provability in a logic, or behavior in a concurrency system. Formal
systems of these kinds typically manipulate syntactically complex
objects such as formulas, proofs, and programs. Mechanized
specification and reasoning about such systems has proven difficult to
achieve through the use of traditional tools and techniques
\cite{aydemir05tphols}. We propose a framework here which overcomes
these difficulties and, through this process, brings the benefits of
automation and computer-aided verification to bear on the development
of these types of systems. In particular, this thesis proposes a
framework that facilitates the development of such systems by
providing 1) a specification language which supports the concise and
direct description of a system based on its informal presentation, 2)
a mechanism for animating the specification language so that
descriptions written in it can quickly and effectively be turned into
prototypes of the systems they are about, and 3) a logic for proving
properties of descriptions provided in the specification language and
thereby of the systems they encode.

\section[A Specification, Prototyping, and Reasoning Framework]
        {A Framework for Specification, Prototyping, and Reasoning}


The formal systems that we would like to specify and reason about are
all characterized by the fact that they are based on syntactic
expressions and their behavior is determined by the structure of these
expressions. For brevity we will refer to such systems simply as {\em
  computational systems}. A popular approach to describing such
systems starts by describing various possible judgments over the
syntax of the systems. Then rule schemas are presented where each
schema allows a judgment to be formed from other judgments, often in a
compositional manner. Finally, instances of these rules schemas are
chained together into a {\em derivation} where each premise judgment
of a rule instance is the consequence judgment of another rule
instance. A judgment is said to hold if and only if it is the final
conclusion judgment of derivation. Thus one can understand the
behavior of a system by studying the rule schemas for forming
judgments about the system. This approach to describing a
computational system is known as structural operational semantics
\cite{plotkin81}.

Structural operational semantics descriptions have a logical flavor in
that one simply describes a few declarative rules for manipulating
syntax and these are orchestrated together to reach larger conclusions
about the behavior of the system. The framework we propose allows for
such descriptions to be formally specified via a {\em specification
  logic} similar to the logic of Horn clauses. We call such an
encoding of a computational system into this logic a {\em
  specification}. More specifically, the system syntax is encoded as
specification logic terms, judgments are encoded as specification
logic atomic formulas, and rules are encoded as richer specification
logic formulas. Derivations of atomic formulas within the
specification logic then correspond to derivations in structural
operational semantics descriptions. Thus we can study a wide variety
of computational systems via a study of the specification logic.


In order to interact with computational systems, our proposed
framework supports prototyping based on the system specification. This
prototyping is driven directly by the formal specification, by giving
a computational interpretation of the specification logic in the same
sense that Prolog provides a computational interpretation to the logic
of Horn clauses. This eliminates the need for the framework user to
manually develop a prototype based on the specification, thus avoiding
a source of potential errors. Also, as the specification evolves this
ensures that the prototype remains faithful to the current
specification.


The specification of a computational system consists of local rules
about the system behavior, but one is often interested in global
properties of the system. For example, programming language designers
often describe the rules for evaluation and typing judgments for a
language and then prove properties which relate the two judgments
together such as that the evaluation judgment preserves the typing
judgment. Such properties ensure that the language is well-behaved
relative to programmers' expectations. In order to prove these
properties about a structural operational semantics description one
must be able to analyze the ways in which derivations may be formed.
In the example of proving that evaluation preserves typing, one may
inductively analyze the possible forms that a derivation of an
evaluation judgment may have and for each possibility argue that the
typing judgment for the evaluated term can be restructured into a
typing judgment for the term which results from the evaluation.

The proposed framework allows for reasoning over structural
operational semantics descriptions via a {\em meta-logic}. The
meta-logic contains mechanisms such as induction and co-induction
which are essential to sophisticated reasoning. The meta-logic also
contains a mechanism called {\em definitions} which allows one to
connect atomic judgments to descriptions of behavior in a ``closed
world'' fashion. Thus, it allows for both positive reasoning, \ie,
showing that a judgment holds, and negative reasoning, \ie, analyzing
why a judgment holds. This allows one to easily carry out the case
analysis-like reasoning described in the example of typing and
evaluation.

We refer to this second logic as a meta-logic because, in our
approach, we use it to encode the entire specification logic, rather
than to encode each specification independently. We then reason about
particular specifications by reasoning about their descriptions in the
specification logic. This style of reasoning was pioneered by McDowell
and Miller \cite{mcdowell02tocl} and is called the {\em two-level
  logic approach to reasoning}. One of its benefits is that it allows
us to reason over specifications exactly as they are written and used
in prototyping. Another is that it allows properties of the
specification logic to be formally proven once and for all in the
meta-logic and then used freely during reasoning. In practice, many
tedious substitution lemmas proven about particular specifications are
subsumed by these more general properties of the specification logic.


A pervasive issue in the computational systems of interest is dealing
with the binding structure of syntactic objects. For example, to
develop a programming language we need to formalize the rules for
binding local variables which requires a systematic way 1) to
associate variable occurrences with their binders, 2) to treat objects
which differ only in the name of bound variables as being identical,
and 3) to realize a logically correct notion of capture-avoiding
substitution which respects the binding structure of objects. Our
proposed framework addresses all of these issue by mapping the binding
structure of objects into the abstraction mechanism of the
meta-language, \ie, the specification logic during specification and
the meta-logic during reasoning. This is called a {\em higher-order
  abstract syntax} representation \cite{miller87slp, pfenning88pldi}.
In this way, the meta-language notion of binding describes how
variables occurrences are associated to the binder, the meta-language
notion of equality provides a way to identify objects differing only
in the names of bound variables, and meta-language function
application and reduction realize capture-avoiding substitution.

\section{An Illustration of the Application of the Framework}
\label{sec:example}

Throughout this thesis we will use the example of the simply-typed
$\lambda$-calculus \cite{church41,barendregt84}. This is a compact
example which highlights many of the essential difficulties involved
in specifying, prototyping, and reasoning about a computational system
with binding. Anytime we use such a system as the focus of study we
shall refer to it as the {\em object language} or the {\em object
  logic}.

The syntax of the simply-typed $\lambda$-calculus is made up of two
classes of expressions called {\em types} and {\em pre-terms} which
are defined, respectively, by the following grammar rules.
\begin{align*}
a &::= i \mid a \to a &
t &::= x \mid (\lambdat x a t) \mid (t\ t)
\end{align*}
Here $x$ is variable occurrence and in the expression $(\lambdat x a
t)$ the $x$ is to be considered bound within the expression $t$. We
assume the standard notions of binding including free and bound
variables, equivalence under renaming of bound variables, and a notion
of capture-avoiding substitution denoted by $t[x := s]$. Note,
however, that when one formally specifies this system within a
framework, these notions will need to be dealt with somehow. We will
denote types using variables named $a$, $b$, $c$, and $d$, pre-terms
using variables named $m$, $n$, $r$, $s$, $t$, and $v$, and object
language variables using $x$, $y$, and $z$.

\begin{figure}[t]
\begin{align*}
\infer
 {(\lambdat x a r) \Downarrow (\lambdat x a r)}
 {}
&&
\infer
 {(m\ n) \Downarrow v}
 {m \Downarrow (\lambdat x a r) &
  r[x := n] \Downarrow v}
\end{align*}
\caption{Evaluation in the simply-typed $\lambda$-calculus}
\label{fig:stlc-eval}
\end{figure}


We define a notion of big-step call-by-name weak reduction which we
call simply {\em evaluation}. This is denoted by the judgment $t
\Downarrow v$ which can be read as ``$t$ evaluates to $v$.'' The rules
for forming derivations of this judgment are presented in
Figure~\ref{fig:stlc-eval}.

\begin{figure}[t]
\begin{align*}
\infer
 {\Gamma \vdash x : a}
 {x : a \in \Gamma}
&&
\infer[x\notin\dom(\Gamma)]
  {\Gamma \vdash (\lambdat x a r) : a\to b}
  {\Gamma, x : a \vdash r : b}
&&
\infer
  {\Gamma \vdash (m\ n) : b}
  {\Gamma \vdash m : a \to b & \Gamma \vdash n : a}
\end{align*}
\caption{Typing in the simply-typed $\lambda$-calculus}
\label{fig:stlc-typing}
\end{figure}

We define a notion of typing via the judgment $\Gamma \vdash t : a$
which can be read as ``$t$ has type $a$ relative to the context
$\Gamma$.'' Here $\Gamma$ is called a {\em typing context} and is
described by the following grammar.
\begin{equation*}
\Gamma ::= \cdot \mid \Gamma, x : a
\end{equation*}
We will write a context of the form $\cdot, x_1 : a_1, \ldots, x_n :
a_n$ simply as $x_1 : a_1, \ldots, x_n : a_n$. We define $\dom(x_1 :
a_1, \ldots, x_n : a_n)$ as $\{x_1,\ldots,x_n\}$. In $\Gamma, x : a$
we require that $x\notin\dom(\Gamma)$. We satisfy this restriction by
renaming bound variables as needed. The rules for forming derivations
of the typing judgment are presented in Figure~\ref{fig:stlc-typing}.
If $t$ is a pre-term such that there exists a type $a$ for which
$\Gamma \vdash t : a$ holds, then we call $t$ a {\em term}.

\begin{figure}[t]
\centering
\begin{align*}
&\forall a, m. [\eval {(\abs a m)} {(\abs a m)}] \\[3pt]
&\forall m, a, r, n, v.[\eval m (\abs a r) \supset \eval {(r\ n)} v \supset \eval {(\app m n)} {v}]\\[10pt]
&\forall m, a, b, n.[\of m (\arr a b) \supset \of n a
\supset \of{(\app m n)} b] \\[3pt]
&\forall a, r, b.[(\forall x. \of x a \supset \of{(r\
  x)}{b}) \supset \of{(\abs a r)}{(\arr a b)}]
\end{align*}
\caption{A Horn clause-like encoding of evaluation and typing}
\label{fig:stlc-horn-like}
\end{figure}

We can now think of encoding the simply-typed $\lambda$-calculus into
our specification logic. This begins with the constructors $i$ and
{\sl arr} for representing the base and arrow types. We also use the
constructors {\sl app} and {\sl abs} for representing applications and
abstractions. Using a higher-order abstract syntax encoding there is
no constructor for variables, and instead the {\sl abs} constructor
takes two arguments: 1) the type of the abstracted variable and 2) a
specification logic abstraction representing the body. For example,
the object language term $(\lambdat x i (\lambdat y i x))$ is denoted by
$(\abs i (\lambda x. \abs i (\lambda y. x)))$ where these latter
$\lambda$s are specification logic abstractions.

We introduce the specification logic predicates {\sl eval} and {\sl
  of} for representing evaluation and typing judgments respectively.
Assuming a Horn clause-like specification logic, the rules for forming
evaluation and typing judgments are encoded into the specification
logic formulas shown in Figure~\ref{fig:stlc-horn-like}. This
specification uses various features of the specification logic which
go beyond simple Horn clauses such as function application for
realizing capture-avoiding substitution, universal quantification to
avoid explicit side-conditions, and specification logic hypotheses for
representing typing contexts. The complete details of this
specification are presented in Chapter~\ref{ch:specification-logic}.
For now it is sufficient to appreciate that the structural operational
semantics description of the simply-typed $\lambda$-calculus can be
encoded very directly into the specification logic. Moreover, a
Prolog-like operational interpretation of proof search for the
specification logic yields a prototype for our specification.

Returning to the original structural operational semantics description
of the simply-typed $\lambda$-calculus for the moment, let us think of
proving some global property of the system. One such property of
interest is that evaluation preserves the type of a term, called the
{\em type preservation} property. Let us consider how such a property
can be proved in an informal, mathematical setting. We might proceed
by first showing the auxiliary properties of typing judgments that are
contained in the following two lemmas.

\begin{lemma}
\label{lem:stlc-perm}
If $\Gamma \vdash t : a$ and $\Gamma'$ is a permutation of $\Gamma$,
then $\Gamma' \vdash t : a$. Moreover, the derivations have the same
height.
\end{lemma}
\begin{proof}
The proof is by induction on the height of the derivation of $\Gamma
\vdash t : a$.
\end{proof}

\begin{lemma}
\label{lem:stlc-sub}
If $\Gamma, x : a \vdash t : b$ and $\Gamma \vdash s : a$ then $\Gamma
\vdash t[x := s] : b$.
\end{lemma}
\begin{proof}
The proof is by induction on the height of the derivation of $\Gamma,
x : a \vdash t : b$. In the case where $t$ is an abstraction we use
Lemma~\ref{lem:stlc-perm} to permute the assumption $x:a$ to the end
of the context.
\end{proof}

We can now state and prove the main property of interest.

\begin{theorem}
\label{thm:stlc-sr}
If $t \Downarrow v$ and $\vdash t : a$ then $\vdash v : a$.
\end{theorem}
\begin{proof}
The proof is by induction on the height of the derivation of $t
\Downarrow v$.

{\em Base case.} If the derivation has height one then it must end
with the following.
\begin{equation*}
\infer
 {(\lambdat x a r) \Downarrow (\lambdat x a r)}
 {}
\end{equation*}
Then $t = v$ and the result is trivial.

{\em Inductive case.} If the derivation has height greater than one,
then it must end with the following.
\begin{equation*}
\infer
 {(m\ n) \Downarrow v}
 {m \Downarrow (\lambdat x b r) &
  r[x := n] \Downarrow v}
\end{equation*}
Here $t = (m\ n)$ and we have shorter derivations of $m
\Downarrow (\lambdat x b r)$ and $r[x := n] \Downarrow v$. By
assumption we know that $\vdash (m\ n) : a$ holds which means
it has a derivation which must end with
\begin{equation*}
\infer
 {\vdash (m\ n) : a}
 {\vdash m : c \to a &
  \vdash n : c}
\end{equation*}
for some type $c$. Now we can apply the inductive hypothesis to $m
\Downarrow (\lambdat x b r)$ and $\vdash m : c \to a$ to obtain
a derivation of $\vdash (\lambdat x b r) : c \to a$. Then it must
be that $b = c$ and this derivation ends with the following rule.
\begin{equation*}
\infer
 {\vdash (\lambdat x b r) : b \to a}
 {x : b \vdash r : a}
\end{equation*}
By Lemma~\ref{lem:stlc-sub} we have a derivation of $\vdash r[x
:= n] : a$. Finally, we use the inductive hypothesis on $r[x := n]
\Downarrow v$ and this typing judgment to conclude $\vdash v : a$.
\end{proof}

Our objective is to carry out the style of reasoning described above
in a formalized, computer-supported way. The framework that we will
develop in this thesis will support such an ability. The key to doing
this is designing a meta-logic for reasoning directly about the
specification logic and, in this particular instance, the descriptions
of evaluation and typing that have been encoded in it. The meta-logic
that we will describe will allow the specification logic to be encoded
as a definition in it, which then leads to the ability to reason,
within the meta-logic, about the structure of specification logic
derivations. Since these derivations have a close correspondence to
the structural operational semantics derivations, a reasoning process
very similar to that in Theorem~\ref{thm:stlc-sr} can be carried out
within the meta-logic. Moreover, Lemmas \ref{lem:stlc-sub} and
\ref{lem:stlc-perm} turn out to be instances of more general
properties of the specification logic, and thus one can essentially
obtain these results for free.

\section{The Contributions of this Thesis}
\label{sec:contributions}

The framework that we are interested in developing in this thesis is
characterized by a specification logic, a meta-logic, and an
integration of these logics in a way that supports the two-level logic
approach to reasoning. We shall base our specification logic on the
intuitionistic theory of higher-order hereditary Harrop formulas
\cite{miller91apal}. This theory, which supports higher-order abstract
syntax, underlies the $\lambda$Prolog programming language
\cite{nadathur88iclp} and descriptions written in it can be animated
using the Teyjus system \cite{teyjus.website,qi09phd}. Our focus in
this work is on developing the meta-logic and the two-level logic
approach to reasoning and on demonstrating their practical usefulness.

The starting point for our work will be a variant of the meta-logic
called \FOLDN described by McDowell and Miller \cite{mcdowell00tcs}
that also supports the notion of higher-order abstract syntax.
Further, our work will be inspired by the two-level logic approach to
reasoning also
described by McDowell and Miller \cite{mcdowell02tocl}; from one
perspective, we will mainly be strengthening the foundations of this
approach and demonstrating how that it can be exploited effectively in
practice. The specific realization of the two-level logic approach to
reasoning in
the work of McDowell and Miller is based on \FOLDN together with the
same specification logic that we will be using in our framework. One
of the most significant components of \FOLDN is a definition mechanism
which allows one to reason about ``closed'' descriptions of systems.
Thus, one can use the logic to perform case analysis-like reasoning
about the behavior of an encoded system.
This definition mechanism is based
on earlier work on closed-world reasoning by many others, but most
notably by Schroeder-Heister \cite{schroeder-Heister93lics}, Eriksson
\cite{eriksson91elp}, and Girard \cite{girard92mail}. The \FOLDN logic
includes within it a mechanism for induction on natural numbers. Tiu
extended this capability in the meta-logic Linc to a more general one
that allows definitions themselves to be treated inductively and
co-inductively. The co-inductive treatment was initially limited, but
Tiu and Momigliano have subsequently developed the logic Linc$^-$
which removes these limitations \cite{tiu.momigliano}.

McDowell and Miller's original meta-logic has also evolved in another
way: the idea of generic judgments has been added to it to provide a
better treatment of binding structure in higher-order abstract syntax
representations than that afforded by the universal judgments
originally used for this purpose. More specifically, Miller and Tiu
introduced a new quantifier called $\nabla$ which provides an elegant
way to decompose higher-order abstract syntax representations by
mapping term-level binding structure into a closely related
proof-level binding structure. However, the original treatment of the
$\nabla$ quantifier interacted poorly with inductive and co-inductive
reasoning. This has motivated Tiu to develop the logic \LG which
refines the treatment of this new quantifier by including certain
structural rules for it \cite{tiu06lfmtp}.

This thesis makes contributions to the setting described above by
further strengthening the meta-logic, by using it to develop an actual
computer-based system for reasoning about specifications, and by
demonstrating the benefit of the overall framework through actual
reasoning applications. We discuss these contributions in more detail
below.

\begin{enumerate}
\item We define a meta-logic called \logic which improves on previous
logics such as Linc and \LG. These other logics allow one to decompose
higher-order abstract syntax by introducing $\nabla$-quantified
variables into the structure of terms. These variables act like
proof-level binders and allowed one to reason about the binding
structure of objects without explicitly selecting variable names.
However, these logics do not have any way to analyze the structure of
terms with respect to the occurrences of such proof-level bound
variables, a task which is common to almost all reasoning about
binding structure. The meta-logic \logic rectifies this situation by
providing a generalization of the notion of equality which allows for
exactly the type of analysis described. This generalized notion of
equality behaves well with respect to definitions, induction, and
co-induction. We establish consistency and more generally the
cut-elimination property for \logic, and we find that this meta-theory
is a natural and pleasing extension of the meta-theory of previous
logics. These contributions are the contents of
Chapters~\ref{ch:meta-logic}~and~\ref{ch:meta-theory}.

\item The two-level logic approach had previously not been implemented
and, hence, tested and the Linc logic had received only a partial
implementation in a system called Bedwyr \cite{baelde06manual}. This
thesis develops, for the first time, a complete realization of the
reasoning component of the proposed framework. In particular, it
develops a system called Abella that implements the meta-logic \logic
and supports the two-level logic approach to reasoning. Abella greatly
extends the capabilities of Bedwyr by incorporating full inductive and
co-inductive reasoning capabilities. Experiments with Abella have
largely verified the effectiveness of the framework it supports, and
this aspect of our work has consequently contributed significantly to
demonstrating the practicality of the two-level logic approach to
reasoning. The discussion of Abella and its architecture is the
content of Chapter~\ref{ch:architecture}.

\item We use Abella to expose a methodology of proof construction
within the proposed framework which has a close correspondence with
traditional pencil-and-paper proofs. We formally prove part of this
correspondence through adequacy results for our two-level logic
approach, and we demonstrate how to prove the full correspondence
between the two-level logic approach to reasoning and traditional
pencil-and-paper
proofs. Finally, though concrete examples, we showcase the expressive
power of the meta-logic \logic and the practical benefits of the
two-level logic approach to reasoning. These contributions are the
contents of
Chapters~\ref{ch:two-level-reasoning}~and~\ref{ch:applications}.
\end{enumerate}

We note that the work described in this thesis has already contributed
to the tools and techniques used by other researchers. The Abella
system, that has been freely distributed, has been downloaded and
experimented with by several researchers. It has also been used in at
least one instance to verify a paper-and-pencil proof in a research
paper \cite{tiu.tocl}.

\section{Overview of the Thesis}
\label{sec:overview-thesis}

In Chapter~\ref{ch:specification-logic} we present the specification
logic used in our proposed framework. We prove properties of this
logic which make it a good basis for reasoning about object systems.
We then encode the example of the simply-typed $\lambda$-calculus
within the specification logic and prove the type preservation
property via this encoding. The reasoning techniques used in this
proof motivate some of the design of the meta-logic \logic. We pick up
on the specification logic again when we discuss the two-level
approach to reasoning in Chapter~\ref{ch:two-level-reasoning}.

Chapter~\ref{ch:meta-logic} introduces the meta-logic \logic and its
various features including an extended notion of equality, a
definition mechanism for encoding specifications, and induction and
co-induction capabilities. We show how the extended notion of equality
can be combined with the definition mechanism to produce a useful way
of describing certain objects which occur frequently when reasoning
over higher-order abstract syntax descriptions. Finally, we provide
examples which highlight the expressiveness of the new extended notion
of equality. The contents of this chapter and the next also appear
in \cite{gacek08lics,gacek.na}.

We develop the meta-theory of the meta-logic \logic in
Chapter~\ref{ch:meta-theory}. The primary result of this chapter is
the proof of cut-elimination which we use to prove other useful
properties relative to our meta-logic. We discover here that there is
a nice (meta-theoretic) modularity to our use of an extended notion of
equality as the basis for endowing \logic with richer capabilities
than the logics it builds on. In particular, we are able to reuse in
this chapter much of the meta-theory already developed for Linc$^-$
\cite{tiu.momigliano}, thereby greatly reducing the effort that is
needed for proving properties such as cut-elimination.

In Chapter~\ref{ch:architecture} we describe the Abella system and its
architecture. We describe the role of lemmas and
lemma-like hypotheses during proof construction, and we show how the
induction and co-induction rules of \logic can be presented to the
user in a very natural way.

Chapter~\ref{ch:two-level-reasoning} brings together the specification
logic and the meta-logic to develop the two-level logic approach to
reasoning. In particular, this chapter describes how the specification
logic can be embedded in the meta-logic and what benefit this has
towards formalizing the properties of the specification logic. We
reconsider the example of the simply-typed $\lambda$-calculus and
using the two-level logic approach to reasoning we provide a very
short and elegant
proof of type preservation. Finally, we show that our encoding of the
specification logic is adequate subject to some minor conditions.

Using the two-level logic approach to reasoning and its embodiment in
the Abella
theorem prover we present larger applications of our framework in
Chapter~\ref{ch:applications}. These applications are intended to
highlight the strengths and weaknesses of the two-level logic approach
to reasoning. They include examples such as the POPLmark challenge
\cite{aydemir05tphols} and Girard's proof of strong normalization for
the simply-typed $\lambda$-calculus.

In Chapter~\ref{ch:related-work} we compare our framework against
other approaches to specifying, prototyping, and reasoning about
computational systems with binding.

We conclude this thesis in Chapter~\ref{ch:future-work} and discuss
various avenues of future work. These range from foundational
extensions which would increase the expressive power of the meta-logic
to more implementation oriented extensions which would better
facilitate the reasoning process.



\chapter{A Logic for Specifying Computational Systems}
\label{ch:specification-logic}

The primary requirement of a specification logic within the framework
that we want to develop is that it allow for a transparent encoding of
the kinds of formal systems that are of interest to us. In particular,
such an encoding should cover both the objects manipulated within the
formal system and the rules by which they are manipulated. In the
context of our work, we are particularly concerned with the
representation of objects that incorporate a variable binding
structure. A logically precise encoding of such structure plays an
important role in the overall treatment of the relevant computational
systems. An encoding that has this character usually requires the
treatment of concepts related to binding, such as equality under bound
variable renaming and capture-avoiding substitution. If these aspects
are not dealt with in a systematic way within the specification logic,
they can overwhelm the process of constructing encodings and can make
the subsequent process of reasoning about specifications unnecessarily
complex. We therefore seek a specification logic which incorporates a
flexible and sophisticated treatment of variable binding structure and
which also builds in the related binding notions.

In this chapter we introduce the specification logic of second-order
hereditary Harrop formulas, abbreviated \hh. This logic is essentially
a restriction of the logic of higher-order hereditary Harrop formulas
\cite{miller91apal} that underlies the language $\lambda$Prolog
\cite{nadathur88iclp}. The \hh logic can be seen as an extension of
the Horn clause logic, the logic that underlies Prolog, with devices
for representing, examining, and manipulating objects with binding
structure. In particular, \hh allows for a higher-order abstract
syntax representation of objects with binding structure
\cite{miller87slp, pfenning88pldi}. Thus issues
of variable renaming and capture-avoiding substitution are taken care
of once and for all in the specification logic, leaving particular
specifications free to focus on the more essential aspects of the
system they encode. Furthermore, like the logic of higher-order
hereditary Harrop formulas that it derives from, \hh admits an
operational semantics which allows specifications to be animated
automatically thus yielding quick prototypes of the computational
systems they encode.

In this chapter we formally define the \hh logic, describe its
operational semantics, state and prove properties of the logic, and
demonstrate its use through a concrete example.

\section{The Syntax and Semantics of the Logic}

Following Church \cite{church40}, terms in \hh are constructed using
abstraction and application from constants and bound variables. All
terms are typed using a monomorphic typing system. The provability
relation concerns well-formed terms of the the distinguished type $o$
that are also called formulas. Logic is introduced by including
special constants representing the propositional connectives $\top$,
$\land$, $\lor$, $\supset$ and, for every type $\tau$ that does not
contain $o$, the constants $\forall_\tau$ and $\exists_\tau$ of type
$(\tau \rightarrow o) \rightarrow o$.
We do not allow any other constants or variables to have a type
containing the type $o$.
The binary propositional
connectives are written as usual in infix form and the expressions
$\forall_\tau x. B$ and $\exists_\tau x. B$ abbreviate the formulas
$\forall_\tau \lambda x.B$ and $\exists_\tau \lambda x.B$,
respectively. Type subscripts will be omitted from quantified formulas
when they can be inferred from the context or are not important to the
discussion. We also use a shorthand for iterated quantification: if
${\cal Q}$ is a quantifier, we will often abbreviate ${\cal
  Q}x_1\ldots{\cal Q}x_n.P$ to ${\cal Q}x_1,\ldots,x_n.P$ or simply
${\cal Q}\vec{x}.P$. We consider the scope of $\lambda$-binders (and
therefore quantifiers) as extending as far right as possible. We
further assume that $\supset$ is right associative and has lower
precedence than $\land$ and $\lor$. For example, $\forall x. t_1
\supset t_2 \supset t_3 \land t_4$ should be read as $\forall x.(t_1
\supset (t_2 \supset (t_3\land t_4)))$.

We restrict our attention to two classes of formulas in \hh described
by the following grammar.
\begin{align*}
G &::= \top \mid A \mid A \supset G \mid \forall_\tau x. G \mid
\exists_\tau x . G \mid G \land G \mid G \lor G \\
D &::= A \mid G \supset D \mid \forall_\tau x. D
\end{align*}
Here $A$ denotes an atomic formula. The formulas denoted by $G$ are
called {\em goals} and represent the conclusions we can infer in the
logic. A notable restriction on implication in goal formulas is that
the left hand side must be an atomic formula.
Formulas denoted by $D$ are
called {\em definite clauses} and represent the hypotheses we can
assume in the logic. Notice that disjunctions and existentials are not
allowed in definite formulas because they represent indefinite
knowledge. For simplicity, we also disallow conjunction, but the
effect of conjunctions can be recovered by using a set of clauses in
place of a single clause.
The order of a formula is the depth of implications which are nested
to the left of other implications. Our restriction on implication
means goal formulas are at most first-order and definite clauses
are at most second-order. It is precisely this restriction which
carves out the logic of second-order hereditary Harrop formulas from
the larger logic of higher-order hereditary Harrop formulas.
Finally, by using logical equivalences we
can percolate universal quantifiers to the top, to rewrite all
definite clauses to be of the form $\forall x_1 \ldots \forall x_n .
(G_1 \supset \cdots \supset G_m \supset A)$ where $n$ and $m$ may both
be zero. In the future we will assume all definite clauses are in this
form.

The semantics of \hh are formalized by means of a proof-theoretic
presentation of what it means for a goal to follow from a set of
definite clauses. Specifically, we will be concerned with the
derivation of {\em sequents} of the form $\Sigma : \Delta \vdash G$
where $\Delta$ is a list of $D$-formulas, $G$ is a $G$-formula, and
$\Sigma$ is a set of variables called eigenvariables. For such a
sequent to be well-formed, we require that the formulas in
$\Delta\cup\{G\}$ must be constructed using using only the logical and
non-logical constants of the language and the eigenvariables in
$\Sigma$. This well-formedness condition is guaranteed for every
sequent considered in a derivation by ensuring that we try to
construct derivations only for well-formed ones at the top-level and
by the use of typing judgments of the form $\Sigma \vdash t : \tau$ in
rules that introduce new terms when these rules are interpreted in a
proof search direction. The meaning of this typing judgment, that we
do not explicitly formalize here, is the following: for it to hold,
the term $t$ must have the type $\tau$ and it must also be constructed
using only the non-logical constants and the eigenvariables in
$\Sigma$.

The rules for constructing proofs for such sequents are presented in
Figure~\ref{fig:hh-rules}. The GENERIC rule introduces an
eigenvariable when read in a proof search direction. There is a
freshness side-condition associated with this eigenvariable: $c$ must
not already be in $\Sigma$. Note that for this to be possible, we must
assume that there is an unlimited supply of eigenvariables of each
type.
In the INSTANCE rule $t$ is required to be a term such that $\Sigma
\vdash t : \tau$ holds. Similarly, in the BACKCHAIN rule for each term
$t_i \in \vec{t}$ we must have $\Sigma \vdash t_i : \tau_i$ where
$\tau_i$ is the type of the quantified variable $x_i$. An important
property to note about these rules is that if we use them to search
for a proof of the sequent $\Delta \vdash G$, then all the
intermediate sequents that we will encounter will have the form
$\Delta, \mathcal{L} \vdash G'$ for some $G$-formula $G'$ and some
list of atomic formulas $\mathcal{L}$. Thus the initial context
$\Delta$ is global, and only atomic formulas are added to the context
during proof construction.

In presenting sequents in later parts of this thesis, we shall
occasionally omit writing the signature. We will do this only when
either the identity of the signature is irrelevant to the discussion
of when it can be inferred from the context.

\begin{figure}[t]
\centering
\begin{equation*}
\infer[\TRUE]
      {\Sigma : \Delta \vdash \top}
      {}
\end{equation*}
\medskip
\begin{equation*}
\infer[\OR_1]
      {\Sigma : \Delta \vdash G_1 \lor G_2}
      {\Sigma : \Delta \vdash G_1}
\hspace{1cm}
\infer[\OR_2]
      {\Sigma : \Delta \vdash G_1 \lor G_2}
      {\Sigma : \Delta \vdash G_2}
\end{equation*}
\medskip
\begin{equation*}
\infer[\AND]
      {\Sigma : \Delta \vdash G_1 \land G_2}
      {\Sigma : \Delta \vdash G_1 &
       \Sigma : \Delta \vdash G_2}
\hspace{1cm}
\infer[\INSTANCE]
      {\Sigma : \Delta \vdash \exists x.G}
      {\Sigma : \Delta \vdash G[t/x]}
\end{equation*}
\medskip
\begin{equation*}
\infer[\AUGMENT]
      {\Sigma : \Delta \vdash A \supset G}
      {\Sigma : \Delta, A \vdash G}
\hspace{1cm}
\infer[\GENERIC]
      {\Sigma : \Delta \vdash \forall_\tau x.G}
      {\Sigma \cup \{c\!:\!\tau\} : \Delta \vdash G[c/x]}
\end{equation*}
\medskip
\begin{equation*}
\infer[\BACKCHAIN]
      {\Sigma : \Delta \vdash A}
      {\Sigma : \Delta \vdash G_1[\vec{t}/\vec{x}] &
       \cdots &
       \Sigma : \Delta
       \vdash G_m[\vec{t}/\vec{x}]}
\end{equation*}
where $\forall \vec{x} . (G_1 \supset \cdots \supset
G_m \supset A') \in \Delta$ and $A'[\vec{t}/\vec{x}] = A$
\caption{Derivation rules for the \hh logic}
\label{fig:hh-rules}
\end{figure}

The rules of \hh admit a simple proof search procedure: given a
sequent $\Delta \vdash G$ we decompose the goal $G$ until we reach an
atomic formula at which point we backchain and attempt to prove the
resulting goals. This is, in fact, a manifestation of the {\em uniform
  proofs} property that \hh inherits from the parent logic of
higher-order hereditary Harrop formulas \cite{miller91apal}. The
resulting procedure is non-deterministic since we have a choice 
when the goal is a disjunction, an existential, or an atomic formula
(we can choose which clause to backchain on). The non-determinism
induced by existentials can be handled using the standard notion of
instantiatable variables and unification while the non-determinism of the OR and
BACKCHAIN rules can be handled using depth-first search complemented
with backtracking. Computations 
described by \hh are included within those corresponding to
$\lambda$Prolog and can therefore be compiled and executed efficiently,
\eg, by the Teyjus system \cite{teyjus.website,qi09phd}.

\section{Properties of the Specification Logic}
\label{sec:prop-spec-logic}

We will eventually encode object logic judgments into specification
logic judgments.  By doing this, we enable ourselves to use properties
of the specification logic in proving properties of the object logic.
Therefore in this section we enumerate the various properties of the
\hh logic which may be useful in such reasoning. The proofs of these
properties will be based on induction over the {\em height} of a
derivation, a notion we define now.

\begin{definition}
\label{def:hh2-ht}
The {\em height} of a derivation $\Pi$, denoted by $\htf(\Pi)$, is $1$
if $\Pi$ has no premise derivations and is $\max
\{\htf(\Pi_i)+1\}_{i\in 1..n}$ if $\Pi$ has the premise derivations
$\{\Pi_i\}_{i\in 1..n}$.
\end{definition}

The {\em monotonicity property} of \hh states that the eigenvariables
and the context of a sequent can always be expanded while preserving
provability.
\begin{lemma}
Let $\Sigma : \Delta \vdash G$ be a well-formed sequent, let $\Delta'$
be a list of definite clauses such that $\Delta \subseteq \Delta'$,
and let $\Sigma'$ be a set of eigenvariables such that $\Sigma
\subseteq \Sigma'$ and $\Sigma'$ contains all the eigenvariables of
$\Delta'$. If $\Sigma : \Delta \vdash G$ has a derivation then
$\Sigma' : \Delta' \vdash G$ has a derivation. Moreover, the height of
the derivation does not increase.
\end{lemma}
\begin{proof}
Induction on the height of the derivation of $\Sigma : \Delta \vdash
G$.
\end{proof}
The {\em instantiation property} states that a eigenvariable $c$
which arises from a use of the GENERIC rule can always be instantiated
with a particular value while preserving provability. As a result, our
use of eigenvariables to denote universal quantification in \hh is
well justified.
\begin{lemma}
Let $c$ be a variable not in $\Sigma$. If $\Sigma \cup \{c:\tau\} :
\Delta \vdash G$ has a derivation then for all terms $t$ such that
$\Sigma \vdash t : \tau$ there is a derivation of $\Sigma :
\Delta[t/c] \vdash G[t/c]$. Moreover, the height of the derivation
does not increase.
\end{lemma}
\begin{proof}
Induction on the height of the derivation of $\Sigma \cup \{c:\tau\} :
\Delta \vdash G$.
\end{proof}
Finally, the {\em cut admissibility property} says that the assumption
of an atomic formula can be discharged if the atomic formula is itself
provable.
\begin{lemma}
If $\Sigma : \Delta, A \vdash G$ and $\Sigma : \Delta \vdash A$ then
$\Sigma : \Delta \vdash G$.
\end{lemma}
\begin{proof}
Induction on the height of the derivation of $\Sigma : \Delta, A
\vdash G$. There are two interesting cases. The first case is when $G$
is $A' \supset G'$ in which case we must apply the monotonicity
property to move from $\Sigma : \Delta, A, A' \vdash G$ to $\Sigma :
\Delta, A', A \vdash G$. The other case is when the BACKCHAIN rule
selects $A$, in which case the derivation of $\Sigma : \Delta \vdash
A$ can be substituted.
\end{proof}

\section{Example Encoding in the Specification Logic}
\label{sec:spec-example}

We now take the example of evaluation and typing for the simply-typed
$\lambda$-calculus from Section~\ref{sec:example}, and we encode it
into the specification logic. We introduce the specification logic
types $tp$ and $tm$ for representing types and pre-terms respectively
in the simply-typed $\lambda$-calculus. Types in the simply-typed
$\lambda$-calculus will be mapped to specification logic terms
constructed from the constants {\sl i} and {\sl arr} of types $tp$ and
$tp \to tp \to tp$, respectively. Pre-terms in the simply-typed
$\lambda$-calculus will be mapped to specification logic terms
constructed from the constants {\sl app} and {\sl abs} of types $tm
\to tm \to tm$ and $tp \to (tm \to tm) \to tm$, respectively. Notice
that the second argument of {\sl abs} is expected to be an abstraction
over $tm$ in the specification logic. Finally, we will have two
constants {\sl of} and {\sl eval} of types $tm \to tp \to o$ and $tm
\to tm \to o$, respectively, which denote typing and evaluation,
respectively. The clauses for these predicates are presented in
Figure~\ref{fig:stlc-hh}. Here and in the future we use the convention
that tokens given by capital letters denote variables that are
implicitly universally quantified over the entire formula. In the
second clause for evaluation, $R$ is an abstraction in the
specification logic and thus the built-in notion of $\beta$-reduction
means that $(R\ N)$ realizes capture-avoiding substitution of $N$ in
for the bound variable in $R$. For the typing judgment, we do not keep
an explicit context of typing assumptions, instead relying on the
specification logic context. This is reflected in the rule for typing
abstractions where we use the $\forall$ quantifier to create a fresh
eigenvariable and we assume that this eigenvariable has the proper
type while we derive a typing assignment for the body of the
abstraction. In this way, we avoid having an explicit base case for
typing. Next, when we reason about this specification we will be able
to exploit this encoding of the typing context.

\begin{figure}[t]
\centering
\begin{align*}
&\eval {(\abs A M)} {(\abs A M)} \\[3pt]
&\eval M (\abs A R) \supset \eval {(R\ N)} V \supset \eval {(\app M
  N)} {V}\\[10pt]
&\of M (\arr A B) \supset \of N A
\supset \of{(\app M N)} B \\[3pt]
&(\forall x. \of x A \supset \of{(R\
  x)}{B}) \supset \of{(\abs A R)}{(\arr A B)}
\end{align*}
\caption{\hh specification of evaluation and typing}
\label{fig:stlc-hh}
\end{figure}

Using this encoding, we can now repeat the proof of type preservation
and leverage on the properties we have shown of the \hh logic. Let
$\Delta$ be the clauses from Figure~\ref{fig:stlc-hh}.
\begin{theorem}
If $\Delta \vdash \eval e v$ holds and $\Delta \vdash \of e t$ holds
then $\Delta \vdash \of v t$ holds.
\end{theorem}
\begin{proof}
By induction on the height of the derivation of $\Delta \vdash \eval e
v$. We proceed by cases on the derivation of $\Delta \vdash \eval e
v$. This judgment must have been derived by backchaining on one of the
clauses for {\sl eval}. If it was by the first clause, then $e = v$
and the case is complete. Otherwise it was by the second clause so $e$
must be $(\app m n)$ for some $m$ and $n$ and we have shorter
derivations of $\Delta \vdash \eval m (\abs a r)$ and and $\Delta
\vdash \eval {(r\ n)} v$ for some $a$ and $r$. By similarly examining
the derivation $\Delta \vdash \of {(\app m n)} t$ we must have
derivations of $\Delta \vdash \of m (\arr b t)$ and $\Delta \vdash \of
n b$ for some $b$. Applying the inductive hypothesis to $\Delta \vdash
\eval m (\abs a r)$ and $\Delta \vdash \of m (\arr b t)$ we have
$\Delta \vdash \of {(\abs a r)} (\arr b t)$. This derivation could
only result if $a = b$ and we have a derivation of $\Delta \vdash
\forall x[\of x a \supset \of {(r\ x)} t]$ and thus a derivation of
$\Delta, \of c a \vdash \of {(r\ c)} t$ for some eigenvariable $c$.
Now we can apply the instantiation property of our specification logic
to get a derivation of $\Delta, \of n a \vdash \of {(r\ n)} t$. Next
we apply the cut property with our derivation of $\Delta \vdash \of n
a$ to get $\Delta \vdash \of {(r\ n)} t$. Finally, we apply the
inductive hypothesis again to $\Delta \vdash \of {(r\ n)} t$ and
$\Delta \vdash \eval {(r\ n)} v$ to get $\Delta \vdash \of v t$ which
completes the proof.
\end{proof}

It is important to note in this proof that we did not have to prove a
type substitution property for the object logic. Instead, the object
logic inherited this property from the more general instantiation and
cut properties of the specification logic. Also, induction over the
height of specification logic derivations corresponded with induction
over the height of object logic derivations. Thus we can reason about
computational systems through their encoding in the specification
logic with little overhead cost.

\section{Adequacy of Encodings in the Specification Logic}
\label{sec:spec-adequacy}

A tacit assumption in the example we considered in the previous
section is that the specification of pre-terms, types, typing, and
evaluation are all faithful representations of the corresponding
concepts in the object logic. This kind of property of encodings is
referred to as the {\em adequacy} property. We must, of course, prove
such a property before we can derive benefit from it. To do this, we
need to prove that there is a bijection between components of the
object logic and their specification logic representations and that
this bijection preserves properties of relevance in the two systems. With
specific reference to the example encoding we have considered, we have
to show that each object in the simply-typed $\lambda$-calculus has a
unique representation in the specification logic, and each
representation in the specification logic corresponds to a unique
object in the simply-typed $\lambda$-calculus. We show below how such
arguments are typically carried out.

To simplify the argument we will assume an implicit mapping between
bound variables in the object language and bound variables in the
specification language, and between free variables in the object
language and eigenvariables in the specification language. A more
rigorous treatment of adequacy would make this mapping explicit
\cite{felty89phd}.

We define the bijections $\phi_{tp}$, $\phi_{tm}$, $\phi_{eval}$,
$\phi_{ctx}$, and $\phi_{of}$ which are used to map types,
pre-terms, evaluation judgments, typing contexts, and typing judgments
in the simply-typed $\lambda$-calculus to their corresponding
representations in the specification logic. We will omit the
subscripts on $\phi$ when they can be inferred from context. The proofs
that these mappings are bijective are always by straightforward
induction on the size of terms or strong induction on the height of
derivations.

Types in the simply-typed $\lambda$-calculus map to terms of type $tp$
in the specification logic. We formalize this mapping as follows.
\begin{align*}
\phi(i) = i &&
\phi(a \to b) = \arr {\phi(a)} \phi(b)
\end{align*}
This function is clearly a bijection.

Next we define the mapping between $\alpha$-equivalence classes of
pre-terms in the object logic and terms of type $tm$ in the
specification logic.
\begin{align*}
\phi(x) = x &&
\phi(m\ n) = \app {\phi(m)} \phi(n) &&
\phi(\lambdat x a r) = \abs {\phi(a)} (\lambda x.\phi(r))
\end{align*}
In the last rule for this mapping note that the $\lambda$ within the
$\phi$ is that of the simply-typed $\lambda$-calculus while the one
outside of $\phi$ is from the specification logic. This mapping
is clearly bijective under the assumption that $\alpha$-convertible
terms in the specification logic are considered to be identical.

Let $\Delta$ be the clauses from Figure~\ref{fig:stlc-hh}. Then
derivations of evaluation judgments in the simply typed
$\lambda$-calculus correspond to derivations of the sequent $\Delta
\vdash \eval e v$ in the specification logic as follows. First
consider the translation of evaluation for abstractions:
\begin{align*}
\phi\left(\raisebox{-1.5ex}{
  \infer[]
        {\lambdat x a t \Downarrow \lambdat x a t}
        {}
  }\right)
&=
\raisebox{-1.5ex}{
   \infer[]
         {\Delta \vdash \eval {\phi(\lambdat x a t)}
                               \phi(\lambdat x a t)
         }
         {}
   } \\
&=
\raisebox{-1.5ex}{
   \infer[]
         {\Delta \vdash \eval
            {(\abs {\phi(a)} (\lambda x.\phi(t)))}
            {(\abs {\phi(a)} (\lambda x.\phi(t)))}}
         {}
   }
\end{align*}
Here and in the future we propagate the mapping $\phi$ to make it
clear that the specification logic inference rules are well-formed. In
this case, the right-hand inference rule an instance of the BACKCHAIN
rule over the clause for evaluating abstractions.

The translation for evaluations of applications is the following.
\begin{tabbing}
\hspace{1cm}\= $\phi\left(\raisebox{-3.5ex}{
  \infer[]
        {m\ n \Downarrow v}
        {\deduce{m \Downarrow \lambdat x a r}{\vdots} &
         \deduce{r[x := n] \Downarrow v}{\vdots}}
     }\right)$ \\[20pt]
\>\hspace{1cm}\= $=
\raisebox{-1.5ex}{
  \infer[]
        {\Delta \vdash \eval {\phi(m\ n)} \phi(v)}
        {\phi\left(\raisebox{-1.5ex}{
          \deduce{m \Downarrow \lambdat x a r}{\vdots}
         }\right) &
         \phi\left(\raisebox{-1.5ex}{
          \deduce{r[x := n] \Downarrow v}{\vdots}
         }\right)
        }
  }$ \\[10pt]
\>\> $=
\raisebox{-1.5ex}{
  \infer[]
        {\Delta \vdash \eval {(\app {\phi(m)} \phi(n))} \phi(v)}
        {\raisebox{-1.5ex}{
          \deduce{\Delta \vdash \eval
                    {\phi(m)}
                    {(\abs {\phi(a)} (\lambda x. \phi(r)))}}
                 {\phi(\vdots)}
         } &
         \raisebox{-1.5ex}{
          \deduce{\Delta \vdash \eval
                    {(\phi(r)[\phi(n)/x])}
                    \phi(v)}
                 {\phi(\vdots)}}
        }
  }$
\end{tabbing}
In the final formula, we make use of the automatic $\beta$-conversion
in the specification logic where $(\lambda x. \phi(r))\ \phi(n) =
\phi(r)[\phi(n)/x]$, and we use the following compositional
property of the bijection for terms.
\begin{equation*}
\phi(r[x := n]) = \phi(r)[\phi(n)/x]
\end{equation*}
This equation relates the substitution of the simply-typed
$\lambda$-calculus on the left with the substitution in the
specification logic on the right. The proof of this equality is by
induction on the structure of $r$. Thus the inference rule on the
right-hand side above is a proper instance of the BACKCHAIN rule over
the clause for evaluating applications. The inverse of the $\phi$
mapping is defined in the natural way and thus $\phi$ is a bijection.

Finally, we look at derivations of typing judgments in the
simply-typed $\lambda$-calculus and we map these to derivations of
sequents of the form $\Delta, \mathcal{L} \vdash \of e t$ where
$\mathcal{L}$ is a list of atomic formulas of the form $\of {x_1}
{a_1}, \ldots, \of {x_k} {a_k}$ where each $x_i$ is a unique
eigenvariable. We first define the following bijection between a list
of typing assumptions $\Gamma$ from the simply-typed
$\lambda$-calculus and a list of atomic formulas of the form described
for $\mathcal{L}$.
\begin{equation*}
\phi(x_1 : a_1, \ldots, x_k : a_k) =
\of{x_1}{\phi(a_1)}, \ldots, \of{x_k}{\phi(a_k)}
\end{equation*}
Given this, we can define the mapping for typing variables as follows.
\begin{equation*}
\phi\left(\raisebox{-1.5ex}{
  \infer[]
        {\Gamma \vdash x_i : a_i}
        {}
  }\right)
=
\raisebox{-1.5ex}{
   \infer[]
         {\Delta, \phi(\Gamma) \vdash \of{x_i}{\phi(a_i)}}
         {}
   }
\end{equation*}
If the typing derivation within the $\phi$ is correct then it must be
that $x_i : a_i \in \Gamma$. Thus the right-hand side is an
instance of the BACKCHAIN rule on the clause $\of {x_i} \phi(a_i)$
which is in $\phi(\Gamma)$.

The typing rule for applications is mapped in the expected way:
\begin{tabbing}
\hspace{1cm}\= $\phi\left(\raisebox{-3.5ex}{
  \infer[]
        {\Gamma \vdash m\ n : b}
        {\deduce{\Gamma \vdash m : a \to b}{\vdots\ \ \ } &
         \deduce{\Gamma \vdash n : a}{\vdots}}
  }\right)$ \\[20pt]
\>\hspace{1cm}\=
$= \raisebox{-3.5ex}{
   \infer[]
         {\Delta, \phi(\Gamma) \vdash \of {\phi(m\ n)} \phi(b)}
         {\phi\left(\raisebox{-1.5ex}{
           \deduce{\Gamma \vdash m : a \to b}{\vdots\ \ \ }
          }\right)
          &
          \phi\left(\raisebox{-1.5ex}{
            \deduce{\Gamma \vdash n : a}{\vdots}
          }\right)
         }
  }$ \\[10pt]
\>\>
$= \raisebox{-3.5ex}{
   \infer[]
         {\Delta, \phi(\Gamma) \vdash
            \of {(\app {\phi(m)} \phi(n)))} \phi(b)}
          {\raisebox{-1.5ex}{
             \deduce{\Delta, \phi(\Gamma) \vdash
               \of {\phi(m)} {(\arr {\phi(a)} \phi(b))}}
               {\phi(\vdots)}
           }
          &
          \raisebox{-1.5ex}{
            \deduce{\Delta, \phi(\Gamma) \vdash
              \of {\phi(n)} {\phi(a)}}
             {\phi(\vdots)}
          }
         }
  }$
\end{tabbing}

For mapping the abstraction typing rule, we need to be mindful of the
variable naming restriction and how this is realized in the
specification logic. Suppose we want to define the following mapping.
\begin{equation*}
\phi\left(\raisebox{-3.5ex}{
  \infer[]
        {\Gamma \vdash (\lambdat x a r) : a \to b}
        {\deduce{\Gamma, x : a \vdash r : b}{\vdots}}
  }\right)
\end{equation*}
Here we assume that $x$ does not appear in $\Gamma$ so that the naming
restriction is satisfied. We map this to the following specification
logic derivation.
\begin{equation*}
\raisebox{-6ex}{
   \infer[\BACKCHAIN]
     {\Delta, \phi(\Gamma) \vdash
        \of {(\abs {\phi(a)} (\lambda x. \phi(r)))}
            (\arr {\phi(a)} \phi(b))}
     {\infer[GENERIC]
        {\Delta, \phi(\Gamma) \vdash \forall x . [\of x \phi(a)
          \supset \of {((\lambda x. \phi(r))\ x)} \phi(b)]}
        {\infer[AUGMENT]
          {\Delta, \phi(\Gamma) \vdash \of x \phi(a)
           \supset \of {\phi(r)} \phi(b)}
          {\infer[]
            {\Delta, \phi(\Gamma), \of x \phi(a) \vdash
               \of {\phi(r)} \phi(b)}
            {\phi(\vdots)}}}}
  }
\end{equation*}
In the GENERIC rule we overload notation to let $x$ be the eigenvariable
we select. Since it does not appear in $\Gamma$ it will not appear in
$\phi(\Gamma)$, and thus the freshness side-condition on the GENERIC
rule is satisfied. In fact, the naming restriction in the object logic
matches up with the freshness side-condition in the specification
logic exactly as needed.

The inverse of the $\phi$ mapping for typing judgments can be defined
in the expected way, and thus $\phi$ is a bijection. This concludes
the proof of adequacy for our specification. In the future we will
omit such arguments since our specifications are often transparent
encodings of the systems they represent.



\chapter{A Logic for Reasoning About Specifications}
\label{ch:meta-logic}

In this chapter we present the meta-logic \logic. This logic allows
for encoding descriptions of computational systems and for reasoning
over those descriptions. The logic includes traditional reasoning
devices such as case analysis, induction, and co-induction as well as
new devices specifically designed for working with higher-order abstract
syntax.

The relevant history of \logic begins with the meta-logic \FOLDN
developed by McDowell and Miller for the purposes of inductive
reasoning over higher-order abstract syntax descriptions
\cite{mcdowell02tocl, mcdowell00tcs}. This logic contains a definition
mechanism which allows one to specify and reason about closed-world
descriptions, \ie, allows one to form judgments and to perform case
analysis on them. This definition mechanism is based on earlier work
on closed-world reasoning by many others, but most notably by
Schroeder-Heister \cite{schroeder-Heister93lics}, Eriksson
\cite{eriksson91elp}, and Girard \cite{girard92mail}. The primary
contribution of \FOLDN was the recognition that definitions provided a
way of encoding higher-order abstract syntax descriptions in such a
way that does not conflict with inductive reasoning. In particular,
\FOLDN allowed for natural number induction, and so many reasoning
tasks could be naturally encoded. More recently, Tiu \cite{tiu04phd}
developed the meta-logic Linc which extends the mechanism of
definitions to integrate notions of generalized induction and
co-induction over the structure of definitions. These more general
notions are present in \logic as well.

Another central advancement in the development of logics for reasoning
over higher-order abstract syntax descriptions was the recognition
that one needed a way to reflect the binding structure of terms into
the structure of proofs. This was realized in earlier logics by using
universal judgments. However, this kind of correspondence was always
an uneasy one and the mismatch became explicit when it was necessary
to use case analysis arguments over binding structure as must be done,
for example, in bisimilarity proofs associated with $\pi$-calculus
models of concurrent systems. The desire to provide a logically
precise and cleaner treatment led to the development of the
$\nabla$-quantifier and the associated generic judgment by Miller and
Tiu in the meta-logic \foldnb \cite{miller05tocl}. Tiu later refined
this notion in the meta-logic \LG so that $\nabla$-quantifier behaved
well with respect to inductive reasoning \cite{tiu06lfmtp}. This
interpretation of the $\nabla$-quantifier is present in \logic, and in
this context it can be understood as quantifying over fresh names.

The meta-logic \logic is a continuation of the research surrounding
inductive reasoning and higher-order abstract syntax descriptions. In
particular, it extends the notion of equality in the logic to one
which can describe the binding structure of terms relative to the
proof context in which they occur. This turns out to be essential to
describing the structure of terms which are generated during inductive
reasoning over higher-order abstract syntax descriptions. Moreover,
\logic identifies how this extended notion of equality can be
integrated with the definition mechanism to allow a succinct
description of such objects.

The presentation of \logic is divided into three parts. First,
Section~\ref{sec:logic} contains the core of the logic including
generic quantification. Then Section~\ref{sec:nominal-abstraction}
introduces the extended notion of equality known as {\em nominal
  abstraction} and rules for treating this notion within the logic.
Finally, Section~\ref{sec:definitions} presents rules for treating
fixed-points in the logic including mechanisms for induction and
co-induction. Although the logical features of \logic are described in
their entirety in the first three sections, it is sometimes convenient
to use an alternative presentation for fixed-point definitions. This
form, which uses patterns to distinguish different cases in the
structure of the atom being defined, is introduced in
Section~\ref{sec:pattern-form} and is elaborated as an interpretation
of the basic form of definitions that uses nominal abstractions
explicitly. Rules for treating this alternative form of fixed-points
are presented and proven to be admissible. Finally,
Section~\ref{sec:examples} provides some small examples to illustrate
the expressive power of the logic.

\section{A Logic with Generic Quantification}
\label{sec:logic}

In this section we present the core logic underlying \logic. This
logic is obtained by extending an intuitionistic and predicative
subset of Church's Simple Theory of Types with a treatment of generic
judgments. The encoding of generic judgments is based on the
quantifier called $\nabla$ (pronounced nabla) introduced by Miller and
Tiu \cite{miller05tocl} and further includes the structural rules
associated with this quantifier in the logic \LG described by Tiu
\cite{tiu06lfmtp}.

\subsection{The Basic Syntax}

Following Church \cite{church40}, terms are constructed from constants
and variables using abstraction and application. All terms are
assigned types using a monomorphic typing system; these types also
constrain the set of well-formed expressions in the expected way. The
collection of types includes $o$, a type that corresponds to
propositions. Well-formed terms of this type are also called formulas.
Two terms are considered to be equal if one can be obtained from the
other by a sequence of applications of the $\alpha$-, $\beta$- and
$\eta$-conversion rules, \ie, the $\lambda$-conversion rules. This
notion of equality is henceforth assumed implicitly wherever there is
a need to compare terms. Logic is introduced by including special
constants representing the propositional connectives $\top$, $\bot$,
$\land$, $\lor$, $\supset$ and, for every type $\tau$ that does not
contain $o$, the constants $\forall_\tau$ and $\exists_\tau$ of type
$(\tau \rightarrow o) \rightarrow o$. The binary propositional
connectives are written as usual in infix form and the expressions
$\forall_\tau x. B$ and $\exists_\tau x. B$ abbreviate the formulas
$\forall_\tau \lambda x.B$ and $\exists_\tau \lambda x.B$,
respectively. Type subscripts will be omitted from quantified formulas
when they can be inferred from the context or are not important to the
discussion. We also use a shorthand for iterated quantification: if
${\cal Q}$ is a quantifier, we will often abbreviate ${\cal
  Q}x_1\ldots{\cal Q}x_n.P$ to ${\cal Q}x_1,\ldots,x_n.P$ or simply
${\cal Q}\vec{x}.P$. We consider the scope of $\lambda$-binders (and
therefore quantifiers) as extending as far right as possible. We
further assume that $\supset$ is right associative and has lower
precedence than $\land$ and $\lor$. For example, $\forall x. t_1
\supset t_2 \supset t_3 \land t_4$ should be read as $\forall x.(t_1
\supset (t_2 \supset (t_3\land t_4)))$.

The usual inference rules for the universal quantifier can be seen as
equating it to the conjunction of all of its instances: that is, this
quantifier is treated extensionally.
There are several situations  where one
wishes to treat an expression such as ``$B(x)$ holds for all
$x$'' as a statement about the existence of a uniform argument for
every instance rather than the truth of a particular property for each
instance \cite{miller05tocl};
such situations typically arise when one is reasoning about the
binding structure of formal objects represented using the
$\lambda${\em -tree syntax} \cite{miller00cl} version of {\em
higher-order abstract syntax} \cite{pfenning88pldi}.
The $\nabla$-quantifier serves to encode judgments that have this kind
of a ``generic'' property associated with them. Syntactically, this
quantifier corresponds to
including a constant $\nabla_\tau$ of type $(\tau \rightarrow o)
\rightarrow o$ for each type $\tau$ not containing $o$.\footnote{
We may choose to allow $\nabla$-quantification at fewer types in
particular applications; such a restriction may be
  useful in adequacy arguments for reasons we discuss later.}
As with the
other quantifiers, $\nabla_\tau x. B$ abbreviates $\nabla_\tau \lambda
x. B$ and the type subscripts are often suppressed for readability.

\subsection{Generic Judgments and $\nabla$-quantification}

Sequents in intuitionistic logic can be written as
\[\Sigma : B_1, \ldots, B_n \lra B_0 \qquad (n\ge0)
\]
where $\Sigma$ is the ``global signature'' for the sequent that
contains the {\em eigenvariables} ({\em i.e.}, variables associated to
the $\existsL$ and $\forallR$ inference rules) relevant to the sequent
proof. We shall think of $\Sigma$ in this prefix position as an
operator that binds each of the variables it contains and that has the
rest of the sequent as its scope. To treat the $\nabla$-quantifier,
the \foldnb logic \cite{miller05tocl} extends the notion of a judgment
from just a formula to a formula paired with a ``local signature.''
Thus, sequents within this logic are written more elaborately as
\[\Sigma : \sigma_1\triangleright B_1, \ldots,
           \sigma_n\triangleright B_n \lra
           \sigma_0\triangleright B_0,
\]
where each $\sigma_0, \ldots, \sigma_n$ is a list of variables that are
bound locally in the formula adjacent to it.  Such local signatures
correspond to a proof-level encoding of binding that is expressed
within formulas through the $\nabla$-quantifier. In particular, the
judgment $x_1,\ldots,x_n\triangleright B$ and the formula $\nabla
x_1\cdots\nabla x_n.B$ for $n \ge 0$ have the same proof-theoretic
force. In keeping with this observation, we shall refer to a judgment
of the form $\sigma \triangleright B$ as a {\it generic judgment}.

As part of a generalization of sequents that bases them on generic
judgments rather than on formulas, we need to
define when two such judgments are equal: this is necessary
for describing at least the initial and cut inference rules.  The
\foldnb logic \cite{miller05tocl} uses a simple form of equality for
this purpose. It deems two generic judgments of the form
$x_1,\ldots,x_n\triangleright B$ and $y_1,\ldots,y_m\triangleright C$
to be equal exactly when the $\lambda$-terms $\lambda
x_1\ldots\lambda x_n. B$ and $\lambda y_1\ldots\lambda y_m. C$ are
$\lambda$-convertible; notice that this necessarily implies that
$n=m$. An equality notion is also needed in formulating an induction
rule. Unfortunately, the simple form of equality present in \foldnb
leads to a rather weak version of such a rule. To overcome this difficulty,
Tiu proposed the addition to the logic of two natural ``structural''
identities between generic judgments.
These identities are the $\nabla${\em -strengthening rule}
$\nabla x. F = F$, provided $x$ is not free in $F$, and the
$\nabla${\em -exchange rule} $\nabla x\nabla y. F = \nabla y\nabla
x. F$.  In its essence, the \LG proof system \cite{tiu06lfmtp} is
obtained from \foldnb by strengthening its notion of equality based on
$\lambda$-conversion through the addition of these two structural
rules for $\nabla$.

The move from the weaker
logic \foldnb to the stronger logic \LG involves an ontological
commitment and has a proof-theoretic consequence.

At the ontological level, the strengthening rule implies that every
type at which one is
willing to use $\nabla$-quantification is non-empty and, in fact,
contains an unbounded number of members.  For example, the formula
$\exists_\tau x.\top$ is always provable, even if there are no closed
terms of type $\tau$ because this formula is equivalent to
$\nabla_\tau y.\exists_\tau x.\top$, which is provable.
Similarly, for any given $n\geq 1$, the following formula
is provable
\[
\exists_\tau x_1\ldots\exists_\tau x_n.
 \left[\bigwedge_{1\leq i,j\leq n, i\not= j} x_i \not= x_j\right].
\]

At the proof-theoretic level, an acceptance of the strengthening and
exchange rules means
that the length of a local context and the order of variables within
it are unimportant.  For example, a sequent that contains the generic
judgments $x_1,\ldots,x_n\triangleright B$  and
$y_1,\ldots,y_m\triangleright C$ can be rewritten (assuming  $n\ge m$)
using
$\alpha$-conversion and strengthening into the judgments
 $z_1,\ldots,z_n\triangleright B'$  and
$z_1,\ldots,z_n\triangleright C'$ where $B'$ and $C'$ are
equal to $B$ and $C$ modulo variable renamings.  In this fashion, all
local bindings in a sequent can be made to involve the same
variables, and, hence, the local bindings can be seen as a global
binding over a sequent that contains formulas and not generic
judgments.  The resulting sequent-level variable bindings will be represented by
specially designated {\it nominal constants}.
Notice, however, that each of these nominal ``constants'' has as its scope only a single
formula.  Thus, we must distinguish the same nominal constant when it
appears in two different formulas and we should treat judgments
as being equal if they are identical up to permutations of these
constants.

\subsection{A Sequent Calculus Presentation of the Core Logic}

The logic \logic inherits from \LG the shift from a local to a global
scope in the treatment of the $\nabla$-quantifier.  In particular,
we assume that the collection of constants is partitioned into the set
$\mathcal{C}$ of nominal constants and the set $\mathcal{K}$ of
usual, non-nominal constants.
We assume the set $\mathcal{C}$ contains an infinite number of nominal
constants for each type at which $\nabla$ quantification is permitted.
We define the {\it support} of a term (or
formula), written $\supp(t)$, as the set of nominal constants
appearing in it.
A permutation of nominal constants is a type-preserving bijection $\pi$ from
$\mathcal{C}$ to $\mathcal{C}$ such that $\{ x\ |\ \pi(x) \neq x\}$ is
finite.  We denote the application of such a
permutation to a term or formula $t$ by $\pi . t$ and define this as
follows:
\[
\begin{array}{l@{\qquad\qquad }l}
\pi.a = \pi(a), \mbox{ if $a \in \mathcal{C}$} &
\pi.c = c, \mbox{ if $c\notin \mathcal{C}$ is atomic} \\
\pi.(\lambda x.M) = \lambda x.(\pi.M) &
\pi.(M\; N) = (\pi.M)\; (\pi.N)
\end{array}
\]
We extend the notion of equality between terms to encompass also
the application of permutations to nominal constants appearing in
them. Specifically, we write $B \approx B'$ to denote the fact that
there is a permutation $\pi$ such that $B$ $\lambda$-converts to
$\pi.B'$. Using the observations that permutations are invertible and
composable and that $\lambda$-convertibility is an equivalence
relation, it is easy to see that $\approx$ is also an equivalence
relation.

\begin{figure*}[t]
\small
\begin{align*}
\infer[id]{\Sigma : \Gamma, B \lra B'}{B \approx B'} &&
\infer[\cut]{\Sigma : \Gamma, \Delta \lra C} {\Sigma : \Gamma \lra B &
  \Sigma : B, \Delta \lra C} &&
\infer[\cL]{\Sigma : \Gamma, B \lra C} {\Sigma : \Gamma, B, B \lra C}
\end{align*}
\begin{align*}
\infer[\botL]{\Sigma :\Gamma,\bot \lra C}{} &&
\infer[\topR]{\Sigma : \Gamma \lra \top}{}
\end{align*}
\begin{align*}
\infer[\lorL]{\Sigma :\Gamma,B\lor D \lra C} {\Sigma :\Gamma,B\lra C &
  \Sigma:\Gamma,D\lra C} &&
\infer[\lorR,i\in\{1,2\}]{\Sigma : \Gamma \lra B_1 \lor B_2} {\Sigma :
  \Gamma \lra B_i}
\end{align*}
\begin{align*}
\infer[\landL,i\in\{1,2\}]{\Sigma : \Gamma, B_1 \land B_2 \lra
  C}{\Sigma : \Gamma, B_i \lra C} &&
\infer[\landR]{\Sigma : \Gamma \lra B \land C}{\Sigma : \Gamma \lra B
  & \Sigma : \Gamma \lra C}
\end{align*}
\begin{align*}
\infer[\supsetL]{\Sigma : \Gamma, B \supset D \lra C} {\Sigma : \Gamma
  \lra B & \Sigma : \Gamma, D \lra C} &&
\infer[\supsetR]{\Sigma : \Gamma \lra B \supset C} {\Sigma : \Gamma, B
  \lra C}
\end{align*}
\begin{align*}
\infer[\forallL]{\Sigma : \Gamma, \forall_\tau x.B \lra C} {\Sigma,
  \mathcal{K}, \mathcal{C} \vdash t : \tau & \Sigma : \Gamma, B[t/x]
  \lra C} &&
\infer[\forallR,h\notin\Sigma, \supp(B)=\{\vec{c}\}]
{\Sigma : \Gamma \lra \forall x.B} {\Sigma, h : \Gamma \lra B[h\
  \vec{c}/x]}
\end{align*}
\begin{align*}
\infer[\existsL,h\notin\Sigma,\supp(B)=\{\vec{c}\}] {\Sigma : \Gamma,
  \exists x. B \lra C} {\Sigma, h : \Gamma, B[h\; \vec{c}/x] \lra C} &&
\infer[\existsR]{\Sigma : \Gamma \lra \exists_\tau x.B} {\Sigma,
  \mathcal{K}, \mathcal{C} \vdash t:\tau & \Sigma : \Gamma \lra
  B[t/x]}
\end{align*}
\begin{align*}
\infer[\nablaL,a\notin \supp(B)] {\Sigma : \Gamma, \nabla x. B \lra C}
{\Sigma : \Gamma, B[a/x] \lra C} &&
\infer[\nablaR,a\notin \supp(B)]
{\Sigma : \Gamma \lra \nabla x.B} {\Sigma : \Gamma \lra B[a/x]}
\end{align*}
\caption{The core rules of \logic}
\label{fig:core-rules}
\end{figure*}

The rules defining the core of \logic are presented in Figure
\ref{fig:core-rules}. Sequents in this logic have the form $\Sigma :
\Gamma \lra C$ where $\Gamma$ is a multiset and the signature $\Sigma$
contains all the free variables of $\Gamma$ and $C$. We use
expressions of the form $B[t/x]$ in the quantifier rules to denote the
result of substituting the term $t$ for $x$ in the formula $B$. Note
that such a substitution must be done carefully, making sure to rename
bound variables in $B$ to avoid capture of variables appearing in $t$.
In the $\nabla\mathcal{L}$ and $\nabla\mathcal{R}$ rules, $a$ denotes
a nominal constant of an appropriate type. In the $\exists\mathcal{L}$
and $\forall\mathcal{R}$ rule we use raising \cite{miller92jsc} to
encode the dependency of the quantified variable on the support of
$B$; the expression $(h\ \vec{c})$ in which $h$ is a fresh
eigenvariable is used in these two rules to denote the (curried)
application of $h$ to the constants appearing in the sequence
$\vec{c}$. The $\forall\mathcal{L}$ and $\exists\mathcal{R}$ rules
make use of judgments of the form $\Sigma, \mathcal{K}, \mathcal{C}
\vdash t : \tau$. These judgments enforce the requirement that the
expression $t$ instantiating the quantifier in the rule is a
well-formed term of type $\tau$ constructed from the eigenvariables in
$\Sigma$ and the constants in ${\cal K} \cup {\cal C}$. Notice that in
contrast the $\forall\mathcal{R}$ and $\exists\mathcal{L}$ rules seem
to allow for a dependency on only a restricted set of nominal
constants. However, this asymmetry is not significant:
Corollary~\ref{cor:extend} in Section~\ref{sec:meta-theory} will tell
us that the dependency expressed through raising in the latter rules
can be extended to any number of nominal constants that are not in the
relevant support set without affecting the provability of sequents.

Equality modulo $\lambda$-conversion is built into the rules in
Figure~\ref{fig:core-rules}, and also into later extensions of
this logic, in a fundamental way: in particular, proofs are preserved
under the replacement of formulas in sequents by ones to which they
$\lambda$-convert.  A more involved observation
is that we can replace a formula $B$ in a sequent by another formula
$B'$ such that $B\approx B'$ without affecting the provability of the
sequent or even the very structure of the proof. For the core logic,
this observation follows from the form of the $id$ rule and the fact
that permutations distribute over logical structure. We shall prove
this property explicitly for the full logic in
Chapter~\ref{ch:meta-theory}.

\section{Characterizing Occurrences of Nominal Constants}
\label{sec:nominal-abstraction}

We are interested in adding to our logic the capability of
characterizing occurrences of nominal constants within terms and also
of analyzing the structure of terms with respect to such
occurrences. For example, we may want to define a predicate called
{\em name} that holds of a term exactly when that term is a nominal
constant. Similarly, we might need to identify a binary relation
called {\em fresh} that holds between two terms just in the case that
the first term is a nominal constant that does not occur in the second
term. Towards supporting such possibilities, we define in this section
a special binary relation called {\it nominal abstraction} and then
present proof rules that incorporate an understanding of this relation
into the logic. A formalization of these ideas requires a careful
treatment of substitution. In particular, this operation must be
defined to respect the intended formula-level scope of nominal
constants. We begin our discussion with an elaboration of this aspect.

\subsection{Substitutions and their Interaction with Nominal Constants}

The following definition reiterates a common view of substitutions
in logical contexts.

\begin{definition}\label{subst}
A substitution is a type preserving mapping from variables
to terms that is the identity at all but a finite number of variables.
The domain of a substitution is the set of variables that are
not mapped to themselves and its range is the
set of terms resulting from applying it to the variables in its
domain.  We write a substitution as $\{t_1/x_1,\ldots,t_n/x_n\}$
where $x_1,\ldots,x_n$ is a list of variables that contains the
domain of the substitution and $t_1,\ldots,t_n$ is the value of the
map on these variables. The support of a substitution $\theta$,
written as $\supp(\theta)$, is the set of nominal constants that appear in
the range of $\theta$. The restriction of a substitution $\theta$ to
the set of variables $\Sigma$, written as $\theta \restrict
\Sigma$, is a mapping that is like $\theta$ on the variables in
$\Sigma$ and the identity everywhere else.
\end{definition}

A substitution essentially calls for the replacement of
variables by their associated terms in any context to which it is
applied. A complicating factor in our setting is that nominal
constants can appear in the terms that are to replace
particular variables. A substitution may be determined relative to one
formula in a sequent but may then have to be applied to other formulas
in the same sequent. In doing this, we have to take into account the
fact that the scopes of the implicit quantifiers over nominal
constants are restricted to individual formulas. Thus, the logically
correct application of a substitution should be accompanied by a
renaming of these constants in the term being substituted into so as to
ensure that they are not confused with the ones appearing in
the range of the substitution.

\begin{definition}\label{ncasubst}
The ordinary application of a substitution $\theta$ to a term $B$ is
denoted by $B[\theta]$ and corresponds to the replacement of the
variables in $B$ by the terms that $\theta$ maps them to, making sure,
as usual, to avoid accidental binding of the variables appearing in
the range of $\theta$. More precisely, if $\theta = \{t_1/x_1,\ldots,
t_n/x_n\}$, then $B[\theta]$ is the term $(\lambda x_1\ldots\lambda x_n.B)\;
t_1\; \ldots\; t_n$; this term is, of course, considered to be equal
to any other term that it $\lambda$-converts to. By contrast,
the {\em nominal capture avoiding application} of $\theta$ to $B$ is
written as $B\cas{\theta}$ and is defined as follows. Assuming that
$\pi$ is a permutation of nominal constants that maps those appearing
in $supp(B)$ to ones not appearing in $\supp(\theta)$, let $B' =
\pi.B$. Then $B\cas{\theta} = B'[\theta]$.
\end{definition}

The notation $B[\theta]$ generalizes the one
used in the quantifier rules in Figure~\ref{fig:core-rules}.
The definition of the nominal capture avoiding application of a
substitution is ambiguous in that we do not uniquely specify the
permutation to be used.  We resolve this ambiguity by deeming as
acceptable {\it any} permutation that avoids conflicts. As a special
instance of the lemma below, we see that for any given formula $B$ and
substitution $\theta$,
all the possible values for $B\cas{\theta}$ are equivalent modulo the
$\approx$ relation. Moreover, as we show in
Chapter~\ref{ch:meta-theory}, formulas that are equivalent under
$\approx$ are interchangeable in the contexts of proofs.

\begin{lemma}
\label{lem:approx-cas}
If $t \approx t'$ then $t\cas{\theta} \approx t'\cas{\theta}$.
\end{lemma}
\begin{proof}
Let $t$ be $\lambda$-convertible to $\pi_1.t'$, let $t\cas{\theta} =
(\pi_2.t)[\theta]$ where $\supp(\pi_2.t) \cap \supp(\theta) =
\emptyset$, and let $t'\cas{\theta}$ be $\lambda$-convertible to
$(\pi_3.t')[\theta]$ where $\supp(\pi_3.t')\cap \supp(\theta) =
\emptyset$. Then we define a function $\pi$ partially by the following
rules:
\begin{enumerate}
\item $\pi(c) = \pi_2.\pi_1.\pi_3^{-1}(c)$ if $c \in \supp(\pi_3.t')$ and
\item $\pi(c) = c$ if $c \in \supp(\theta)$.
\end{enumerate}
Since $\supp(\pi_3.t') \cap \supp(\theta) = \emptyset$, these rules
are not contradictory, \ie, this (partial) function is well-defined.
The range of the first rule is $\supp(\pi_2.\pi_1.\pi_3^{-1}.\pi_3.t')
= \supp(\pi_2.\pi_1.t') = \supp(\pi_2.t)$ which is disjoint from the
range of the second rule, $\supp(\theta)$. Since the mapping in each
rule is determined by a permutation, these rules together define a
one-to-one partial mapping that can be extended to a bijection on
$\mathcal{C}$. We take any such extension to be the complete
definition of $\pi$ that must therefore be a permutation.

To prove that $t\cas{\theta} \approx t'\cas{\theta}$ it suffices to
show that $(\pi_2.t)[\theta]$ is $\lambda$-convertible to
$\pi.((\pi_3.t')[\theta])$. We do this by induction on the structure
of $t'$ under the further assumption that $t$ $\lambda$-converts to
$\pi_1.t'$. Suppose $t'$ is an abstraction. Then, it is easy to see
that $(\pi_2.t)[\theta]$ $\lambda$-converts to $\lambda x.
((\pi_2.s)[\theta])$ and $\pi.((\pi_3.t')[\theta])$ $\lambda$-converts
to $\lambda x. (\pi.((\pi_3.s')[\theta]))$ for some choice of variable
$x$ and terms $s$ and $s'$ such that $s'$ is structurally less complex
than $t'$ and $s$ $\lambda$-converts to $\pi_1.s'$. But then, by the
induction hypothesis, $(\pi_2.s)[\theta]$ $\lambda$-converts to
$\pi.((\pi_3.s')[\theta])$ and hence $(\pi_2.t)[\theta]$ is
$\lambda$-convertible to $\pi.((\pi_3.t')[\theta])$. A similar and, in
fact, simpler argument can be provided in the case where $t'$ is an
application. If $t'$ is a nominal constant $c$ then
$(\pi_2.t)[\theta]$ must be $\lambda$-convertible to
$(\pi_2.\pi_1.c)[\theta] = \pi_2.\pi_1.c$. Also,
$\pi.((\pi_3.t')[\theta])$ must be $\lambda$-convertible to
$\pi.\pi_3.c$. Further, in this case the first rule for $\pi$ applies
which means $\pi.\pi_3.c = \pi_2.\pi_1.\pi_3^{-1}.\pi_3.c =
\pi_2.\pi_1.c$. Thus $(\pi_2.t)[\theta]$ is again
$\lambda$-convertible to $\pi.((\pi_3.t')[\theta])$. Finally, suppose
$t'$ is a variable $x$. In this case $t$ must be $\lambda$-convertible
to $x$ so that we must show $x[\theta]$ $\lambda$-converts to
$\pi.(x[\theta])$. If $x$ does not have a binding in $\theta$ then
both terms are equal. Alternatively, if $x[\theta] = s$ then $\pi.s =
s$ by the second rule for $\pi$ and so the two terms are again equal.
Thus $(\pi_2.t)[\theta]$ $\lambda$-converts to
$\pi.((\pi_3.t')[\theta])$, as is required.
\end{proof}

The nominal capture avoiding application of substitutions turns out to
be the dominant notion in the analysis of provability.
For this reason, when we speak of the application of a
substitution in an unqualified way, we shall mean the nominal capture
avoiding form of this notion.

We shall need to consider the composition of substitutions later in
this section. The definition of this notion must also pay attention to
the presence of nominal constants.

\begin{definition}\label{nascomp}
Given a substitution $\theta$ and a permutation $\pi$ of nominal
constants, let $\pi.\theta$ denote
the substitution that is obtained by replacing each $t/x$ in $\theta$
with $(\pi.t)/x$. Given any two substitutions $\theta$ and $\rho$, let
$\theta \circ \rho$ denote the substitution that is such that
$B[\theta\circ \rho] = B[\theta][\rho]$. In this context, the {\em
  nominal capture   avoiding composition} of $\theta$ and $\rho$ is
written as $\theta\bullet\rho$ and defined as follows. Let $\pi$ be a
permutation of nominal constants such that
$\supp(\pi.\theta)$ is disjoint from $\supp(\rho)$. Then
$\theta\bullet\rho = (\pi.\theta)\circ \rho$.
\end{definition}\label{substequiv}
The notation $\theta \circ \rho$ in the above definition represents
the usual composition of $\theta$ and $\rho$ and can, in fact, be
given in an explicit form based on these substitutions. Thus, $\theta
\bullet \rho$ can also be presented in an explicit form. Notice that our
definition of nominal capture avoiding composition is, once again,
ambiguous because it does not fix the permutation to be used,
accepting instead any one that satisfies the constraints. However, as
before, this ambiguity is harmless. To understand this, we first
extend the notion of equivalence under permutations to substitutions.
\begin{definition}
Two substitutions $\theta$ and $\rho$ are considered to be permutation
equivalent, written $\theta \approx \rho$, if and only if there is a
permutation of nominal constants $\pi$ such that $\theta =
\pi.\rho$. This notion of equivalence may also be parameterized by a
set of variables $\Sigma$ as follows: $\theta \approx_\Sigma \rho$
just in the case that $\theta \restrict \Sigma \approx \rho \restrict
\Sigma$.
\end{definition}
It is easy to see that all possible choices for $\theta \bullet \rho$
are permutation equivalent and that if $\varphi_1 \approx \varphi_2$
then $B\cas{\varphi_1} \approx B\cas{\varphi_2}$ for any term $B$.
Thus, if our focus is on provability, the ambiguity in
Definition~\ref{nascomp} is inconsequential by a result to be
established in Chapter~\ref{ch:meta-theory}. As a further observation,
note that $B\cas{\theta\bullet\rho} \approx B\cas{\theta}\cas{\rho}$
for any $B$. Hence our notion of nominal capture avoiding composition
of substitutions is sensible.

The composition operation can be used to define an ordering
relation between substitutions:
\begin{definition}\label{nasordering}
Given two substitutions $\rho$ and $\theta$, we say $\rho$ is {\em
  less general than} $\theta$, notated as $\rho \leq \theta$, if and
only if there exists a $\sigma$ such that $\rho \approx
\theta\bullet\sigma$. This relation can also be parameterized by a
set of variables: $\rho$ is less general than $\theta$
relative to $\Sigma$, written as $\rho \leq_\Sigma \theta$, if and
only if $\rho
\restrict \Sigma \leq \theta \restrict \Sigma$.
\end{definition}
The notion of generality between substitutions that is based on
nominal capture avoiding composition has a different flavor from that
based on the traditional form of substitution composition. For
example, if $a$ is a nominal constant, the substitution $\{a/x\}$ is
strictly less general than $\{a/x, y' a/y\}$ relative to $\Sigma$ for
any $\Sigma$ which contains $x$ and $y$. To see this, note that we can
compose the latter substitution with $\{(\lambda z.y)/y'\}$ to obtain
the former, but the naive attempt to compose the former with
$\{y'a/y\}$ yields $\{b/x, y'a/y\}$ where $b$ is a nominal constant
distinct from $a$. In fact, the ``most general'' solution relative to
$\Sigma$ containing $\{a/x\}$ will be $\{a/x\}\cup \{z'a/z \mid
z\in\Sigma\backslash\{x\}\}$.

\subsection{Nominal Abstraction}

The nominal abstraction relation allows implicit formula-level
bindings represented by nominal constants to be moved into explicit
abstractions over terms. The following notation is useful for
defining this relationship.

\begin{notation}
Let $t$ be a term, let $c_1,\ldots,c_n$ be distinct nominal constants that
possibly occur in $t$, and let $y_1,\ldots,y_n$ be distinct variables
not occurring in $t$ and such that, for $1 \leq i \leq n$, $y_i$ and
$c_i$ have the same type. Then we write $\lambda c_1
\ldots\lambda c_n. t$ to denote the term $\lambda y_1 \ldots \lambda
y_n . t'$ where $t'$ is the term obtained from $t$ by replacing
$c_i$ by $y_i$ for $1\leq i\leq n$.
\end{notation}

There is an ambiguity in the notation introduced above in that
the choice of variables $y_1,\ldots,y_n$ is not fixed. However, this
ambiguity is harmless: the terms that are produced by acceptable
choices are all equivalent under a renaming of bound variables.

\begin{definition}\label{nominal-abstraction}
Let $n\ge 0$ and
let $s$ and $t$ be terms of type $\tau_1 \to \cdots \to \tau_n \to
\tau$ and $\tau$, respectively; notice, in particular, that $s$ takes
$n$ arguments to yield a term of the same type as $t$.
Then the expression $s \unrhd t$ is a formula that is referred to as a
nominal abstraction of degree $n$ or simply as a nominal abstraction. The
symbol $\unrhd$ is used here in an overloaded way in that the degree
of the nominal abstraction it participates in can vary.
The nominal abstraction $s \unrhd t$ of degree $n$ is said to hold just in
the case that $s$ $\lambda$-converts to $\lambda c_1\ldots c_n.t$ for
some nominal constants $c_1,\ldots,c_n$.
\end{definition}

Clearly, nominal abstraction of degree $0$ is the same as equality
between terms based on $\lambda$-conversion, and we will therefore use
$=$  to denote this relation in that situation. In the more general
case,
the term on the left of the operator serves as a pattern for isolating
occurrences of nominal constants. For example, the relation $(\lambda
x.x) \unrhd t$ holds exactly when $t$ is a nominal constant.


The symbol $\unrhd$ corresponds, at the moment, to a mathematical
relation that holds between pairs of terms as explicated by
Definition~\ref{nominal-abstraction}. We now overload this symbol by
treating it also as a binary predicate symbol of \logic. In the next
subsection we shall add inference rules to make the mathematical
understanding of $\unrhd$ coincide with its syntactic use as a
predicate in sequents. It is, of course, necessary to be able to
determine when we mean to use $\unrhd$ in the mathematical sense and
when as a logical symbol. When we write an expression such as $s\unrhd
t$ without qualification, this should be read as a logical formula
whereas if we say that ``$s\unrhd t$ holds'' then we are referring to
the abstract relation from Definition~\ref{nominal-abstraction}. We
might also sometimes use an expression such as ``$(s\unrhd
t)\cas{\theta}$ holds.'' In this case, we first treat $s \unrhd t$ as
a formula to which we apply the substitution $\theta$ in a nominal
capture avoiding way to get a (syntactic) expression of the form
$s'\unrhd t'$. We then read $\unrhd$ in the mathematical sense,
interpreting the overall expression as the assertion that ``$s'\unrhd
t'$ holds.'' Note in this context that $s \unrhd t$ constitutes a
single formula when read syntactically and hence the expression
$(s\unrhd t)\cas{\theta}$ is, in general, {\it not} equivalent to the
expression $s\cas{\theta}\unrhd t\cas{\theta}$.

In the proof-theoretic setting, nominal abstraction will be used with
terms that contain free occurrences of variables for which
substitutions can be made. The following definition is relevant to
this situation.

\begin{definition}\label{nasolution}
A substitution $\theta$ is said to be a solution to the nominal
abstraction $s \unrhd t$ just in the case that $(s\unrhd t)\cas{\theta}$ holds.
\end{definition}

Solutions to a nominal abstraction can be used to provide rich
characterizations of the structures of terms. For example, consider
the nominal abstraction
$(\lambda x.\fresh x T) \unrhd S$ in which $T$ and $S$ are
variables and {\sl fresh} is a binary predicate symbol.  Any solution
to this problem requires that $S$ be
substituted for by a term of the form $\fresh a R$ where $a$
is a nominal constant and $R$ is a term in which $a$ does not appear,
\ie, $a$ must be ``fresh'' to $R$.

An important property of solutions to a nominal abstraction is that
these are preserved under permutations to nominal constants. We
establish this fact in the lemma below; this lemma will be used later
in showing the stability of the provability of sequents with
respect to the replacement of formulas by ones they are equivalent to
modulo the $\approx$ relation.

\begin{lemma}
\label{lem:na-approx}
Suppose $(s\unrhd t) \approx (s'\unrhd t')$. Then $s\unrhd t$ and
$s'\unrhd t'$ have exactly the same solutions. In particular,
$s\unrhd t$ holds if and only if $s'\unrhd t'$ holds.
\end{lemma}
\begin{proof}
We prove the particular result first. It suffices to only show it in
the forward direction since $\approx$ is symmetric. Let $\pi$ be the
permutation such that the expression $s'\unrhd t'$ $\lambda$-converts to
$\pi.(s\unrhd t)$. Now suppose $s \unrhd t$ holds since $s$
$\lambda$-converts to $\lambda\vec{c}.t$. Then $s'$ will
$\lambda$-convert to $\lambda(\pi.\vec{c}).t'$ where $\pi.\vec{c}$ is
the result of applying $\pi$ to each element in the sequence $\vec{c}$.
Thus $s' \unrhd t'$ holds.

For the general result it again suffices to show it in one direction,
\ie, that all the solutions of $s\unrhd t$ are solutions to $s'\unrhd
t'$. Let $\theta$ be a substitution such that $(s\unrhd
t)\cas{\theta}$ holds. By Lemma~\ref{lem:approx-cas}, $(s\unrhd
t)\cas{\theta} \approx (s'\unrhd t')\cas{\theta}$. Thus by the
particular result from the first half of this proof, $(s'\unrhd
t')\cas{\theta}$ holds.
\end{proof}

\subsection{Proof Rules for Nominal Abstraction}

\begin{figure}[t]
\small
\begin{align*}
\infer[\unrhdL]{\Sigma : \Gamma, s \unrhd t \lra C}
{\left\{\Sigma\theta : \Gamma\cas{\theta} \lra C\cas{\theta} \;|\;
  \hbox{$\theta$ is a solution to $(s \unrhd t)$}
  \right\}_\theta}
&&
\infer[\unrhdR,\ \hbox{$s \unrhd t$ holds}]
{\Sigma : \Gamma \lra s \unrhd t}
{}
\end{align*}
\caption{Nominal abstraction rules}
\label{fig:na-rules}
\bigskip
\begin{equation*}
\infer[\unrhdL_{\CSNAS}]
       {\Sigma : \Gamma, s \unrhd t \lra C}
       {\left\{\Sigma\theta : \Gamma\cas{\theta} \lra C\cas{\theta}
            \;|\;
          \theta \in \CSNAS(\Sigma, s, t)
        \right\}_\theta}
\end{equation*}
\caption{A variant of $\unrhdL$ based on \CSNAS}
\label{fig:csnas}
\end{figure}

We now add the left and right introduction rules for $\unrhd$ that are
shown in Figure~\ref{fig:na-rules} to link its use as a predicate
symbol to its mathematical interpretation. The expression $\Sigma
\theta$ in the $\unrhdL$ rule denotes the application of a
substitution $\theta=\{t_1/x_1,\ldots,t_n/x_n\}$ to the signature
$\Sigma$ that is defined to be the signature that results from
removing from $\Sigma$ the variables $\{x_1,\ldots,x_n\}$ and then
adding every variable that is free in any term in
$\{t_1,\ldots,t_n\}$. Notice also that in the same inference rule the
operator $\cas{\theta}$ is applied to a multiset of formulas in the
natural way: $\Gamma\cas{\theta}=\{B\cas{\theta}\;|\; B\in\Gamma\}$.
Note that the $\unrhdL$ rule has an {\it a priori} unspecified number
of premises that depends on the number of substitutions that are
solutions to the relevant nominal abstraction. If $s \unrhd t$
expresses an unsatisfiable constraint, meaning that it has no
solutions, then the premise of $\unrhdL$ is empty and the rule
provides an immediate proof of its conclusion.

The $\unrhdL$ and $\unrhdR$ rules capture nicely the intended
interpretation of nominal abstraction. However, there is an obstacle
to using the former rule in derivations: this rule has an infinite
number of premises any time the nominal abstraction $s \unrhd t$ has a
solution. We can overcome this difficulty by describing a rule that
includes only a few of these premises but in such way that their
provability ensures the provability of all the other premises.  Since
the provability of $\Gamma \lra C$ implies the provability of
$\Gamma\cas{\theta} \lra C\cas{\theta}$ for any $\theta$ (a property
established formally in Chapter\ref{ch:meta-theory}), if the first
sequent is a premise of an occurrence of the $\unrhdL$ rule, the
second does not need to be used as a premise of that same rule
occurrence.  Thus, we can limit the set of premises to be considered
if we can identify with any given nominal abstraction a (possibly
finite) set of solutions from which any other solution can be obtained
through composition with a suitable substitution. The following
definition formalizes the idea of such a ``covering set.''

\begin{definition}\label{csnas}
A {\em complete set of nominal abstraction solutions} (\CSNAS) of $s$
and $t$ on $\Sigma$
is a set $S$ of substitutions such
that
\begin{enumerate}
\item each $\theta \in S$ is a solution to $s \unrhd t$, and
\item for every solution $\rho$ to $s \unrhd t$, there exists a
$\theta \in S$ such that $\rho \leq_\Sigma \theta$.
\end{enumerate}
We denote any such set by $\CSNAS(\Sigma, s, t)$.
\end{definition}
Using this definition we present an alternative version of $\unrhdL$
in Figure~\ref{fig:csnas}. Note that if we can find a finite complete
set of nominal abstraction solutions then the number of premises to
this rule will be finite.

\begin{theorem}\label{thm:csnas}
The rules $\unrhdL$ and $\unrhdL_{\CSNAS}$ are inter-admissible.
\end{theorem}
\begin{proof}
Suppose we have the following arbitrary instance of $\unrhdL$ in a
derivation:
\begin{equation*}
\infer
 [\unrhdL]
 {\Sigma : \Gamma, s \unrhd t \lra C}
 {\left\{
\Sigma\theta : \Gamma\cas{\theta} \lra C\cas{\theta} \;|\;
     \hbox{$\theta$ is a solution to $(s \unrhd t)$}
\right\}_\theta}
\end{equation*}
This rule can be replaced with a use of $\unrhdL_{\CSNAS}$ instead if
we could be certain that, for each $\rho \in \CSNAS(\Sigma,s,t)$, it is
the case that $\Sigma\rho : \Gamma\cas{\rho} \lra
C\cas{\rho}$ is included in the set of premises of the shown rule
instance. But this must be the case: by the
definition  of $\CSNAS$, each such $\rho$ is a solution to $s \unrhd
t$.

In the other direction, suppose we have the following arbitrary
instance of $\unrhdL_{\CSNAS}$.
\begin{equation*}
\infer
 [\unrhdL_{\CSNAS}]
 {\Sigma : \Gamma, s \unrhd t \lra C}
 {\left\{
\Sigma\theta : \Gamma\cas{\theta} \lra C\cas{\theta}
      \;|\;
    \theta \in \CSNAS(\Sigma, s, t)
\right\}_\theta}
\end{equation*}
To replace this rule with a use of the $\unrhdL$ rule
instead, we need to be able to construct a
derivation of $\Sigma\rho : \Gamma\cas{\rho} \lra C\cas{\rho}$ for
each $\rho$ that is a solution to $s \unrhd t$. By the definition of
$\CSNAS$, we know that for any such $\rho$ there exists a $\theta \in
\CSNAS(\Sigma,s,t)$ such that $\rho \leq_\Sigma \theta$, \ie, such
that there
exists a $\sigma$ for which $\rho\restrict\Sigma \approx
(\theta\restrict\Sigma) \bullet \sigma$. Since we are considering the
application of these substitutions to a sequent all of whose
eigenvariables are contained in $\Sigma$, we can drop the restriction
on the substitutions and suppose that $\rho \approx \theta \bullet
\sigma$. Now, we shall show in Chapter~\ref{ch:meta-theory} that if a
sequent has a derivation then the result of applying a substitution to
it in a nominal capture-avoiding way produces a sequent that also has
a derivation. Using this observation, it follows that
$\Sigma\theta\sigma : \Gamma\cas{\theta}\cas{\sigma} \lra
C\cas{\theta}\cas{\sigma}$  has a proof. But this sequent is
permutation equivalent to $\Sigma\rho : \Gamma\cas{\rho} \lra
C\cas{\rho}$ which must, again by a result established explicitly in
Chapter~\ref{ch:meta-theory}, also have a proof.
\end{proof}

Theorem~\ref{thm:csnas} allows us to choose which of the left rules we
wish to consider in any given context. We shall assume the $\unrhdL$
rule in the formal treatment in the rest of this thesis, leaving the
use of the $\unrhdL_{\CSNAS}$ rule to practical applications of the
logic.

\subsection{Computing Complete Sets of Nominal Abstraction
  Solutions}\label{ssec:complete-sets}

For the $\unrhdL_{CSNAS}$ rule to be useful, we need an effective way
to compute restricted complete sets of nominal abstraction
solutions. We show here that the task of finding such complete sets of
solutions can be reduced to that of finding complete sets of unifiers
(\CSU) for higher-order unification problems \cite{huet75tcs}. In the
straightforward approach to finding a solution to a nominal
abstraction $s \unrhd t$, we would first identify a substitution
$\theta$ that we apply to $s \unrhd t$ to get $s' \unrhd t'$ and we
would subsequently look for nominal constants to abstract from $t'$ to
get $s'$.  To relate this problem to the usual notion of unification,
we would like to invert this order: in particular, we would like to
consider all possible ways of abstracting over nominal constants first
and only later think of applying substitutions to make the terms
equal. The difficulty with this second approach is that we do not know
which nominal constants might appear in $t'$ until after the
substitution is applied. However, there is a way around this
problem. Given the nominal abstraction $s \unrhd t$ of degree $n$, we
first consider substitutions for the variables occurring in it that
introduce $n$ new nominal constants in a completely general way.  Then
we consider all possible ways of abstracting over the nominal
constants appearing in the altered form of $t$ and, for each of these
cases, we look for a complete set of unifiers.

The idea described above is formalized in the following definition and
associated theorem. We use the notation $\CSU(s,t)$ in them to denote
an arbitrary but fixed selection of a complete set of unifiers for
the terms $s$ and $t$.

\begin{definition}\label{def:s}
Let $s$ and $t$ be terms of type $\tau_1 \to \ldots \to \tau_n
\to \tau$ and $\tau$, respectively. Let $c_1,\ldots,c_n$ be $n$
distinct nominal constants disjoint from $\supp(s\unrhd t)$ such that,
for $1 \leq i \leq n$, $c_i$ has the type $\tau_i$. Let $\Sigma$ be a
set of variables and for each $h \in \Sigma$ of type $\tau'$, let
$h'$ be a distinct variable not in $\Sigma$ that has type
$\tau_1\to \ldots\to\tau_n\to \tau'$. Let $\sigma = \{h'\ c_1\ \ldots\
c_n/h \mid h \in \Sigma\}$ and let $s' = s[\sigma]$ and $t' =
t[\sigma]$. Let
\begin{equation*}
C = \bigcup_{\vec{a}} \CSU(\lambda \vec{b}.s', \lambda \vec{b}. \lambda\vec{a}.t')
\end{equation*}
where $\vec{a}= a_1,\ldots,a_n$ ranges over all selections of $n$
distinct nominal constants from $\supp(t)\cup \{\vec{c}\}$ such that,
for $1 \leq i \leq n$,
$a_i$ has type $\tau_i$ and $\vec{b}$ is some corresponding listing of
all the nominal constants in $s'$ and $t'$ that are not included in
$\vec{a}$.
Then we define
\begin{equation*}
S(\Sigma, s, t) = \{ \sigma \bullet \rho \mid \rho \in C \}
\end{equation*}
\end{definition}

The use of the substitution $\sigma$ above represents
another instance of the application of the general technique of
raising that allows
certain variables (the $h$ variables in this definition) whose
substitution instances might depend on certain nominal constants
($c_1,\ldots,c_n$ here) to be replaced by new variables of higher type
(the $h'$ variables) whose substitution instances are not allowed to
depend on those nominal constants. This technique was previously used
in the $\existsL$ and $\forallR$ rules presented in
Section~\ref{sec:logic}.

\begin{theorem}
$S(\Sigma, s, t)$ is a complete set of nominal abstraction solutions
for $s\unrhd t$ on $\Sigma$.
\end{theorem}
\begin{proof}
First note that $\supp(\sigma) \cap \supp(s\unrhd t) = \emptyset$ and
thus $(s\unrhd t)\cas{\theta}$ is equal to $(s' \unrhd t')$.
Now we must show that every element of $S(\Sigma, s, t)$ is a
solution to $s \unrhd t$. Let $\sigma\bullet\rho \in S(\Sigma, s, t)$
be an arbitrary element where $\sigma$ is as in Definition~\ref{def:s},
$\rho$ is from $\CSU(\lambda\vec{b}.s',
\lambda\vec{b}.\lambda\vec{a}.t')$, and $s' = s[\sigma]$ and $t' =
t[\sigma]$.  By the definition of $\CSU$ we know $(\lambda\vec{b}.s' =
\lambda\vec{b}.\lambda\vec{a}.t')[\rho]$. This means $(s' =
\lambda\vec{a}.t')\cas{\rho}$ holds and thus $(s' \unrhd
t')\cas{\rho}$ holds. Rewriting $s'$ and $t'$ in terms of $s$ and $t$ this
means $(s \unrhd t)\cas{\sigma}\cas{\rho}$. Thus $\sigma\bullet\rho$
is a solution to $s\unrhd t$.

In the other direction, we must show that if $\theta$ is a solution to $s
\unrhd t$ then there exists $\sigma\bullet\rho \in S(\Sigma, s, t)$
such that $\theta \le_\Sigma \sigma\bullet\rho$. Let $\theta$ be a
solution to $s\unrhd t$. Then we know $(s\unrhd t)\cas{\theta}$ holds.
The substitution $\theta$ may introduce some nominal constants which
are abstracted out of the right-hand side when determining equality, so
let us call these the {\em important} nominal constants. Let $\sigma =
\{h'\ c_1\ \ldots\ c_n/h \mid h \in \Sigma\}$ be as in
Definition~\ref{def:s} and let $\pi'$ be a permutation which maps the
important nominal constants of $\theta$ to nominal constants from
$c_1, \ldots, c_n$. This is possible since $n$ nominal constants are
abstract from the right-hand side and thus there are at most $n$
important nominal constants. Then let $\theta' = \pi'.\theta$, so that
$(s\unrhd t)\cas{\theta'}$ holds and it suffices to show that $\theta'
\le_\Sigma \sigma\bullet\rho$. Note that all we have done at this
point is to rename the important nominal constants of $\theta$ so that
they match those introduced by $\sigma$. Now we define $\rho' = \{
\lambda c_1\ldots\lambda c_n.r / h' \mid r / h \in \theta'\}$ so that
$\theta' = \sigma\bullet\rho'$. Thus $(s\unrhd
t)\cas{\sigma}\cas{\rho'}$ holds. By construction, $\sigma$ shares no
nominal constants with $s$ and $t$, thus we know $(s'\unrhd
t')\cas{\rho'}$ where $s' = s[\sigma]$ and $t' = t[\sigma]$. Also by
construction, $\rho'$ contains no interesting nominal constants and
thus $(s' = \lambda\vec{a}. t')\cas{\rho}$ holds for some nominal
constants $\vec{a}$ taken from $\supp(t) \cup \{\vec{c}\}$. If we let
$\vec{b}$ be a listing of all nominal constants in $s'$ and $t'$ but
not in $\vec{a}$, then $(\lambda\vec{b}. s' =
\lambda\vec{b}.\lambda\vec{a}. t')\cas{\rho}$ holds. At this point the inner
equality has no nominal constants and thus the substitution $\rho$ can
be applied without renaming: $(\lambda\vec{b}. s' =
\lambda\vec{b}.\lambda\vec{a}. t')[\rho']$ holds. By the definition of
$\CSU$, there must be a $\rho \in \CSU(\lambda\vec{b}.s',
\lambda\vec{b}.\lambda\vec{a}.t')$ such that $\rho' \le \rho$. Thus
$\sigma\bullet\rho' \le_\Sigma \sigma\bullet\rho$ as desired.
\end{proof}

\section{Definitions, Induction, and Co-induction}
\label{sec:definitions}

\begin{figure}[t]
\begin{center}
$\infer[\defL]
      {\Sigma : \Gamma, p\ \vec{t} \lra C}
      {\Sigma : \Gamma, B\ p\ \vec{t} \lra C}
\hspace{1in}
\infer[\defR]
      {\Sigma : \Gamma \lra p\ \vec{t}}
      {\Sigma : \Gamma \lra B\ p\ \vec{t}}
$
\end{center}
\caption{Introduction rules for atoms whose predicate is defined as $\forall
  \vec{x}.~p\ \vec{x} \triangleq B\ p\ \vec{x}$}
\label{fig:defrules}
\end{figure}

The sequent calculus rules presented in Figure~\ref{fig:core-rules}
treat atomic judgments as fixed, unanalyzed objects.
We now add the capability of defining such judgments by means of
formulas, possibly involving other predicates. In particular, we shall
assume that we are given a
fixed, finite set of \emph{clauses} of the
form $\forall \vec{x}.~p\ \vec{x} \triangleq B\ p\ \vec{x}$ where $p$
is a predicate constant that takes a number of arguments equal to the
length of $\vec{x}$. Such a clause is said to define $p$ and the
entire collection of clauses is called a {\em
  definition}. The expression $B$, called the {\em body} of the
clause, must be a term that does not contain $p$ or
any of the variables in $\vec{x}$ and must have a type such that
$B\ p\ \vec{x}$ has type $o$.  Definitions are also restricted so that
a predicate is defined by at most one clause.
The intended interpretation of a clause $\forall \vec{x}.~p\ \vec{x}
\triangleq B\ p\ \vec{x}$ is that the atomic
formula $p\ \vec{t}$, where $\vec{t}$ is a list of terms of the same
length and type as the variables in $\vec{x}$, is true if and only if
$B\ p\ \vec{t}$ is true.
This interpretation is realized by adding to the calculus the rules
$\defL$ and $\defR$ shown in Figure~\ref{fig:defrules} for unfolding
predicates on the left and the right of sequents using their defining
clauses.

A definition can have a recursive structure. For example, in the clause
$\forall \vec{x}.~p\ \vec{x} \triangleq B\ p\ \vec{x}$, the predicate
$p$ can appear free in $B\ p\ \vec{x}$.  In this setting, the meanings
of predicates are intended to be given by any one of the fixed points
that can be associated with the definition.  Such an interpretation may
not always be sensible. In particular, without further restrictions,
the resulting proof system may not be consistent.  There are two
constraints that suffice to ensure consistency. First, the body of a
clause must not contain any nominal constants. This restriction
can be justified from another perspective as well: as we see in
Chapter~\ref{ch:meta-theory}, it helps in establishing that $\approx$
is a provability preserving equivalence between formulas. Second, definitions
should be {\em stratified} so that clauses, such as $a\triangleq
(a\supset \bot)$, in which a predicate has a negative dependency on
itself, are forbidden.  While such stratification can be enforced in
different ways, we use a simple approach to doing this in this
thesis. This approach is based on associating with each predicate $p$
a natural number that is called its {\em level} and that is denoted
by $\lvl(p)$.  This measure is then extended to arbitrary formulas by
the following definition.
\begin{definition}
Given an assignment of levels to predicates, the function $\lvl$ is
extended to all formulas in $\lambda$-normal form as follows:
\begin{enumerate}
\item $\lvl(p\ \bar{t}) = \lvl(p)$
\item $\lvl(\bot) = \lvl(\top) = \lvl(s\unrhd t) = 0$
\item $\lvl(B\land C) = \lvl(B\lor C) = \max(\lvl(B),\lvl(C))$
\item $\lvl(B\supset C) = \max(\lvl(B)+1,\lvl(C))$
\item $\lvl(\forall x.B) = \lvl(\nabla x.B) = \lvl(\exists x.B) =
\lvl(B)$
\end{enumerate}
In general, the level of a formula $B$, written as $\lvl(B)$, is the
level of its $\lambda$-normal form.
\end{definition}

A definition is {\em stratified} if we can assign levels to predicates
in such a way that $\lvl(B\ p\ \vec{x}) \leq \lvl(p)$ for each clause
$\forall \vec{x}.~p\ \vec{x} \triangleq B\ p\ \vec{x}$ in that
definition.

\begin{figure}[t]
\begin{center}
$\infer[\IL]
      {\Sigma : \Gamma, p\; \vec{t} \lra C}
      {\vec{x} : B\; S\; \vec{x} \lra S\; \vec{x} \qquad
       \Sigma : \Gamma, S\; \vec{t} \lra C}$\\[5pt]
provided $p$ is defined as $\forall  \vec{x}.~ p\  \vec{x}
\mueq  B\ p\  \vec{x}$ and $S$ is a term that has the same type as $p$
      \\[15pt]
$\infer[\CIR]
      {\Sigma : \Gamma \lra p\; \vec{t}}
      {\Sigma : \Gamma \lra S\; \vec{t} \qquad
       \vec{x} : S\; \vec{x} \lra B\; S\; \vec{x}}
$\\[5pt]
provided $p$ is defined as $\forall  \vec{x}.~ p\  \vec{x}
\nueq  B\ p\  \vec{x}$ and $S$ is a term that has the same type as $p$
\end{center}
\caption{The induction left and co-induction right rules}
\label{fig:indandcoind}
\end{figure}

The $\defL$ and $\defR$ rules do not discriminate between any of the
fixed points of a definition.
We now allow the selection of least and greatest fixed points so as to
support inductive and co-inductive definitions of predicates.
Specifically, we denote an inductive clause by $\forall \vec{x}.~ p\
\vec{x} \mueq B\ p\ \vec{x}$ and a co-inductive one by $\forall
\vec{x}.~ p\ \vec{x} \nueq B\ p\ \vec{x}$. As a refinement of the
earlier restriction on definitions, a predicate may have at most one
defining clause that is designated to be inductive, co-inductive or
neither. The $\defL$ and $\defR$ rules may be used with clauses in any
one of these forms. Clauses that are inductive admit additionally the
left rule $\IL$ shown in Figure~\ref{fig:indandcoind}. This rule is
based on the observation that the least fixed point of a monotone
operator is the intersection of all its pre-fixed points; intuitively,
anything that follows from any pre-fixed point should then also follow
from the least fixed point. In a proof search setting, the term
corresponding to the schema variable $S$ in this rule functions like
the induction hypothesis and is accordingly called the invariant of
the induction. Clauses that are co-inductive, on the other hand, admit
the right rule $\CIR$ also presented in Figure~\ref{fig:indandcoind}.
This rule reflects the fact that the greatest fixed point of a
monotone operator is the union of all the post-fixed points; any
member of such a post-fixed point must therefore also be a member of
the greatest fixed point. The substitution that is used for $S$ in
this rule is called the co-invariant or the simulation of the
co-induction. Just like the restriction on the body of clauses, in
both $\IL$ and $\CIR$, the (co-)invariant $S$ must not contain any
nominal constants.

As a simple illustration of the use of these rules, consider the
clause $p \mueq p$. The desired
inductive reading of this clause
implies that $p$ must be false. In a proof-theoretic setting, we would
therefore expect that the sequent $\cdot : p \lra \bot$ can be
proved. This can, in fact, be done by using $\IL$ with the invariant
$S = \bot$. On the other hand, consider the clause $q
\nueq q$. The co-inductive reading intended here implies that $q$ must
be true. The logic \logic satisfies this expectation: the
sequent $\cdot : \cdot \lra q$ can be proved using $\CIR$ with the
co-invariant $S = \top$.

The addition of inductive and co-inductive forms of clauses and the
mixing of these forms in one setting might be expected to require
stronger conditions than those described earlier in this section to
guarantee consistency. One condition, in addition to the absence of
nominal constants in the bodies of clauses and stratification based on
levels, that suffices and that is also practically acceptable is the
following that is taken from \cite{tiu.momigliano}: in a clause of any
of the forms $\forall \vec{x}.~ p\ \vec{x} \triangleq B\ p\ \vec{x}$,
$\forall \vec{x}.~ p\ \vec{x} \mueq B\ p\ \vec{x}$ or $\forall
\vec{x}.~ p\ \vec{x} \nueq B\ p\ \vec{x}$, it must be that $\lvl(B\
(\lambda\vec{x}.\top)\ \vec{x}) < \lvl(p)$. This disallows any mutual
recursion between clauses, a restriction which can easily be overcome
by merging mutually recursive clauses into a single clause. We
henceforth assume that all definitions satisfy all three conditions
described for them in this section.
Corollary \ref{consistency} in
Chapter~\ref{ch:meta-theory} establishes the consistency of the logic
under these restrictions.

\section{A Pattern-Based Form for Definitions}
\label{sec:pattern-form}

When presenting a definition for a predicate, it is often convenient
to write this as a collection of clauses whose applicability is also
constrained by patterns appearing in the head. For example, in logics
that support equality but not nominal abstraction, list membership
may be defined by the two pattern based clauses shown below.
\begin{equation*}
\member X (X::L) \triangleq \top \hspace{2cm}
\member X (Y::L) \triangleq \member X L
\end{equation*}
These logics also include rules for directly treating definitions
presented in this way. In understanding these rules, use may be made
of the translation of the extended form of definitions to a version
that does not use patterns in the head and in which there is at most
one clause for each predicate. For example, the definition of the list
membership predicate would be translated to the following form:
\begin{equation*}
\member X K \triangleq (\exists L.~ K = (X :: L)) \lor
 (\exists Y \exists L.~ K = (Y :: L) \land \member X L)
\end{equation*}
The treatment of patterns and multiple clauses can now be understood
in terms of the rules for definitions using a single clause and the
rules for equality, disjunction, and existential quantification.

In the logic \logic, the notion of equality has been generalized to
that of nominal abstraction. This allows us also to expand the
pattern-based form of definitions to use nominal abstraction in
determining the selection of clauses. By doing this, we would allow
the head of a clausal definition to describe not only the term
structure of the arguments, but also to place restrictions on the
occurrences of nominal constants in these arguments.
For example, suppose we want to describe the contexts in typing
judgments by lists of the form $\of {c_1} {T_1} :: \of {c_2} {T_2} ::
\ldots :: nil$ with the further proviso that each $c_i$ is a distinct
nominal constant. We will allow this to be done by using the following
pattern-based form of definition for the predicate $\ctx$:
\begin{equation*}
\ctx nil \triangleq \top \hspace{2cm}
(\nabla x. \ctx (\of x T :: L)) \triangleq \ctx L
\end{equation*}
Intuitively, the $\nabla$ quantifier in the head of the second clause
imposes the requirement that, to match it, the argument of $\ctx$
should have the form $\of x T :: L$ where $x$ is a nominal constant
that does not occur in either $T$ or $L$. To understand this
interpretation, we could think of the earlier definition of {\sl ctx}
as corresponding to the following one that does not use patterns or
multiple clauses:
\begin{equation*}
\ctx K \triangleq (K = nil) \lor
(\exists T \exists L.~ (\lambda x . \of x T :: L) \unrhd K \land \ctx
L)
\end{equation*}
Our objective in the rest of this section is to develop machinery
for allowing the extended form of definitions to be used directly. We
do this by presenting its
syntax formally, by describing rules that allow us to work off of such
definitions and, finally, by justifying the new rules by means of a
translation of the kind indicated above.

\begin{definition}
A {\em pattern-based definition} is a finite collection of clauses of
the form
\[\forall \vec{x}.(\nabla \vec{z}. p\ \vec{t}) \triangleq
B\ p\ \vec{x}\] where $\vec{t}$ is a sequence of terms that do not
have occurrences of nominal constants in them, $p$ is a constant such
that $p\ \vec{t}$ is of type $o$ and $B$ is a term devoid of
occurrences of $p$, $\vec{x}$ and nominal constants and such that $B\
p\ \vec{t}$ is of type $o$. Further, we expect such a collection of
clauses to satisfy a stratification condition: there must exist an
assignment of levels to predicate symbols such that for any clause
$\forall \vec{x}.(\nabla \vec{z}. p\ \vec{t}) \triangleq B\ p\
\vec{x}$ occurring in the set, assuming $p$ has arity $n$, it is the
case that $\lvl(B\ (\lambda \vec{x}.\top)\ \vec{x}) <
\lvl(p)$. Notice that we allow the collection to contain more than one
clause for any given predicate symbol.
\end{definition}

\begin{figure}[t]
\begin{center}
$\infer[\defR^p]
      {\Sigma : \Gamma \lra p\; \vec{s}}
      {\Sigma : \Gamma \lra (B\; p\; \vec{x})[\theta]}$\\[5pt]
for any clause $\forall \vec{x}.(\nabla \vec{z}. p\ \vec{t}) \triangleq
B\ p\ \vec{x}$ in $\cal D$ and any $\theta$ such that
$range(\theta)\cap \Sigma = \emptyset$ and $(\lambda \vec{z} . p\ \vec{t})[\theta]
                        \unrhd p\ \vec{s}$ holds\\[20pt]
$\infer[\defL^p]
      {\Sigma : \Gamma, p\; \vec{s} \lra C}
      {\left\{
        \begin{tabular}{l|l}
         $\Sigma\theta : \Gamma\cas{\theta}, (B\; p\;
         \vec{x})\cas{\theta} \lra C\cas{\theta}$ &
           $\forall \vec{x}.(\nabla \vec{z}. p\ \vec{t}) \triangleq
                              B\ p\ \vec{x} \in {\cal D}$ and \\
        &
           $\theta$ is a solution to $((\lambda \vec{z} . p\ \vec{t}) \unrhd p\
                        \vec{s})$
       \end{tabular}
              \right\}
      }
$
\end{center}
\caption{Introduction rules for a pattern-based definition $\cal D$}
\label{fig:patterndefrules}
\end{figure}

The logical rules for treating pattern-based definitions are presented
in Figure~\ref{fig:patterndefrules}. These rules encode the
idea of matching an instance of a predicate with the head of a
particular clause and then replacing the predicate with the
corresponding clause body. The kind of matching involved is made
precise through the construction of a nominal abstraction after
replacing the $\nabla$ quantifiers in the head of the clause by
abstractions. The right rule embodies the fact that it is enough if
an instance of any one clause can be used in this way to yield a
successful proof. In this rule, the substitution $\theta$ that results
from the matching must be applied in a nominal capture avoiding way to
the body. However, since $B$ does not contain nominal constants,
the ordinary application of the substitution also suffices.
To accord with the treatment in the right rule, the left rule
must consider all possible ways in which an instance of an atomic
assumption  $p\ \vec{s}$ can be matched by a clause and must show that
a proof can be constructed in each such case.

The soundness of these rules is the content of the following theorem
whose proof also makes explicit the intended interpretation of the
pattern-based form of definitions.

\begin{theorem}
The pattern-based form of definitions and the associated proof rules
do not add any new power to the logic. In particular, the $\defL^p$
and $\defR^p$ rules are admissible under the intended interpretation
via translation of the pattern-based form of
definitions.
\end{theorem}
\begin{proof}
Let $p$ be a predicate whose clauses in the definition being
considered are given by the following set of clauses.
\begin{equation*}
\{\forall \vec{x}_i.~ (\nabla \vec{z}_i. p\ \vec{t}_i) \triangleq
B_i\ p\ \vec{x}_i\}_{i\in 1..n}
\end{equation*}
Let $p'$ be a new constant symbol with the same argument types as
$p$. Then the intended interpretation of the definition of $p$ in a
setting that does not allow the use of patterns in the head and that
limits the number of clauses defining a predicate to one is given by
the clause
\begin{equation*}
\forall \vec{y} . p\ \vec{y} \triangleq \bigvee_{i\in 1..n} \exists \vec{x}_i
. ((\lambda \vec{z}_i . p'\ \vec{t}_i) \unrhd p'\ \vec{y}) \land B_i\
p\ \vec{x}_i
\end{equation*}
in which the variables $\vec{y}$ are chosen such that they do not
appear in the terms $\vec{t}_i$ for $1 \leq i \leq n$. Note also that we are using the term
constructor $p'$ here so as to be able to match the entire
head of a clause at once, thus ensuring that the $\nabla$-bound
variables in the head are assigned a consistent value for all
arguments of the predicate.

Based on this translation, we can replace
an instance of $\defR^p$,
\begin{equation*}
\infer[\defR^p]
      {\Gamma \lra p\; \vec{s}}
      {\Gamma \lra (B_i\; p\; \vec{x}_i)[\theta]}
\end{equation*}
with the following sequence of rules, where a double inference line
indicates that a rule is used multiple times.
\begin{equation*}
\infer[\defR]{\Gamma \lra p'\; \vec{t}}
{\infer=[\lorR]
 {\Gamma \lra \bigvee_{i\in 1..n} \exists \vec{x}_i
  . ((\lambda \vec{z}_i . p'\ \vec{t}_i) \unrhd p'\ \vec{s}) \land
     B_i\ p\ \vec{x}_i
 }
 {\infer=[\existsR]
  {\Gamma \lra \exists \vec{x}_i
   . ((\lambda \vec{z}_i . p'\ \vec{t}_i) \unrhd p'\ \vec{s}) \land
     B_i\ p\ \vec{x}_i
  }
  {\infer[\landR]
   {\Gamma \lra ((\lambda \vec{z}_i . p'\ \vec{t}_i)[\theta] \unrhd
     p'\ \vec{s}) \land (B_i\ p\ \vec{x}_i)[\theta]
   }
   {\infer[\unrhdR]{\Gamma \lra (\lambda \vec{z}_i . p'\
       \vec{t}_i)[\theta] \unrhd p'\ \vec{s}}{}
    &
    \Gamma \lra (B_i\; p\; \vec{x}_i)[\theta]
   }
  }
 }
}
\end{equation*}
Note that we have made use of the fact that $\theta$ instantiates only
the variables $x_i$ and thus has no effect on $\vec{s}$. Further, the
side condition associated with the $\defR^p$ rule ensures that the
$\unrhdR$ rule that appears as a left leaf in this derivation is well
applied.

Similarly, we can replace an instance of $\defL^p$,
\begin{equation*}
\infer[\defL^p]
      {\Sigma : \Gamma, p\; \vec{s} \lra C}
      {\left\{
         \Sigma\theta : \Gamma\cas{\theta}, (B_i\; p\;
         \vec{x}_i)\cas{\theta} \lra C\cas{\theta}\ |\
           \hbox{$\theta$ is a solution to $((\lambda \vec{z} . p\ \vec{t}_i) \unrhd p\
                        \vec{s})$}
              \right\}_{i\in 1..n}
      }
\end{equation*}
with the following sequence of rules
\begin{equation*}
\hspace{-1.8cm}
\infer[\defL]{\Gamma, p\; \vec{s} \lra C}
{
 \infer=[\lorL]
 {\Gamma, \bigvee_{i\in 1..n} \exists \vec{x}_i
  . ((\lambda \vec{z}_i . p'\ \vec{t}_i) \unrhd p'\ \vec{s}) \land
     B_i\ p\ \vec{x}_i
  \lra C
 }
 {\hspace{2.2cm}\left\{\raisebox{-6ex}{
  \infer=[\existsL]
   {\Gamma, \exists \vec{x}_i
     . ((\lambda \vec{z}_i . p'\ \vec{t}_i) \unrhd p'\ \vec{s}) \land
     B_i\ p\ \vec{x}_i
    \lra C
   }
   {\infer[\landL^*]
    {\Gamma, ((\lambda \vec{z}_i . p'\ \vec{t}_i) \unrhd p'\ \vec{s}) \land
     B_i\ p\ \vec{x}_i
     \lra C
    }
    {\infer[\unrhdL]
     {\Gamma, (\lambda \vec{z}_i . p'\ \vec{t}_i) \unrhd p'\ \vec{s},
       B_i\ p\ \vec{x}_i
       \lra C
     }
     {\left\{\hbox{
       $\Gamma\cas{\theta}, (B_i\; p\; \vec{x}_i)\cas{\theta}
         \lra C\cas{\theta}\ |\ \theta$ is a solution to
	   $((\lambda \vec{z} . p'\ \vec{t}_i) \unrhd p'\ \vec{s})$
      }
      \right\}
     }
    }
   }
 }\right\}_{i \in 1..n}\hspace{3cm}
 }
}
\end{equation*}
Here $\landL^*$ is an application of $\cL$ followed by $\landL_1$ and
$\landL_2$ on the contracted formula. It is easy to see that the
solutions to $(\lambda \vec{z}.p\;\vec{t}_i)
\unrhd p\;\vec{s}$ and $(\lambda \vec{z}.p'\;\vec{t}_i)
\unrhd p'\;\vec{s}$ are identical and hence the leaf sequents in this
partial derivation are exactly the same as the upper sequents of the
instance of the $\defL^p$ rule being considered.
\end{proof}

A weak form of a converse to the above theorem also holds. Suppose
that the predicate $p$ is given by the following clauses
\begin{equation*}
\{\forall \vec{x}_i.~ (\nabla \vec{z}_i. p\ \vec{t}_i) \triangleq
B_i\ p\ \vec{x}_i\}_{i\in 1..n}
\end{equation*}
in a setting that uses pattern-based definitions and that has the
$\defL^p$ and $\defR^p$ but not the $\defL$ and $\defR$ rules. In such
a logic, it is easy to see that the following is provable:
\begin{equation*}
\forall \vec{y} . \left[p\ \vec{y} \equiv \bigvee_{i\in 1..n} \exists
  \vec{x}_i . ((\lambda \vec{z}_i . p'\ \vec{t}_i) \unrhd p'\ \vec{y})
  \land B_i\ p\ \vec{x}_i\right]
\end{equation*}
Where $B \equiv C$ denotes $(B \supset C) \land (C \supset B)$. Thus,
in the presence of \cut, the $\defL$ and $\defR$ rules can be treated
as derived ones relative to the translation interpretation of
pattern-based definitions.

We would like also to allow patterns to be used in the heads of
clauses when writing definitions that are intended to pick out the
least and greatest fixed points, respectively. Towards this end we
admit in a definition also clauses of the form $\forall
\vec{x}.(\nabla \vec{z}. p\ \vec{t}) \mueq B\ p\ \vec{x}$ and $\forall
\vec{x}.(\nabla \vec{z}. p\ \vec{t}) \nueq B\ p\ \vec{x}$ with the
earlier provisos on the form of $B$ and $\vec{t}$ and the types of $B$ and
$p$ and with the additional requirement that all the clauses for any
given predicate are un-annotated or annotated uniformly with either $\mu$ or
$\nu$. Further, a definition must satisfy stratification conditions as
before. In reasoning about the least or greatest fixed point forms of
definitions, we may use the translation into the earlier, non-pattern
form together with the rules $\IL$ and $\CIR$. It is possible to
formulate an induction rule that works directly from pattern-based
definitions using the idea that to show $S$ to be an induction
invariant for the predicate $p$, one must show that every clause of
$p$ preserves $S$. A rule that is based on this intuition is presented
in Figure~\ref{fig:pattern-induction-rule}. The soundness of this rule
is shown in the following theorem.

\begin{figure}[t]
\begin{equation*}
\infer[\IL^p]
{\Sigma : \Gamma, p\ \vec{s} \lra C}
{\left\{\vec{x}_i : B_i\ S\ \vec{x}_i \lra \nabla \vec{z}_i.S\
  \vec{t}_i\right\}_{i\in 1..n} \quad
  \Sigma : \Gamma, S\ \vec{s} \lra C}
\end{equation*}
\begin{center}
assuming $p$ is defined by the set of clauses $\{\forall \vec{x}_i.
(\nabla \vec{z}_i. p\ \vec{t}_i) \mueq B_i\ p\ \vec{x}_i\}_{i\in 1..n}$
\end{center}
\caption{Induction rule for pattern-based definitions}
\label{fig:pattern-induction-rule}
\end{figure}

\begin{theorem}
The $\IL^p$ rule is admissible under the intended translation of
pattern-based definitions.
\end{theorem}
\begin{proof}
Let the clauses for $p$ in the pattern-based definition be given by
the set
\[\{\forall \vec{x}_i.
(\nabla \vec{z}_i. p\ \vec{t}_i) \mueq B_i\ p\ \vec{x}_i\}_{i\in
  1..n}\]
in which case the translated form of the definition for $p$ would be
\begin{equation*}
\forall \vec{y} . p\ \vec{y} \mueq \bigvee_{i\in 1..n} \exists \vec{x}_i
. ((\lambda \vec{z}_i . p'\ \vec{t}_i) \unrhd p'\ \vec{y}) \land B_i\
p\ \vec{x}_i.
\end{equation*}
In this context, the rightmost upper sequents of the $\IL^p$ and the $\IL$
rules that are needed to derive a sequent of the form $\Sigma :
\Gamma, p\ \vec{s} \lra C$ are identical. Thus, to show
that $\IL^p$ rule is admissible, it suffices to show that the left
upper sequent in the $\IL$ rule can be derived in the original
calculus from all but the rightmost upper sequent in an $\IL^p$
rule. Towards this end, we observe that we can construct the following
derivation:
\begin{equation*}
\small
\hspace{-2.3cm}
\infer=[\lorL]
{\vec{y} : \bigvee_{i\in 1..n} \exists \vec{x}_i
  . ((\lambda \vec{z}_i . p'\ \vec{t}_i) \unrhd p'\ \vec{y}) \land
  B_i\ S\ \vec{x}_i
  \lra S\ \vec{y}
}
{\hspace{2.8cm}\left\{\raisebox{-6ex}{
    \infer=[\existsL]
    {\vec{y} : \exists \vec{x}_i
      . ((\lambda \vec{z}_i . p'\ \vec{t}_i) \unrhd p'\ \vec{y}) \land
      B_i\ S\ \vec{x}_i \lra S\ \vec{y}
    }
    {\infer[\landL^*]
      {\vec{y}, \vec{x}_i :
        ((\lambda \vec{z}_i . p'\ \vec{t}_i) \unrhd p'\ \vec{y}) \land
        B_i\ p\ \vec{x}_i \lra S\ \vec{y}
      }
      {\infer[\unrhdL]
        {\vec{y}, \vec{x}_i :
          (\lambda \vec{z}_i . p'\ \vec{t}_i) \unrhd p'\ \vec{y},
          B_i\ S\ \vec{x}_i
          \lra S\ \vec{y}
        }
        {\left\{\hbox{
            $(\vec{y}, \vec{x}_i)\theta : (B_i\; p\; \vec{x}_i)\cas{\theta}
            \lra (S\ \vec{y})\cas{\theta}\ |\ \theta$ is a solution to
            $((\lambda \vec{z} . p'\ \vec{t}_i) \unrhd p'\ \vec{y})$
          }
          \right\}
        }
      }
    }
  }\right\}_{i \in 1..n}\hspace{3cm}
}
\end{equation*}
Since the variables $\vec{y}$ are distinct and do not occur in
$\vec{t}_i$, the solutions to $(\lambda \vec{z} . p'\ \vec{t}_i)
\unrhd p'\ \vec{y}$ have a simple form. In particular, let $\vec{t}'_i$
be the result of replacing in $\vec{t}_i$ the variables $\vec{z}$ with
distinct nominal constants. Then $\vec{y} = \vec{t}'_i$ will be a most
general solution to the nominal abstraction. Thus the upper sequents
of the invariant derivation above will be
\begin{equation*}
\vec{x}_i : B_i\ p\ \vec{x}_i \lra S\ \vec{t}'_i
\end{equation*}
which are derivable if and only if the sequents
\begin{equation*}
\vec{x}_i : B_i\ p\ \vec{x}_i \lra \nabla\vec{z}_i. S\ \vec{t}_i
\end{equation*}
are derivable.
\end{proof}

We do not introduce a co-induction rule for pattern-based
definitions largely because it seems that there are few interesting
co-inductive definitions that require patterns and multiple clauses.

\section{Examples}
\label{sec:examples}

We now provide some examples to illuminate the properties of nominal
abstraction and its usefulness in both specification and reasoning
tasks; while \logic has many more features, their characteristics and
applications have been exposed in other work (\eg, see
\cite{mcdowell02tocl,momigliano03types,tiu04phd,tiu.tocl}). In the
examples that are shown, use will be made of the pattern-based form of
definitions described in Section~\ref{sec:pattern-form}. We will also
use the convention that tokens given by capital letters denote
variables that are implicitly universally quantified over the entire
clause.

\subsection{Properties of $\nabla$ and Freshness}
\label{sec:nabla-freshness}

We can use nominal abstraction to gain a better insight into the
behavior of the $\nabla$ quantifier. Towards this end, let the {\sl fresh}
predicate be defined by the following clause.
\begin{equation*}
(\nabla x.\fresh x E) \triangleq \top
\end{equation*}
We have elided the type of {\sl fresh} here; it will have to be
defined at each type that it is needed in the examples we consider
below. Alternatively, we can ``inline'' the definition by using nominal
abstraction directly, \ie, by replacing occurrences of of $\fresh
{t_1} {t_2}$ with $\exists E. (\lambda x. \tup{x, E} \unrhd \tup{t_1,
  t_2})$ for a suitably typed pairing construct $\tup{\cdot,\cdot}$.

Now let $B$ be a formula whose free variables are among $z, x_1,
\ldots, x_n$, and let $\vec{x} = x_1 :: \ldots :: x_n :: nil$ where
$::$ and $nil$ are constructors in the logic.\footnote{We are, once
  again, finessing typing issues here in that the $x_i$ variables may
  not all be of the same type. However, this problem can be solved by
  surrounding each of them with a constructor that yields a term with
  a uniform type.} Then the following
formulas logically imply one another in \logic.
\[
\nabla z. B \qquad\quad
\exists z. (\fresh z \vec{x} \land B) \qquad\quad
\forall z. (\fresh z \vec{x} \supset B)
\]
Note that the type of $z$ allows it to be an arbitrary term in the
last two formulas, but its occurrence as the first argument of {\sl
  fresh} will restrict it to being a nominal constant (even when
$\vec{x} = nil$).

In the original presentation of the $\nabla$ quantifier
\cite{miller03lics}, it was shown that one can move a $\nabla$
quantifier inwards over universal and existential quantifiers by using
raising to encode an explicit dependency. To illustrate this, let $B$
be a formula with two variables abstracted out, and let $C \equiv D$
be shorthand for $(C \supset D) \land (D\supset C)$. The the following
formulas are provable in the logic.
\begin{align*}
\nabla z. \forall x. (B\ z\ x) &\equiv \forall h. \nabla z.
(B\ z\ (h\ z)) &
\nabla z. \exists x. (B\ z\ x) &\equiv \exists h. \nabla z.
(B\ z\ (h\ z))
\end{align*}
In order to move a $\nabla$ quantifier outwards over universal and
existential quantifiers, one would need a way to make non-dependency
(\ie, freshness) explicit. This is now possible using nominal
abstraction as shown by the following equivalences.
\begin{align*}
\forall x. \nabla z. (B\ z\ x) &\equiv
\nabla z. \forall x. (\fresh z x \supset B\ z\ x)
&
\exists x. \nabla z. (B\ z\ x) &\equiv
\nabla z. \exists x. (\fresh z x \land B\ z\ x)
\end{align*}
Finally, we note that the two sets of equivalences for moving the
$\nabla$ quantifier interact nicely. Specifically, starting with a
formula like $\nabla z. \forall x. (B\ z\ x)$ we can push the $\nabla$
quantifier inwards and then outwards to obtain $\nabla z. \forall h.
(\fresh z (h\ z) \supset B\ z\ (h\ z))$. Here $\fresh z (h\ z)$ will
only be satisfied if $h$ projects away its first argument, as
expected.

\subsection{Polymorphic Type Generalization}

In addition to reasoning, nominal abstraction can also be useful in
providing declarative specifications of computations. We consider the
context of a type inference algorithm that is also discussed in
\cite{cheney08toplas} to illustrate such an application. In this
setting, we might need a predicate {\sl spec} that relates a
polymorphic type $\sigma$, a list of distinct variables
list of distinct variables $\vec{\alpha}$ (represented by nominal
constants) and a monomorphic type $\tau$ just in the case that $\sigma =
\forall\vec{\alpha}.\tau$. Using nominal abstraction, we can define
this predicate as follows.
\begin{align*}
&\spec {(\monoTy T)} {nil} T \mueq \top \\
(\nabla x. &\spec {(\polyTy P)} {(x::L)} {(T\ x)}) \mueq
\nabla x. \spec {(P\ x)} L {(T\ x)}.
\end{align*}
Note that we use $\nabla$ in the head of the second clause to
associate the variable $x$ at the head of the list $L$ with its
occurrences in the type $(T\ x)$. We then use $\nabla$ in the body of
this clause to allow for the recursive use of {\sl spec}.

\subsection{Arbitrarily Cascading Substitutions}

Many reducibility arguments, such as Tait's proof of normalization for
the simply typed $\lambda$-calculus \cite{tait67jsl}, are based on
judgments over closed terms. During reasoning, however, one has often
to work with open terms. To accommodate this requirement, the closed
term judgment is extended to open terms by considering all possible
closed instantiations of the open terms. When reasoning with \logic,
open terms are denoted by terms with nominal constants representing
free variables. The general form of an open term is thus $M\; c_1\;
\cdots\; c_n$, and we want to consider all possible instantiations
$M\; V_1\; \cdots\; V_n$ where the $V_i$ are closed terms. This type
of arbitrary cascading substitutions is difficult to realize in
reasoning systems where variables are given a simple type since $M$
would have an arbitrary number of abstractions but the type of $M$
would {\em a priori} fix that number of abstractions.

We can define arbitrary cascading substitutions in \logic using
nominal abstraction. In particular, we can define a predicate which
holds on a list of pairs $\tup{c_i,V_i}$, a term with the form $M\; c_1\;
\cdots\; c_n$ and a term of the form $M\; V_1\; \cdots\; V_n$. The
idea is to iterate over the list of pairs and for each pair $\tup{c,V}$
use nominal abstraction to abstract $c$ out of the first term and then
substitute $V$ before continuing. The following definition of the
predicate {\sl subst} is based on this idea.
\begin{align*}
& \subst {nil} T T \mueq \top \\
(\nabla x. &\subst {(\tup{x,V}::L)} {(T\; x)} S) \mueq
\subst L {(T\; V)} S
\end{align*}

Given the definition of {\sl subst} one may then show that arbitrary
cascading substitutions have many of the same properties as normal
higher-order substitutions. For instance, in the domain of the untyped
$\lambda$-calculus, we can show that {\sl subst} acts compositionally via
the following lemmas.
\begin{align*}
&\forall \ell, t, r, s.~
\subst \ell {(\app t r)} s \supset
\exists u, v. (s = \app u v \land \subst \ell t u \land \subst \ell r v)
\\
&\forall \ell, t, r.~
\subst \ell {(\uabs t)} r \supset
\exists s. (r = \uabs s \land \nabla z. \subst \ell {(t\; z)} (s\; z))
\end{align*}
Both of these lemmas have straightforward proofs by induction on {\sl
  subst}.

We use this technique for describing arbitrary cascading substitutions again in
Section~\ref{sec:girards-strong-norm} to formalize Girard's strong
normalization argument for the simply-typed $\lambda$-calculus.



\chapter{Some Properties of the Meta-logic}
\label{ch:meta-theory}

In this chapter we study some of the meta-theory of \logic. There are
two parts to our discussion. In the first part of the chapter, we
prove various properties of the logic which show that the logic is
well-designed and which are also useful when working within the logic.
Most significantly, we prove the cut-elimination property for \logic
and then use this to establish the consistency of the logic. In the
second part of the chapter we look at the question of how we can
formally relate an object system to a potential encoding of it in
\logic. The naturalness of such a relationship is a strong
recommendation for the meta-logic: it is ultimately this
correspondence that allows us to use \logic in establishing properties
of an object system. Showing this type of relationship depends
crucially on the earlier cut-elimination result which further
justifies the emphasis we place on it.

\section{Consistency of the Meta-logic}
\label{sec:meta-theory}

The logic \logic, whose proof rules consist of the ones
Figures~\ref{fig:core-rules}, \ref{fig:na-rules}, \ref{fig:defrules},
and \ref{fig:indandcoind}, combines and extends the features in
several logics such as $\FOLDN$ \cite{mcdowell00tcs}, \foldnb
\cite{miller05tocl}, $LG^\omega$ \cite{tiu08lgext} and Linc$^-$
\cite{tiu.momigliano}. The relationship to Linc$^-$ is of special
interest to us below: \logic is a conservative extension to this logic
that is obtained by adding a treatment of the $\nabla$ quantifier and
the associated nominal constants and by generalizing the proof rules
pertaining to equality to ones dealing with nominal abstraction. This
correspondence will allow the proof of the critical meta-theoretic
property of cut-elimination for Linc$^-$ to be lifted to \logic.

We shall actually establish three main properties of \logic in this
section.  First, we shall show that the provability of a sequent is
unaffected by the application of permutations of nominal constants to
formulas in the sequent.  This property consolidates our understanding
that nominal constants are quantified implicitly at the formula level;
such quantification also renders irrelevant the particular names chosen
for such constants. Second, we show that the application of substitution
in a nominal capture-avoiding way preserves provability; by contrast,
ordinary application of substitution does not have this property.
Finally, we show that the $\cut$ rule can be
dispensed with from the logic without changing the set of provable
sequents. This implies that the left and right rules of the logic are
balanced and moreover, that the logic is consistent. This is the main
result of this section and its proof uses the earlier two results
together with the argument for cut-elimination for Linc$^-$.

Several of our arguments will be based on induction on the heights
of proofs. This measure is defined formally below. Notice that
the height of a proof can be an infinite ordinal because the $\unrhdL$
rule can have an infinite number of premises. Thus, we will be using
a transfinite form of induction.

\begin{definition}
\label{def:ht}
The {\em height} of a derivation $\Pi$, denoted by $\htf(\Pi)$, is $1$
if $\Pi$ has no premise derivations and is the least upper bound of
$\{\htf(\Pi_i)+1\}_{i\in\mathcal{I}}$ if $\Pi$ has the premise
derivations $\{\Pi_i\}_{i\in\mathcal{I}}$ where $\mathcal{I}$ is some
index set.
\end{definition}

Many proof systems, such as Linc$^-$, include a weakening rule that
allows formulas to be dropped (reading proofs bottom-up) from the
left-hand sides of sequents.
While \logic does not include such a rule directly, its effect is
captured in a strong sense as we show in the lemma below. Two proofs
are to be understood here and elsewhere as having the same structure
if they are isomorphic as trees, if the same rules appear at
corresponding places within them and if these rules pertain to
formulas that can be obtained one from the other via a renaming of
eigenvariables and nominal constants.

\begin{lemma}
\label{lem:proof-weak}
Let $\Pi$ be a proof of $\Sigma : \Gamma \lra B$ and let $\Delta$ be
a multiset of formulas whose eigenvariables are contained in $\Sigma$.
Then there exists a proof of $\Sigma : \Delta, \Gamma \lra B$ which
has the same structure as $\Pi$. In particular $\htf(\Pi) =
\htf(\Pi')$ and $\Pi$ and $\Pi'$ end with the same rule application.
\end{lemma}
\begin{proof}
The lemma can be proved by an easy induction on $\htf(\Pi)$. We omit
the details.
\end{proof}

The following lemma shows a strong form of the preservation of
provability under permutations of nominal constants appearing in
formulas, the first of our mentioned results.

\begin{lemma}
\label{lem:proof-perm}
Let $\Pi$ be a proof of $\Sigma : B_1, \ldots, B_n \lra B_0$ and let
$B_i \approx B_i'$ for $i \in \{0,1,\ldots,n\}$. Then there exists a
proof $\Pi'$ of $\Sigma : B_1', \ldots, B_n' \lra B_0'$ which has
the same structure as $\Pi$. In particular $\htf(\Pi) = \htf(\Pi')$
and $\Pi$ and $\Pi'$ end with the same rule application.
\end{lemma}
\begin{proof}
The proof is by induction on $\htf(\Pi)$ and proceeds specifically by
considering the last rule used in $\Pi$. When this is a left rule, we
shall assume without loss of generality that it operates on $B_n$.

The argument is easy to provide when the last rule in $\Pi$ is one of $\botL$
or $\topR$. If this rule is an $id$, \ie, if $\Pi$ is of the form
\begin{equation*}
\infer[id]{\Sigma : B_1, \ldots, B_n \lra B_0}
             {B_j \approx B_0}
\end{equation*}
then, since $\approx$ is an equivalence relation, it must be the case
that $B_j' \approx B_0'$. Thus, we can let
$\Pi'$ be the derivation
\begin{equation*}
\infer[id]{\Sigma : B_1', \ldots, B_n' \lra B_0'}
          {B_j' \approx B_0'}
\end{equation*}
If the last rule is a $\unrhdL$ applied to a nominal abstraction $s
\unrhd t$ that has no solutions, then, by Lemma~\ref{lem:na-approx},
the sequent $\Sigma : B_1',\ldots B_n' \lra
B_0'$ also has a nominal abstraction with no solutions. Thus, $\Pi'$
can be a derivation consisting of the single rule
$\unrhdL$. Lemma~\ref{lem:na-approx} similarly provides the key
observation when the last rule in $\Pi$ is an $\unrhdR$.

All the remaining cases correspond to derivations of height greater
than 1. We shall show that the last rule in $\Pi$ in all these cases could
also have $\Sigma : B_1', \ldots, B_n' \lra B_0'$ as a conclusion
with the premises in this application of the rule being related via
permutations in the way required by the lemma to the premises of the
rule application in $\Pi$. The lemma then follows from the induction
hypothesis.

In the case when the last rule in $\Pi$ pertains to a binary
connective---\ie, when the rule is one of $\lorL$, $\lorR$, $\landL$,
$\landR$, $\supsetL$ or $\supsetR$---the desired conclusion follows
naturally from the observation that permutations distribute over the
connective. The proof can be similarly completed when a
$\existsL$, $\existsR$, $\forallL$ or $\forallR$ rule ends the
derivation, once we have noted that the application of permutations can
be moved under the $\exists$ and $\forall$ quantifiers. For the
$\cut$ and $\cL$ rules, we have to show that permutations
can be extended to include the newly introduced formula in the upper
sequent(s). This is easy: for the $\cut$ rule we use the identity
permutation and for $\cL$ we replicate the permutation used to obtain
$B_n'$ from $B_n$.

The two remaining rules from the core logic are $\nablaL$ and
$\nablaR$. The argument in these cases are similar and we consider
only the later in detail. In this case, the last rule in $\Pi$ is of
the form
\begin{equation*}
\infer[\nablaR]
      {\Sigma : B_1, \ldots, B_n \lra \nabla x. C}
      {\Sigma : B_1, \ldots, B_n \lra C[a/x]}
\end{equation*}
where $a
\notin \supp(C)$. Obviously, $B_0' = \nabla x.C'$ for some $C'$ such
that $C \approx C'$. Let $d$ be a nominal constant such that $d
\notin
\supp(C)$ and $d \notin \supp(C')$. Such a constant must exist since
both sets are finite. Then $C[a/x] \approx C[d/x] \approx C'[d/x]$.
Thus the following
\begin{equation*}
\infer[\nablaR]
      {\Sigma : B_1', \ldots, B_n' \lra \nabla x. C'}
      {\Sigma : B_1', \ldots, B_n' \lra C'[d/x]}
\end{equation*}
is also an instance of the $\nablaR$ rule and its upper sequent has
the form desired.

The only case that remains to be treated when the last rule applies to
a nominal abstraction is that of $\unrhdL$ that has at least one
upper sequent. In this case the rule has the structure
\begin{equation*}
\infer[\unrhdL]
      {\Sigma : B_1, \ldots, s \unrhd t \lra B_0}
      {\left\{ \Sigma\theta : B_1\cas{\theta}, \ldots,
                 B_{n-1}\cas{\theta} \lra B_0\cas{\theta} \;|\;
                 \theta\ \mbox{is a solution to}\ s \unrhd t
       \right\}}
\end{equation*}
Here we know that $B_n'$ is a nominal abstraction $s'\unrhd t'$ that,
by Lemma~\ref{lem:na-approx}, has the same solutions as $s \unrhd
t$. Further, by Lemma~\ref{lem:approx-cas}, $B_i\cas{\theta} \approx
B_i'\cas{\theta}$ for any substitution $\theta$. Thus
\begin{equation*}
\infer[\unrhdL]
      {\Sigma : B_1', \ldots, s'\unrhd t' \lra B_0'}
      {\left\{ \Sigma\theta : B_1'\cas{\theta}, \ldots,
                B_{n-1}'\cas{\theta} \lra
                B_0'\cas{\theta} \;|\;
          \theta\ \mbox{is a solution to}\ s' \unrhd t'
       \right\}}
\end{equation*}
is also an instance of the $\unrhdL$ rule and its upper sequents have
the required property.

The arguments for the rules $\defL$ and $\defR$ are similar and we
therefore only consider the case for the former rule in detail. Here,
$B_n$ must be of the form $p\; \vec{t}$ where $p$ is a predicate
symbol and the upper sequent must be identical to the lower one except
for the fact that $B_n$ is replaced by a formula of the form $B\ p\;
\vec{t}$ where $B$ contains no nominal constants. Further, $B_n'$ is
of the form $p\; \vec{s}$ where $p\;\vec{t} \approx p\;\vec{s}$.
From this it follows
that $B\ p\; \vec{t} \approx B\ p\; \vec{s}$ and hence that $\Sigma:
B_1',\ldots,B_n' \lra B_0'$ can be the lower sequent of a rule whose
upper sequent is related in the desired way via permutations to the
upper sequent of the last rule in $\Pi$.

The only remaining rules to consider are $\IL$ and $\CIR$. Once again,
the arguments in these cases are similar and we therefore consider
only the case for $\IL$ in detail. Here, $\Pi$ ends with a rule
of the form
\begin{equation*}
\infer[\IL]
      {\Sigma : B_1, \ldots, p\; \vec{t} \lra B_0}
      {\vec{x} : B\; S\; \vec{x} \lra S\; \vec{x} \qquad
       \Sigma : B_1, \ldots, S\; \vec{t} \lra B_0}
\end{equation*}
where $p$ is a predicate symbol defined by a clause of the form
$\forall \vec{x}.~p\;\vec{x} \mueq B\ p\;\vec{x}$
and $S$
contains no nominal constants. Now, $B_n'$ must be of the form
$p\;\vec{r}$ where $p\;\vec{t}
\approx p\;\vec{r}$. Noting the proviso on $S$, it follows that $S\;
\vec{t} \approx S\;\vec{r}$. But then the following
\begin{equation*}
\infer[\IL]
      {\Sigma : B_1', \ldots, p\; \vec{r} \lra B_0'}
      {\vec{x} : B\; S\; \vec{x} \lra S\; \vec{x} \qquad
       \Sigma : B_1', \ldots, S\; \vec{r} \lra B_0'}
\end{equation*}
is also an instance of the $\IL$ rule and its upper sequents are
related in the manner needed to those of the $\IL$ rule used in $\Pi$.
\end{proof}

Several rules in \logic require the selection of new eigenvariables
and nominal constants. Lemma~\ref{lem:proof-perm} shows that we obtain
what is essentially the same proof regardless of how we choose nominal
constants in such rules so long as the local non-occurrence conditions
are satisfied. A similar observation with regard to the choice of
eigenvariables is also easily verified. We shall therefore identify
below proofs that differ only in the choices of eigenvariables and
nominal constants.

We now turn to the second of our desired results, the preservation of
provability under substitutions.

\begin{lemma}
\label{lem:proof-subst}
Let $\Pi$ be a proof of $\Sigma : \Gamma \lra C$ and let $\theta$ be
a substitution. Then there is a proof $\Pi'$ of $\Sigma\theta :
\Gamma\cas{\theta} \lra C\cas{\theta}$ such that $\htf(\Pi') \leq
\htf(\Pi)$.
\end{lemma}

\begin{proof}
We show how to transform the proof $\Pi$ into a proof $\Pi'$ for the
modified sequent. The transformation is by recursion on $\htf(\Pi)$,
the critical part of it being a consideration of the last rule in
$\Pi$. The transformation is, in fact, straightforward in all cases
other that when this rule is $\unrhdL$, $\forallR$, $\existsL$,
$\existsR$,
$\forallL$, $\IL$ and $\CIR$. In these cases, we simply apply the
substitution in a nominal capture avoiding way to the lower and any
possible upper sequents of the rule. It is easy to see that the resulting
structure is still an instance of the same rule and its upper sequents
are guaranteed to have proofs (of suitable heights) by induction.

Suppose that the last rule in $\Pi$ is an $\unrhdL$, \ie, it is of the form
\begin{equation*}
\infer[\unrhdL]{\Sigma : \Gamma, s\unrhd t \lra C}
      {\left\{\Sigma\rho : \Gamma\cas{\rho} \lra C\cas{\rho} \;|\;
        \rho\ \mbox{is a solution to}\ s \unrhd t\right\}}
\end{equation*}
Then the following
\begin{equation*}
\infer[\unrhdL]{\Sigma\theta : \Gamma\cas{\theta},
                (s\unrhd t)\cas{\theta} \lra C\cas{\theta}}
      {\left\{
        \Sigma(\theta\bullet\rho') : \Gamma\cas{\theta \bullet\rho'}
                 \lra C\cas{\theta \bullet \rho'}
                 \;|\; \rho'\ \mbox{is a solution to}\ (s \unrhd
                t)\cas{\theta}
      \right\}}
\end{equation*}
is also an $\unrhdL$ rule. Noting that if $\rho'$ is a solution to
$(s\unrhd t)\cas{\theta}$, then $\theta\bullet \rho'$ is a solution to
$s\unrhd t$, we see that the upper sequents of this rule are contained
in the upper sequents of the rule in $\Pi$. It follows that we can
construct a proof of the lower sequent whose height is less than or
equal to that of $\Pi$.

The argument is similar in the cases when the last rule in $\Pi$ is a
$\forallR$ or a $\existsL$ so we consider only the former in
detail. In this case the rule has the form
\begin{equation*}
\infer[\forallR]
      {\Sigma : \Gamma \lra \forall x.B}
      {\Sigma, h : \Gamma \lra B[h\;\vec{c}/x]}
\end{equation*}
where $\{\vec{c}\} = \supp(\forall x.B)$. Let $\{\vec{a}\} =
\supp((\forall x. B)\cas{\theta})$. Further, let $h'$ be a new variable
name. We assume without loss of generality that neither $h$ nor $h'$
appear in the domain or range of $\theta$. Letting $\rho = \theta \cup
\{\lambda\vec{c}.h'\;\vec{a}/h\}$, consider the structure
\begin{equation*}
\infer[]
      {\Sigma\theta : \Gamma\cas{\theta} \lra (\forall x.B)\cas{\theta}}
      {(\Sigma, h)\rho :
               \Gamma\cas{\rho} \lra B[h\;\vec{c}/x]\cas{\rho}}
\end{equation*}
The upper sequent here is equivalent under $\lambda$-conversion to
$\Sigma\theta, h' : \Gamma\cas{\theta} \lra (B\cas{\theta})[h'\;
  \vec{a}/x]$ so this structure is, in  fact, also an instance of the
$\forallR$ rule. Moreover, its upper sequent is obtained via
substitution from the upper sequent of the rule in $\Pi$. The lemma
then follows by induction.

The arguments for the cases when the last rule is an $\existsR$ or an
$\forallL$ are similar and so we provide it explicitly only for the
former. In this case, we have the rule
\begin{equation*}
\infer[\existsR]{\Sigma : \Gamma \lra \exists_\tau x.B}
      {\Sigma, \mathcal{K}, \mathcal{C} \vdash t:\tau &
       \Sigma : \Gamma \lra B[t/x]}
\end{equation*}
ending $\Pi$.
Assuming that the substitution $(\exists_\tau x.B)\cas{\theta}$ uses
the permutation $\pi$ to avoid the capture of nominal constants,
consider the structure
\begin{equation*}
\infer[]{\Sigma\theta : \Gamma\cas{\theta} \lra
                                (\exists_\tau x.B)\cas{\theta}}
      {\Sigma, \mathcal{K}, \mathcal{C} \vdash \pi.t:\tau &
       \Sigma\theta : \Gamma\cas{\theta} \lra
                B\cas{\theta}[\pi.t/x]}
\end{equation*}
This is also obviously an instance of the $\existsR$ rule and its
right upper sequent is related via substitution to that of the rule in
$\Pi$. The lemma follows from these observations by induction.

The only remaining cases for the last rule are $\IL$ and $\CIR$. The
arguments in these cases are, yet again, similar and it suffices to
make only the former explicit. In this case, the end of $\Pi$ has the form
\begin{equation*}
\infer[\IL]
      {\Sigma : \Gamma, p\; \vec{t} \lra C}
      {\vec{x} : B\; S\; \vec{x} \lra S\; \vec{x} &
       \Sigma : \Gamma, S\; \vec{t} \lra C}
\end{equation*}
But then the following
\begin{equation*}
\infer[]
      {\Sigma\theta : \Gamma\cas{\theta}, (p\; \vec{t})\cas{\theta}
                      \lra C\cas{\theta}}
      {\vec{x} : B\; S\; \vec{x} \lra S\; \vec{x} &
       \Sigma\theta : \Gamma\cas{\theta}, (S\;
         \vec{t})\cas{\theta} \lra C\cas{\theta}}
\end{equation*}
is also an instance of the $\IL$ rule. Moreover, the same proof as in
$\Pi$ can be used for the left upper sequent and the right upper
sequent has the requisite form for using the induction hypothesis.
\end{proof}

The proof of Lemma~\ref{lem:proof-subst} effectively defines a
transformation of a derivation $\Pi$ based on a substitution
$\theta$. We shall use the notation $\Pi\cas{\theta}$ to denote the
transformed derivation. Note that $\htf{(\Pi\cas{\theta})}$ can be less
than $\htf{(\Pi)}$. This may happen because the transformed version of a
$\unrhdL$ rule can have fewer upper sequents.

\begin{corollary}
\label{cor:extend}
The following rules are admissible.
\begin{equation*}
\infer[\forallR^*]
      {\Sigma : \Gamma \lra \forall x.B}
      {\Sigma, h : \Gamma \lra B[h\;\vec{a}/x]}
\hspace{3cm}
\infer[\existsL^*]
      {\Sigma : \Gamma, \exists x.B \lra C}
      {\Sigma, h : \Gamma, B[h\;\vec{a}/x] \lra C}
\end{equation*}
where $h \notin \Sigma$ and $\vec{a}$ is any listing of distinct
nominal constants which contains $\supp(B)$.
\end{corollary}
\begin{proof}
Let $\Pi$ be a derivation for $\Gamma \lra B[h\;\vec{a}/x]$, let $h'$
be a variable that does not appear in $\Pi$, and let $\{\vec{c}\} =
\supp(B)$. By Lemma~\ref{lem:proof-subst}, $\Pi\cas{\lambda\vec{a}.h'\
  \vec{c} / h}$ is a valid derivation. Since $\vec{a}$ contains
$\vec{c}$, no nominal constants appear in the substitution
$\{\lambda\vec{a}.h'\  \vec{c} / h\}$. It can now be seen that the
last sequent in  $\Pi\cas{\lambda\vec{a}.h'\;\vec{c} / h}$ has the
form $\Sigma, h' : \Gamma' \lra B'$ where $B' \approx
B[h'\;\vec{c}/h]$ and $\Gamma'$ results from replacing
some of the formulas in $\Gamma$ by ones that they are equivalent to under
$\approx$. But then,
by Lemma~\ref{lem:proof-perm}, there must be a derivation for $\Sigma,
h' : \Gamma \lra B[h'\;\vec{c}/h]$. Using a $\forallR$ rule below this we
get a derivation for $\Sigma : \Gamma \lra \forall x. B$, verifying
the admissibility of $\forallR^*$. The argument for $\existsL^*$ is
analogous.
\end{proof}

We now turn to the main result of this section, the redundancy from a
provability perspective of the $\cut$ rule in \logic. The usual
approach to proving such a property is to define a set of
transformations called cut reductions on derivations that leave the
end sequent unchanged but that have the effect of pushing occurrences
of $\cut$ up the proof tree to the leaves where they can be
immediately eliminated. The difficult part of such a proof is showing
that these cut reductions always terminate. In simpler sequent
calculi such as the one for first-order logic, this argument can be
based on an uncomplicated measure such as the size of the cut formula.
However, the presence of definitions in a logic like \logic renders
this measure inadequate. For example, the following is a natural way
to define a cut reduction between a $\defL$ and a $\defR$ rule that
work on the cut formula:
\begin{equation*}
\begin{array}{c}
\infer[\cut]{\Sigma : \Gamma, \Delta \lra C}{
  \infer[\defR]
        {\Sigma : \Gamma \lra p\;\vec{t}}
        {\deduce{\Sigma : \Gamma \lra B\ p\;\vec{t}}{\Pi'}} &
  \infer[\defL]
        {\Sigma : p\;\vec{t}, \Delta \lra C}
        {\deduce{\Sigma : B\ p\;\vec{t}, \Delta \lra C}{\Pi''}}}
\\
\vspace{-0.5cm}
\\
\Downarrow
\\
\vspace{-0.5cm}
\\
\infer[\cut]{\Sigma : \Gamma, \Delta \lra C}{
  {\deduce{\Sigma : \Gamma \lra B\ p\;\vec{t}}{\Pi'}} &
  {\deduce{\Sigma : B\ p\;\vec{t}, \Delta \lra C}{\Pi''}}}
\end{array}
\end{equation*}
Notice that $B\ p\;\vec{t}$, the cut formula in the new cut introduced
by this transformation, could be more complex than $p\;\vec{t}$, the
old cut formula.
To overcome this difficulty, a more complicated argument based on the
idea of reducibility in the style of Tait \cite{tait67jsl} is often
used. Tiu and Momigliano
\cite{tiu.momigliano} in fact formulate a notion of parametric
reducibility for derivations that is based on the Girard's proof of
strong normalizability for System F \cite{girard89book} and that works
in the presence of the induction and co-induction rules for
definitions. Our proof makes extensive use of this notion and the
associated argument structure.

\begin{theorem}\label{thm:cut-elim}
The $\cut$ rule can be eliminated from \logic without affecting the
provability relation.
\end{theorem}
\begin{proof}
The relationship between \logic and the logic Linc$^-$
treated by Tiu and Momigliano can be understood as follows: Linc$^-$
does not treat the $\nabla$ quantifier and therefore has no rules for
it. Consequently, it does not have nominal constants, it does not
use raising over nominal constants in the rules $\forallR$ and
$\existsL$, it has no need to consider permutations in the $id$ (or
initial) rule and has equality rules in place of nominal abstraction
rules. The rules in \logic other than the ones for $\nabla$, including
the ones for definitions, induction, and co-induction, are essentially
identical to the ones in Linc$^-$ except for the additional attention
to nominal constants.

Tiu and Momigliano's proof can be extended to \logic in a fairly
direct way since the addition of nominal constants and their
treatment in the rules is quite modular and does not create any new
complexities for the reduction rules. The main issues in realizing this
extension is building in the idea of identity under permutations of
nominal constants and lifting the Linc$^-$ notion of
substitution on terms, sequents, and derivations to a form that
avoids capture of nominal constants. The machinery for doing this has
already been developed in Lemmas~\ref{lem:proof-perm} and
\ref{lem:proof-subst}. In the rest of this proof we
assume a familiarity with the argument for cut-elimination for Linc$^-$
and discuss only the changes to the cut reductions of Linc$^-$ to
accommodate the differences.

The $id$ rule in \logic identifies formulas which are equivalent
under $\approx$ which is more permissive than equality under
$\lambda$-convertibility that is used in the Linc$^-$ initial
rule. Correspondingly, we have to
be a bit more careful about the cut reductions associated with the
$id$ (initial) rule. For example, consider the following reduction:
\begin{equation*}
\infer[\cut]{B, \Gamma, \Delta \lra C}{
  \infer[id]
        {\Sigma : \Gamma, B \lra B'}
        {B \approx B'} &
  \deduce{\Sigma : B', \Delta \lra C}
         {\Pi'}}
\hspace{1cm}
\raisebox{1.5ex}{$\Longrightarrow$}
\hspace{1cm}
\deduce{\Sigma : B', \Delta \lra C}{\Pi'}
\end{equation*}
This reduction has not preserved the end sequent. However, we know $B
\approx B'$ and so we can now use Lemma~\ref{lem:proof-perm} to
replace $\Pi'$ with a derivation of $\Sigma : B, \Delta \lra C$.
Then we can use Lemma~\ref{lem:proof-weak} to produce a derivation of
$\Sigma : B, \Gamma, \Delta \lra C$ as desired. The changes to the
cut reduction when $id$ applies to the right upper sequent of the
$\cut$ rule are similar.

The $\forallR$ and $\existsL$ rules of \logic extend the corresponding
rules of Linc$^-$ by raising over nominal constants in the support of
the quantified formula. The $\forallL$ and $\existsR$ rules of \logic
also extend the corresponding rules in Linc$^-$ by allowing
instantiations which contain nominal constants. Despite these changes,
the cut reductions involving these quantifier rules remain unchanged
for \logic except for the treatment of essential cuts that involve an
interaction between $\forallR$ and $\forallL$ and, similarly, between
$\existsR$ and $\existsL$. The first of these is treated as follows:
\begin{equation*}
\begin{array}{c}
\infer[\cut]{\Sigma : \Gamma, \Delta \lra C}{
  \infer[\forallR]
        {\Sigma : \Gamma \lra \forall x.B}
        {\deduce{\Sigma, h : \Gamma \lra B[h\;\vec{c}/x]}{\Pi'}} &
  \infer[\forallL]
        {\Sigma : \Delta, \forall x.B \lra C}
        {\deduce{\Sigma : \Delta, B[t/x] \lra C}{\Pi''}}}
\\
\vspace{-0.5cm}
\\
\Downarrow
\\
\vspace{-0.5cm}
\\
\infer[\cut]{\Sigma : \Gamma, \Delta \lra C}{
  {\deduce{\Sigma : \Gamma \lra B[t/x]}{\Pi'\cas{\lambda\vec{c}.t/h}}} &
  {\deduce{\Sigma : \Delta, B[t/x] \lra C}{\Pi''}}}
\end{array}
\end{equation*}
The existence of the derivation $\Pi'\cas{\lambda \vec{c}. t/h}$ (with
height at most that of $\Pi'$) is guaranteed by
Lemma~\ref{lem:proof-subst}. The end sequent of this derivation is
$\Sigma : \Gamma\cas{\lambda \vec{c}. t/h} \lra B[h\
\vec{c}/x]\cas{\lambda \vec{c}. t/h}$. However,
$\Gamma\cas{\lambda\vec{c}.t/h}\approx \Gamma$ because $h$ is new to
$\Gamma$ and $B[h\;\vec{c}/x]\cas{\lambda \vec{c}. t/h} \approx
B[t/x]$ because $\{\vec{c}\} = \supp(B)$ and so $\lambda \vec{c}. t$ has
no nominal constants in common with $\supp(B)$. Thus, by
Lemma~\ref{lem:proof-perm} and by an abuse of notation, we may
consider $\Pi'\cas{\lambda\vec{c}./h}$ to also be a derivation of
$\Sigma : \Gamma \lra B[t/x]$. The reduction for a cut involving an
interaction between an $\existsR$ and an $\existsL$ rule is analogous.

The logic \logic extends the equality rules in Linc$^-$ to treat the
more general case of nominal abstraction. Our notion of nominal
capture-avoiding substitution correspondingly generalizes the Linc$^-$
notion of substitution, and we have shown in
Lemma~\ref{lem:proof-subst} that this preserves provability. Thus the
reductions for nominal abstraction are the same as for equality,
except that we use nominal capture-avoiding substitution in place of regular
substitution. For example, the essential cut involving an interaction
between an $\unrhdR$ and an $\unrhdL$ rule is treated as follows:
\begin{equation*}
\infer[\cut]{\Sigma : \Gamma, \Delta \lra C}{
  \infer[\unrhdR]
        {\Sigma : \Gamma \lra s\unrhd t}
        {} &
  \infer[\unrhdL]
        {\Sigma : \Delta, s\unrhd t \lra C}
        {\left\{\raisebox{-1.5ex}{
          \deduce{\Sigma\theta : \Delta\cas{\theta} \lra C\cas{\theta}}
                 {\Pi_\theta}
         }\right\}}}
\hspace{1cm}
\raisebox{1.5ex}{$\Longrightarrow$}
\hspace{1cm}
\deduce{\Sigma : \Delta \lra C}{\Pi_\epsilon}
\end{equation*}
Here we know $s\unrhd t$ holds and thus $\epsilon$, the identity
substitution, is a solution to this nominal abstraction. Therefore we
have the derivation $\Pi_\epsilon$ as needed. We can then apply
Lemma~\ref{lem:proof-weak} to weaken this derivation to one for
$\Sigma : \Gamma, \Delta \lra C$. For the other cuts involving nominal
abstraction, we make use of the fact proved in
Lemma~\ref{lem:proof-subst} that nominal capturing avoiding
substitution preserves provability. This allows us to commute other
rules with $\unrhdL$. For example, consider the following reduction of
a cut where the upper right derivation uses an $\unrhdL$ on a formula
different from the cut formula:
\begin{equation*}
\begin{array}{c}
\infer[\cut]{\Sigma : \Gamma, \Delta, s\unrhd t \lra C}{
  \deduce{\Sigma : \Gamma \lra B}{\Pi'} &
  \infer[\unrhdL]
        {\Sigma : B, \Delta, s\unrhd t \lra C}
        {\left\{\raisebox{-1.5ex}{
          \deduce{\Sigma\theta :
                  B\cas{\theta}, \Delta\cas{\theta} \lra C\cas{\theta}}
                 {\Pi_\theta}
         }\right\}}}
\\
\vspace{-0.5cm}
\\
\Downarrow
\\
\vspace{-0.5cm}
\\
\infer[\unrhdL]
      {\hspace{2.8cm}\Sigma : \Gamma, \Delta, s\unrhd t \lra C\hspace{2.8cm}}
      {\left\{\raisebox{-3ex}{
        \infer[\cut]
              {\Sigma\theta : \Gamma\cas{\theta}, \Delta\cas{\theta} \lra
                C\cas{\theta}}
              {\deduce{\Sigma\theta : \Gamma\cas{\theta} \lra B\cas{\theta}}
                      {\Pi'\cas{\theta}} &
              \deduce{\Sigma\theta : B\cas{\theta}, \Delta\cas{\theta} \lra
                C\cas{\theta}}
               {\Pi_\theta}}
       }\right\}\hspace{0.6cm}}
\end{array}
\end{equation*}

Finally, \logic has new rules for treating the $\nabla$-quantifier.
The only reduction rule which deals specifically with either the
$\nablaL$ or $\nablaR$ rule is the essential cut between both rules
which is treated as follows:
\begin{equation*}
\begin{array}{c}
\infer[\cut]{\Sigma : \Gamma, \Delta \lra C}{
  \infer[\nablaR]
        {\Sigma : \Gamma \lra \nabla x.B}
        {\deduce{\Sigma : \Gamma \lra B[a/x]}{\Pi'}} &
  \infer[\nablaL]
        {\Sigma : \nabla x.B, \Delta \lra C}
        {\deduce{\Sigma : B[a/x], \Delta \lra C}{\Pi''}}}
\\
\vspace{-0.5cm}
\\
\Downarrow
\\
\vspace{-0.5cm}
\\
\infer[\cut]{\Sigma : \Gamma, \Delta \lra C}
      {\deduce{\Sigma : \Gamma \lra B[a/x]}{\Pi'} &
       \deduce{\Sigma : B[a/x], \Delta \lra C}{\Pi''}}.
\end{array}
\end{equation*}

With these changes, the cut-elimination argument for Linc$^-$
extends to \logic, \ie, \logic admits cut-elimination.

\end{proof}

The consistency of \logic is an easy consequence of
Theorem~\ref{thm:cut-elim}.

\begin{corollary}\label{consistency}
The logic \logic is consistent, \ie, not all sequents are provable in
it.
\end{corollary}

\begin{proof} The sequent $\lra \bot$ has no cut-free proof and,
  hence, no proof in \logic.
\end{proof}

\section{Adequacy of Encodings and Theorems in the Meta-logic}
\label{sec:meta-adequacy}

The logic \logic provides various features such as $\lambda$-terms,
definitions, and $\nabla$-quantification which form a convenient
vehicle for encoding computational systems. With all these features,
one might rightfully ask if our encodings in \logic are faithful
representations of the computational systems they describe. This kind
of property for encodings, which is formally known as {\em adequacy},
is similar to the one that we have already encountered with respect to
the specification logic. A proof of adequacy establishes a
relationship between terms and judgments in an object system and their
encoding in \logic in such a way that we can relate reasoning results
proven about the encoding to results about the original system. In
this section we discuss adequacy in more detail, we describe the
general approach to proving adequacy, and we present an example which
illustrates some of the nuances which may arise for particular
encodings.

At a philosophical level, adequacy is the method by which we assign
meaning to our logic. Without adequacy, the logic has only behavior.
Thus, one may naively ask a question such as, ``what does the
$\nabla$-quantifier mean?'' To which a valid answer is that the
$\nabla$-quantifier has no meaning in itself. It has the behavior of
introducing a fresh nominal constant into a formula, but it is only
through adequacy that we can interpret this behavior and provide it
with some meaning. For instance, we might establish a correspondence
between nominal constants in a \logic formula and free variables in a
typing judgment for an object system. In this setting, the meaning of
$\nabla$-quantification can be interpreted as quantifying over fresh
free variables.

A proof of adequacy for an encoding of an object system in \logic
consists of two parts:
\begin{enumerate}
\item the description of a bijection between the terms of the object
system and their encoding in \logic, and

\item a proof, based on this bijection, that a judgment in the object
system holds if and only if its encoding in \logic is provable.
\end{enumerate}
For the second point, the cut-elimination result from
Section~\ref{sec:meta-theory} is of critical importance since it
allows us to restrict the sort of proofs we must consider. Without
an independent proof of the cut-elimination property, proving adequacy
would require establishing something like a cut-elimination theorem
relative to each encoding that we wish to prove adequate.

Our ultimate objective is, of course, to prove theorems about the
original system. However, this follows naturally from the proof of a
relevant theorem in \logic and the adequacy of encodings in the
following way: 1) using adequacy, object level judgments are
translated into \logic formulas, 2) the relevant theorem proven in
\logic is used as a lemma on these formulas, and 3) using adequacy,
the result of that lemma application is then translated back into an
object level judgment. The end result is that the theorem is proven
for the object system while most of the reasoning takes place within
\logic. The {\sl cut} rule plays an essential role here as it allows
us to use theorems proven in \logic as lemmas which is very useful in
reasoning and absolutely vital in the adequacy argument outlined
above. It is for this reason that we cannot simply exclude the {\sl
  cut} rule from our logic and hope to avoid the work involved in
showing cut-elimination.

It is important to remember that adequacy is only an interface issue,
\ie, it is only a question about the ``inputs'' and ``outputs'' of
\logic. We show that an encoding of an object system (the ``input'')
is adequate and we use this to relate reasoning results in \logic
(the ``output'') to results about the original system.
Any auxiliary notions that we use
within the logic in order to establish the results of interest do not
matter for the purposes of adequacy. This is not to say that we do not
care what goes on in between. Certainly we have designed the logic
\logic so that the intermediate reasoning can closely mimic the
informal reasoning that is typically done. But in the end, the
correctness of the reasoning that is performed depends only on the
adequacy results and the cut-elimination property for \logic.

\begin{figure}[t]
\begin{align*}
\infer
 {(\lambda x. r) \Downarrow (\lambda x. r)}
 {}
&&
\infer
 {(m\ n) \Downarrow v}
 {m \Downarrow (\lambda x. r) &
  r[x := n] \Downarrow v}
\end{align*}
\caption{An evaluation relation for untyped $\lambda$-terms}
\label{fig:ulc-eval}
\end{figure}

\begin{figure}[t]
\begin{align*}
&\eval {(\uabs R)} {(\uabs R)} \mueq \top \\
&\eval {(\app M N)} V \mueq \exists R.~ \eval M {(\uabs R)} \land
\eval {(R\ N)} V
\end{align*}
\caption{An encoding of the evaluation relation in Figure~\ref{fig:ulc-eval}}
\label{fig:ulc-eval-enc}
\end{figure}

As an example, let us now consider the adequacy of a proof of
determinacy for an evaluation relation on untyped $\lambda$-terms. The
evaluation relation of interest is presented in
Figure~\ref{fig:ulc-eval}. This example will be sufficient to
illustrate the key issues involved in showing adequacy for an encoding
in \logic, while a more thorough example is presented later in
Section~\ref{sec:adequacy-seq}.

To represent untyped $\lambda$-terms in \logic, we introduce the type
$tm$ along with the constructors $\hsl{app} : tm \to tm \to tm$ and
$\hsl{abs} : (tm \to tm) \to tm$. Then we encode the evaluation
relation as a definition for a predicate $\hsl{eval} : tm \to tm \to o$
as shown in Figure~\ref{fig:ulc-eval-enc}. Given this definition,
we can prove the following determinacy result in
\logic:
\begin{equation*}
\forall t,v_1,v_2. (\eval t v_1 \land \eval t v_2) \supset v_1 = v_2.
\end{equation*}
What we want to do is use this result to obtain a similar determinacy
result for evaluation in the original system. We will develop the
bijections and the associated adequacy lemmas below to be able to
obtain such a translation.

We begin by defining a mapping $\enc{\cdot}$ from untyped
$\lambda$-terms to their representation in \logic:
\begin{align*}
\enc{x} = x && \enc{t_1\ t_2} = \app {\enc{t_1}} {\enc{t_2}} &&
\enc{(\lambda x. t)} = \uabs ({\lambda x. \enc{t}})
\end{align*}
Note that we conflate the names of variables in untyped
$\lambda$-terms with the corresponding names in \logic. In truth, the
bound variables of untyped $\lambda$-terms will be mapped to bound
variables of type $tm$ in \logic, while the free variables of untyped
$\lambda$-terms will be mapped to nominal constants of type $tm$ in
\logic. Assuming a one-to-one correspondence between such terms, the
above mapping is obviously bijective. Moreover, closed untyped
$\lambda$-terms will map to terms in \logic without nominal constants
and vice-versa. Thus our representation of untyped $\lambda$-terms is
adequate.\footnote{A subtle but important point: we do not permit
  $\nabla$-quantification at type $tm \to tm$. Allowing this would
  mean that we will have terms in \logic such as $\abs c$ for a
  nominal constant $c$. Since such a term cannot be the image of any
  untyped $\lambda$-term, the representation would then not be
  adequate.}

Since we use the substitution mechanism of \logic in
the definition of {\sl eval} to encode substitution on untyped
$\lambda$-terms, we will later need to know that these two
substitution relations are related via $\enc{\cdot}$ in the following
sense.
\begin{lemma}\label{lem:enc-comp}
Let $t_1$ and $t_2$ be untyped $\lambda$-terms. Then $\enc{t_1[x :=
  t_2]} = \enc{t_1}[\enc{t_2}/x]$ where the substitution on the left
takes place in the context of untyped $\lambda$-terms and the
substitution on the right takes place in \logic.
\end{lemma}
\begin{proof}
The proof is by a straightforward induction on the structure of $t_1$.
\end{proof}

Next we want to show an if-and-only-if relationship between the
original evaluation judgment and its encoding in \logic. This is
formalized as follows.
\begin{lemma}\label{lem:eval-adq}
$t \Downarrow v$ has a derivation if and only if $\lra \eval {\enc{t}}
{\enc{v}}$ is provable in \logic.
\end{lemma}
\begin{proof}
The proof in the forward direction is by straightforward induction on
the derivation of $t \Downarrow v$.

For the backward direction we first note that $\lra \eval {\enc{t}}
{\enc{v}}$ must have a cut-free derivation by
Theorem~\ref{thm:cut-elim}. The proof will be by induction on the
height of this cut-free derivation. The cut-free derivation must end
with $\defR$ though for ease of presentation we may suppose that it
ends with $\defR^p$.\footnote{Note that cut-elimination was shown for
  the logic containing $\defL$ and $\defR$, whereas $\defL^p$ and
  $\defR^p$ are only admissible additions to the logic.} The
interesting case is when considering the second clause for {\sl eval},
\ie, when $t = (m\ n)$ and the derivation ends as follows.
\begin{equation*}
\infer
 [\defR^p]
 {\lra \eval {(\app {\enc{m}} {\enc{n}})} \enc{v}}
 {\infer[\existsR]
   {\lra \exists r.~ \eval {\enc{m}} {(\uabs r)} \land
                     \eval {(r\ \enc{n})} \enc{v}}
   {\infer[\landR]
     {\lra \eval {\enc{m}} {(\uabs R)} \land
           \eval {(R\ \enc{n})} \enc{v}}
     {\lra \eval {\enc{m}} {(\uabs R)} &
      \lra \eval {(R\ \enc{n})} \enc{v}}
   }
 }
\end{equation*}
Here $R$ is a term of type $tm \to tm$. By the bijectivity of
$\enc{\cdot}$, we know that $(\uabs R)$ is the representation of an
untyped $\lambda$-term and thus we can apply the inductive hypothesis
to the upper left sequent. Similarly, we can apply the inductive
hypothesis to the upper right sequent after using
Lemma~\ref{lem:enc-comp} to convert $(R\ \enc{n})$ to the
representation of a substitution over untyped $\lambda$-terms.
\end{proof}

It was essential to applying the inductive hypothesis in the proof of
the lemma above that our mapping $\enc{\cdot}$ was a bijection. This
property would not hold, for instance, if we restricted attention to
only closed untyped $\lambda$-terms in the object language and we still
allowed $\nabla$-quantification at type $tm$ and, hence, admitted
nominal constants of this type; specifically, we would have terms of
type $tm$ in \logic that do not correspond to any closed untyped
$\lambda$-terms. We would then not have been able to apply the
inductive hypothesis in the proof of Lemma~\ref{lem:eval-adq}
because we would have to consider the possibility that particular
occurrences of the $\existsR$ rule generalize on terms of type $tm$
that contain one or more nominal constants. However, it is still
possible to
use a proof in \logic to establish a property about the original
system even in this case. To do this, we would have to
introduce a definition in \logic for the class of terms of type $tm$
that {\em do not} contain nominal constants and we would have to
relativize the theorem we prove in \logic to the class of terms
satisfying this definition. From this perspective, adequacy is not
always just a matter of mapping terms in the object system to terms in
\logic: we may need to map terms in the object system to terms
satisfying a particular predicate in \logic.

We now return to showing how a theorem in \logic about the determinacy
of the evaluation relation can be combined with the adequacy property
for the encoding of untyped $\lambda$-terms to yield a theorem about
the determinacy of the evaluation relation in the original calculus.

\begin{theorem}
\label{thm:eval-det-roundtrip}
If $t \Downarrow v_1$ and $t \Downarrow v_2$ then $v_1$ equals $v_2$.
\end{theorem}
\begin{proof}
Suppose $t \Downarrow v_1$ and $t \Downarrow v_2$ both have
derivations. By Lemma~\ref{lem:eval-adq}, that means we have proofs of
$\lra \eval {\enc{t}} {\enc{v_1}}$ and $\lra \eval {\enc{t}}
{\enc{v_2}}$. We also know from before that the following has a
derivation in \logic:
\begin{equation*}
\lra \forall t,v_1,v_2. (\eval t v_1 \land \eval t v_2) \supset v_1 =
v_2.
\end{equation*}
Then using the rules $\forallL$, $\supsetL$, $\landR$, $id$, and
$\cut$, we
can construct a derivation of $\lra \enc{v_1} = \enc{v_2}$. By
Theorem~\ref{thm:cut-elim} we know that $\lra \enc{v_1} = \enc{v_2}$
must have a cut-free derivation. This derivation must end with
$\unrhdR$ which applies only if $\enc{v_1}$ is equal to $\enc{v_2}$.
Since $\enc{\cdot}$ is a bijection, this means that $v_1$ equals
$v_2$.
\end{proof}

The discussion of adequacy in this section is reminiscent of an
earlier discussion relative to the specification logic and hence
raises the question of what, if anything, is different. The main
observation here is that the logic \logic is significantly richer than
the \hh logic. In particular, when proving properties about an \hh
specification, reasoning is conducted using general mathematical
techniques, while for proving properties about an encoding in \logic,
the reasoning is conducted within \logic itself. Thus, when working
with \logic, we use adequacy to connect results proven in \logic with
corresponding results about the original system. One may informally
think of this as establishing adequacy for the theorems in \logic
relative to their counterparts about the original system.



\chapter{An Interactive Theorem Prover for the Meta-logic}
\label{ch:architecture}

As part of this thesis, we have developed an interactive theorem
prover called Abella for the logic \logic \cite{gacek08ijcar,
  gacek-abella-website}. Abella is implemented in OCaml and currently
comprises approximately 4,000 lines of code. This system has been
available to the public as open source software since March 2008 and
has, in fact, been downloaded by several researchers. One of the key
components of a theorem prover for \logic is the treatment of nominal
abstraction problems. We have discussed in
Section~\ref{ssec:complete-sets} how the task of finding a solution to
particular instances of the nominal abstraction predicate can be
reduced to solving higher-order unification problems. Abella makes use
of this reduction. Moreover, it assumes that the resulting unification
problems lie within a restricted class known as the {\em higher-order
  pattern unification} class \cite{miller91jlc,nipkow93lics}. To solve
such problems, it uses an algorithm developed by Nadathur and Linnell
\cite{nadathur05iclp} that was initially implemented in Standard ML
and that has subsequently been adapted to OCaml.

In this chapter, we briefly describe the architecture of Abella; this
discussion serves the auxiliary purpose of building up ideas and
terminology that we need for presenting applications of \logic in
Chapter~\ref{ch:applications}. Abella requires proofs to be
constructed through an interaction with a user. At any time, the state
of a proof is represented as a collection of subgoals, all of which
need to be proved for the overall proof to succeed. The user applies a
{\em tactic} to a subgoal in order to make progress towards a
completed proof. If we think of the proof as a derivation constructed
in \logic, then the subgoals in Abella correspond to sequents in the
derivation which do not themselves have derivations as yet. Tactics
then correspond to schemes for applying the rules of \logic to such
sequents in order to (incrementally) fill out their derivations.

There are two guiding principles for designing tactics in Abella:
\begin{enumerate}
\item they should correspond to some combination of rules from \logic,
and
\item they should correspond to natural reasoning steps.
\end{enumerate}
For the most part, the rules of \logic themselves resemble natural
reasoning steps. The role of many tactics therefore, is simply to
chain these together into larger steps. For example, given a goal of
the form
\begin{equation*}
\Sigma : \Gamma \lra \forall \vec{x}.~ H_1 \supset \ldots \supset H_n
\supset C
\end{equation*}
we may want to transition in one step into a goal of the following form:
\begin{equation*}
\Sigma, \vec{x} : \Gamma, H_1, \ldots, H_n \lra C.
\end{equation*}
Tactics are also used to group together many alternative rules. For
example, a ``case analysis'' tactic may actually perform $\lorL$,
$\landL$, $\botL$, $\defL$, $\existsL$, or $\nablaL$ based on the
structure of the formula to which it is applied.

In the rest of this chapter, we describe two areas in which tactics
greatly massage the rules of \logic into a convenient form. The first
concerns how hypotheses or lemmas of a particular form can be applied
to other hypotheses. The second concerns a treatment of induction and
co-induction which can naturally accommodate even sophisticated
inductive and co-inductive arguments.

\section{A Framework for Using Lemmas}
\label{sec:application-lemmas}

Suppose we have a hypothesis of the form
\begin{equation*}
\forall \vec{x}.~ H_1 \supset \ldots \supset H_n \supset C
\end{equation*}
and further hypotheses $H_1', \ldots, H_n'$ which match $H_1, \ldots,
H_n$ under proper instantiations of the $\vec{x}$. Then we would like
a tactic to apply the first hypothesis to $H_1'$, \ldots, $H_n'$, \ie,
a tactic which finds the proper instantiations for $\vec{x}$ and
chains together the rules of \logic to generate a new hypothesis $C'$
that is the corresponding instantiation of $C$. To be more specific,
let $\Gamma$ contain $H_1'$, \ldots, $H_n'$. Then we want a tactic
which constructs the derivation
\begin{equation*}
\infer=[\forallL]
 {\Gamma, \forall \vec{x}.~ H_1 \supset \ldots
   \supset H_n \supset C \lra B}
 {\infer[\supsetL]
   {\Gamma, H_1[\vec{t}/\vec{x}] \supset \ldots
    \supset H_n[\vec{t}/\vec{x}] \supset C[\vec{t}/\vec{x}] \lra B}
   {\deduce{\Gamma \lra H_1[\vec{t}/\vec{x}]}{\Pi_1} &
    \deduce
     {\Gamma, H_2[\vec{t}/\vec{x}] \supset \ldots
      \supset H_n[\vec{t}/\vec{x}] \supset C[\vec{t}/\vec{x}] \lra B}
     {\deduce
       {\vdots}
       {\infer[\supsetL]
         {\Gamma, H_n[\vec{t}/\vec{x}] \supset C[\vec{t}/\vec{x}] \lra B}
         {\deduce{\Gamma \lra H_n[\vec{t}/\vec{x}]}{\Pi_n} &
          \deduce{\Gamma, C[\vec{t}/\vec{x}] \lra B}{\Pi}}
       }
     }
   }
 }
\end{equation*}
where each $\Pi_i$ is just the identity rule. In an actual
implementation, this construction may be accomplished by replacing the
variables $\vec{x}$ with instantiatable meta-variables $\vec{v}$ and
using unification between $H_i[\vec{v}/\vec{x}]$ and $H_i'$ to
determine specific values for the $\vec{v}$.

Using the above construction, we can think of more sophisticated ways
in which $H_i'$ will match $H_i[\vec{t}/\vec{x}]$. All that we
effectively require is that a derivation of $H_i' \lra
H_i[\vec{t}/\vec{x}]$ can be constructed automatically. One useful
case arises when $H_i[\vec{t}/\vec{x}]$ has the form
$\nabla\vec{z}.H_i''$ for some formula $H_i''$, and where $H_i'$ will
match $H_i''[\vec{a}/\vec{z}]$ for some distinct listing of nominal
constants $\vec{a}$ which are not in the support of $H_i''$. If such a
case holds, then a derivation of $H_i' \lra \nabla \vec{z}.H_i''$ can
be constructed by repeated use of $\nablaR$ followed by the initial
rule. As before, in an actual implementation, we might be working with
$H_i[\vec{v}/\vec{x}] = \nabla\vec{z}.H_i'''$ where $\vec{v}$ are
instantiatable meta-variables. In such a case, we can determine proper
instantiations for the $\vec{v}$ by solving the nominal abstraction
$\lambda \vec{z} . H_i''' \unrhd H_i'$.

Typically, lemmas also have the form
\begin{equation*}
\forall \vec{x}.~ H_1 \supset \ldots \supset H_n \supset C.
\end{equation*}
If we have independently proven such a lemma, then we can use $\cut$
to bring it in as a hypothesis at any time. Then we can use this lemma
together with other hypotheses as described above so as to derive a
suitable instance of $C$.

By supporting an easy and direct use of lemmas, the system encourages
large proofs to be broken down into separate lemmas which build
towards a final result. In practice, these intermediate lemmas and the
points at which they are used are often the most important pieces in
the development of a proof. In fact, the structure of most arguments
is the following: use the induction rule, then perform case analysis
and finally use particular lemmas and the induction hypothesis to
obtain the goal. Thus in actual presentation of proofs, the detailed
proof steps are hidden by default, and instead the focus is on the
series of lemmas that lead to the desired conclusions
\cite{gacek-abella-website}.

A final point worth mentioning is that we deliberately consider
formulas of the form
\begin{equation*}
\forall \vec{x}.~ H_1 \supset \ldots \supset H_n \supset C
\end{equation*}
even though the following form is equivalent and perhaps more easy to
read for humans:
\begin{equation*}
\forall \vec{x}.~ H_1 \land \ldots \land H_n \supset C.
\end{equation*}
The reason we prefer the first form is two-fold: 1) it has a recursive
structure which is easier to work with in an implementation, and 2) in
the degenerate case the when $n = 0$, then first form is
$\forall\vec{x}.~C$ while the second is the more obtuse
$\forall\vec{x}. \top \supset C$. In the future, we shall always work
with formulas in the first form.

\section{An Annotation Based Scheme for Induction}
\label{sec:induct-co-induct}

The rule for induction in \logic can be somewhat awkward to use from a
traditional reasoning perspective: it requires one to formulate an
invariant $S$, prove that $S$ is truly an invariant, and then use $S$
in place of the predicate that was given by the inductive definition
under consideration. In traditional reasoning, these steps are often
merged into a single idea which is called simply ``reasoning by
induction.'' In this section we present a treatment of induction based
on annotating formulas which aims to capture this simplified approach
to induction. Further, we justify this treatment by translating the
tactic that underlies it into a particular application of the logical
rules of \logic.

Let us consider a very simple inductive argument to introduce the
annotation based treatment of induction. Suppose we define {\sl even}
and {\sl odd} on natural numbers as follows.
\begin{align*}
\even z &\mueq \top & \odd (s\ z) &\mueq \top \\
\even (s\ (s\ N)) &\mueq \even N & \odd (s\ (s\ N)) &\mueq
\odd N
\end{align*}
Suppose we want to prove that if $N$ is even then $s\ N$ is odd:
\begin{equation*}
\forall N.~ \even N \supset \odd (s\ N).
\end{equation*}
The proof is by induction on the {\sl even} hypothesis. The annotation
based treatment of this induction proceeds by creating a new
hypothesis (called the inductive hypothesis) of the form
\begin{equation*}
\forall N.~ (\even N)^* \supset \odd (s\ N)
\end{equation*}
and changing the goal to
\begin{equation*}
\forall N.~ (\even N)^@ \supset \odd (s\ N).
\end{equation*}
The $*$ annotation indicates that the inductive hypothesis can only be
applied to an argument which has that same annotation. The $@$
annotation indicates that when this atomic formula is subjected to
case analysis, any recursive calls to {\sl even} will be annotated
with $*$. In all other respects, the annotations are to be ignored,
and besides the induction tactic there is no way to introduce these
annotations. In this way, Abella allows the inductive hypothesis to be
applied only when the distinguished inductive argument has been
subjected to case analysis.

Coming back to the proof, let us abbreviate the inductive hypothesis
by $IH$. Then we can eventually do case analysis on the {\sl even}
hypothesis which leads to the following sequents.
\begin{align*}
IH \lra \odd (s\ z) &&
IH, (\even N')^* \lra \odd (s\ (s\ (s\ N')))
\end{align*}
The first of these is easily provable. In the second we apply the
inductive hypothesis which is allowed based on the annotations, and
this produces a hypothesis of $\odd (s\ N')$. The rest of the proof is
straightforward.

We will now show how this annotation based treatment of induction is
sound by translating it to rules from \logic. Suppose we want to prove
the following.
\begin{equation*}
\forall \vec{x}.~ H_1 \supset \ldots \supset H_n \supset C
\end{equation*}
Further, assume that we want to do this by induction on $H_i = p\
\vec{t}$ where $p$ is defined by $\forall \vec{y}. p\ \vec{y} \mueq B\
p\ \vec{y}$. Then we define the invariant $S$ as
\begin{equation*}
S = \lambda \vec{y} . \forall \vec{x}.~ \vec{y} = \vec{t} \supset H_1
\supset \ldots \supset H_n \supset C
\end{equation*}
where $\vec{y} = \vec{t}$ denotes an equality between appropriately
typed tuples involving the indicated terms. Using this invariant, we
can construct the following derivation in \logic.
\begin{equation*}
\infer=[\forallR]
 {\cdot : \cdot \lra \forall \vec{x}.~ H_1 \supset \ldots \supset H_n
   \supset C}
 {\infer=[\supsetR]
   {\vec{x} : \cdot \lra H_1 \supset \ldots \supset H_n \supset C}
   {\infer[\cL]
     {\vec{x} : H_1, \ldots, H_n \lra C}
     {\infer[\IL]
      {\vec{x} : p\ \vec{t}, H_1, \ldots, H_n \lra C}
      {\deduce{\vec{y} : B\ S\ \vec{y} \lra S\ \vec{y}}{\Pi_S} &&&&
       \deduce{\vec{x} : S\ \vec{t}, H_1, \ldots, H_n \lra C}{\Pi}}
     }
   }
 }
\end{equation*}
Now, the missing derivation $\Pi$ is trivial to construct using
$\forallL$, $\supsetL$, $\unrhdR$ and $id$. We fill in the
other missing derivation, $\Pi_S$, as follows:
\begin{equation*}
\infer=[\forallR]
 {\vec{y} : B\ S\ \vec{y} \lra \forall \vec{x}.~ \vec{y} = \vec{t}
   \supset H_1 \supset \ldots \supset H_n \supset C}
 {\infer=[\supsetR]
   {\vec{x}, \vec{y} : B\ S\ \vec{y} \lra \vec{y} = \vec{t}
    \supset H_1 \supset \ldots \supset H_n \supset C}
   {\infer[\unrhdL_{\CSNAS}]
     {\vec{x}, \vec{y} : B\ S\ \vec{y}, \vec{y} = \vec{t},
      H_1, \ldots, H_n \lra C}
     {\deduce
       {\vec{x} : B\ S\ \vec{t}, H_1, \ldots, H_n \lra C}
       {\Pi_S'}
     }
   }
 }
\end{equation*}
Then we fill in $\Pi_S'$ based on the content of the inductive
argument carried out within the annotation based scheme.

To complete this picture, let us consider how uses of the induction
hypothesis in the annotation based treatment of induction correspond to
making use of the hypothesis $B\ S\ \vec{t}$ in constructing the
derivation $\Pi'_S$. Within the annotation based treatment, the
induction hypothesis has the following form:
\begin{equation*}
\forall \vec{x}.~ H_1 \supset \ldots \supset (p\ \vec{t})^* \supset \ldots
\supset H_n \supset C.
\end{equation*}
Given the restrictions on annotations, this hypothesis can only be
used if instantiations are found for the $\vec{x}$ such that $(p\
\vec{t})^*$ is equal to one of the $(p\ \vec{s})^*$ which occurs as a
result of case analysis on the original hypothesis of $(p\
\vec{t})^@$. By understanding case analysis as $\defL$ in \logic, we
see that these occurrences of $(p\ \vec{s})^*$ for which the induction
hypothesis is applicable are exactly those occurrences of $p$ in $B\ p\
\vec{t}$. In turn, the induction invariant is available for those same
occurrences of $p$ when constructing the derivation $\Pi'_S$, which is
precisely what is realized via the hypothesis $B\ S\ \vec{t}$. Thus
the annotation based treatment of induction can be translated to a
proper derivation in \logic, and therefore the treatment is sound.

\section{Extensions to the Basic Scheme for Induction}

The treatment of induction that we have just described can be extended
in a few different ways. Each of these brings some additional
complications to the construction of a corresponding derivation in
\logic. For clarity of presentation, we shall consider each extension
in isolation, but we note that they could all be combined.

\subsection{Induction on a Predicate in the Scope of Generic Quantifiers}
\label{sec:induct-with-nabla}

We can extend the annotation based treatment of induction to work with
predicates which occur underneath $\nabla$-quantifiers. Suppose again
we want to prove
\begin{equation*}
\forall \vec{x}.~ H_1 \supset \ldots \supset H_n \supset C
\end{equation*}
where, this time, we want to induct on $H_i = \nabla\vec{z}.~ p\
\vec{t}$ where $p$ is defined by $\forall \vec{y}. p\ \vec{y} \mueq B\
p\ \vec{y}$. Within the annotation based treatment, nothing needs to
be changed to cater to this situation: $(p\ \vec{t})$ is annotated
with $*$ in the inductive hypothesis and with $@$ in the goal and the
rules for applying an inductive hypothesis with $\nabla$s over the
inductive argument are the same as those described in
Section~\ref{sec:application-lemmas}.

We justify this treatment by defining the invariant $S$ as follows.
\begin{equation*}
S = \lambda \vec{y} . \forall \vec{x}.~ (\lambda\vec{z}. \vec{t}
\unrhd \vec{y}) \supset H_1 \supset \ldots \supset H_n \supset C
\end{equation*}
We can follow the original construction with this invariant, and the
only wrinkle is in the construction of $\Pi_S$, a derivation of
$\vec{y} : B\ S\ \vec{y} \lra S\ \vec{y}$. We construct this as
follows.
\begin{equation*}
\infer=[\forallR]
 {\vec{y} : B\ S\ \vec{y} \lra \forall \vec{x}.~ (\lambda\vec{z}. \vec{t}
   \unrhd \vec{y}) \supset H_1 \supset \ldots \supset H_n \supset C}
 {\infer=[\supsetR]
   {\vec{x}, \vec{y} : B\ S\ \vec{y} \lra (\lambda\vec{z}. \vec{t}
    \unrhd \vec{y}) \supset H_1 \supset \ldots \supset H_n \supset C}
   {\infer[\unrhdL_{\CSNAS}]
     {\vec{x}, \vec{y} : B\ S\ \vec{y}, (\lambda\vec{z}. \vec{t}
      \unrhd \vec{y}), H_1, \ldots, H_n \lra C}
     {\deduce
       {\vec{x} : B\ S\ \vec{t}, H_1, \ldots, H_n \lra C}
       {\Pi_S'}
     }
   }
 }
\end{equation*}
Here and in the future, we simplify the presentation by treating the
free variables $\vec{z}$ in $\vec{t}$ as nominal constants. Now we
fill in $\Pi_S'$ based on the content of the inductive argument
carried out within the annotation based scheme. After using $\nablaL$
and case analysis on $H_i = \nabla\vec{z} . p\ \vec{t}$ we will have
$B\ p\ \vec{t}$ and also $B\ S\ \vec{t}$. Thus we have the inductive
hypothesis available for the recursive calls to $p$. The restrictions
enforced by the nominal abstraction in $S$ are the same as those
enforced when applying hypotheses which have embedded occurrences of
$\nabla$, as per the discussion in
Section~\ref{sec:application-lemmas}. Thus this treatment is sound.

\subsection{Induction in the Presence of Additional Premises}
\label{sec:induct-context}

We extend the annotation based treatment of induction by allowing
induction in the context of other hypotheses. That is, instead of
proving $\cdot : \cdot \lra \forall \vec{x}.~ H_1 \supset \ldots
\supset H_n \supset C$, we prove
\begin{equation*}
\Sigma : \Gamma \lra \forall \vec{x}.~ H_1 \supset
\ldots \supset H_n \supset C
\end{equation*}
Within the annotation based treatment of induction, there is nothing
that needs to be changed to handle this case: we annotate the goal and
generate an annotated induction hypothesis which is added to the other
hypotheses.

To verify the soundness of this extension, we reconstruct the original
soundness argument using the invariant $S' = \lambda \vec{y}. \forall
\Sigma.~ \bigwedge \Gamma \supset S\ \vec{y}$ where $S$ is the
invariant prescribed in the original construction and $\bigwedge
\Gamma$ denotes the conjunction of all formulas in $\Gamma$. Then the
only significant change in the construction is that $\Pi_S$ needs to
be a derivation of $\vec{y} : B\ S'\ \vec{y} \lra S'\ \vec{y}$. Using
$\forallR$, $\supsetR$, and $\landL$ this becomes $\Sigma, \vec{y} :
\Gamma, B\ S'\ \vec{y} \lra S\ \vec{y}$. Finally, we know $\forall
\Sigma.\forall \vec{y}.~ \bigwedge \Gamma \supset S'\ \vec{y} \supset
S\ \vec{y}$ by the definition of $S'$, and since $B$ does not use its
first argument negatively (due to stratification), we know $\forall
\Sigma. \forall \vec{y}. \bigwedge \Gamma \supset B\ S'\ \vec{y}
\supset B\ S\ \vec{y}$. By using this, all we have left to show is
$\Sigma, \vec{y} : \Gamma, B\ S\ \vec{y} \lra S\ \vec{y}$ which we can
unfold as in the original construction and what is left matches the
work done in the annotation based treatment.

\subsection{Delayed Applications of the Induction Hypothesis}
\label{sec:non-immed-induct}

Another extension we can make is to allow the inductive hypothesis to
be applied not just for immediate recursive calls, but for finitely
nested ones as well. This is supported in the annotation based
treatment by saying that case analysis on a hypothesis with a $*$
annotation results in recursive calls which also have the $*$
annotation. For example, taking {\sl even} and {\sl odd} as before,
suppose we want to prove every natural number is either even or odd:
\begin{equation*}
\forall N.~ \nat N \supset \even N \lor \odd N.
\end{equation*}
The proof is by induction on $\nat N$. Thus we have the inductive
hypothesis $IH$ as follows:
\begin{equation*}
\forall N.~ (\nat N)^* \supset \even N \lor \odd N.
\end{equation*}
When we perform case analysis on the hypothesis $(\nat N)^@$ in the
goal it leads to the following sequents.
\begin{align*}
IH \lra \even z \lor \odd z &&
IH, (\nat N')^* \lra \even (s\ N') \lor \odd (s\ N')
\end{align*}
The first sequent is trivial to prove, and we can apply case analysis
to $(\nat N')^*$ in the second to get the following two sequents.
\begin{align*}
IH \lra \even (s\ z) \lor \odd (s\ z) &&
IH, (\nat N'')^* \lra \even (s\ (s\ N'')) \lor \odd (s\ (s\ N''))
\end{align*}
Again the first sequent is trivial. In the second sequent we can apply
the inductive hypothesis to get the sequent
\begin{equation*}
\ldots, \even N'' \lor \odd N'' \lra
\even (s\ (s\ N'')) \lor \odd (s\ (s\ N'')).
\end{equation*}
Now we can apply $\lorL$ and the rest of the proof is trivial to construct.

The justification for this extension in \logic is to use the invariant
$S' = \lambda \vec{y}. S\ \vec{y} \land B\ S\ \vec{y}$ in the original
construction where $S$ is the original invariant. Then only
significant change in the construction is that we are required to fill
out the following derivation
\begin{equation*}
\infer[\landR]
 {\vec{y} : B\ S'\ \vec{y} \lra S'\ \vec{y}}
 {\deduce{\vec{y} : B\ S'\ \vec{y} \lra S\ \vec{y}}{\Pi_1} &&&&
  \deduce{\vec{y} : B\ S'\ \vec{y} \lra B\ S\ \vec{y}}{\Pi_2}}
\end{equation*}
Now note that $\forall \vec{x}.~ S'\ \vec{x} \supset S\ \vec{x}$ and
$\forall \vec{x}.~ S'\ \vec{x} \supset B\ S\ \vec{x}$ are both
trivially provable after expanding the definition of $S'$. Since $B$
does not allow its first argument to occur negatively (due to
stratification) this means we can inductively construct derivations of
$\forall \vec{x}.~ B\ S'\ \vec{x} \supset B\ S\ \vec{x}$ and $\forall
\vec{x}.~ B\ S'\ \vec{x} \supset B\ (B\ S)\ \vec{x}$. The construction
of the derivation $\Pi_2$ follows directly from the first of these.
The derivation $\Pi_1$ contains the real content of the inductive
proof. If case analysis is eventually used on $H_i = p\ \vec{t}$ in
this derivation then the $\vec{y}$ will have been instantiated with
$\vec{t}$ so that we have the hypothesis $B\ S'\ \vec{t}$. Thus we
will have $B\ S\ \vec{t}$ which is the regular inductive hypothesis
and also $B\ (B\ S)\ \vec{t}$ which is the inductive hypothesis
applied to recursive calls nested at depth two. This depth can be
extended to any finite number by repeating the above construction with
the appropriate $S'$.

\subsection{Nested Inductions}
\label{sec:nested-inductions}

The use of annotations can be extended to allow nested inductions. For
example, suppose we define the following predicate {\sl ack} for
computing the Ackermann function.
\begin{align*}
\ack z N (s\ N) &\mueq \top \\
\ack {(s\ M)} z R &\mueq \ack M {(s\ z)} R \\
\ack {(s\ M)} {(s\ N)} R &\mueq \exists R'.~ \ack {(s\ M)} N {R'}
\land \ack M {R'} R
\end{align*}
And suppose we want to prove that this function is total in its first
two arguments:
\begin{equation*}
\forall M, N.~ \nat M \supset \nat N \supset \exists R.~ \nat R \land
\ack M N R
\end{equation*}
The proof requires an outer induction on $\nat M$ and an inner
induction on $\nat N$. In the annotation based treatment of induction,
this is realized as follows. Applying induction to $\nat M$ produces
the outer inductive hypothesis
\begin{equation*}
\forall M, N.~ (\nat M)^* \supset \nat N \supset \exists R.~ \nat R \land
\ack M N R
\end{equation*}
and the goal
\begin{equation*}
\forall M, N.~ (\nat M)^@ \supset \nat N \supset \exists R.~ \nat R \land
\ack M N R.
\end{equation*}
Then applying induction to $\nat N$ in this goal produces the inner
inductive hypothesis
\begin{equation*}
\forall M, N.~ (\nat M)^@ \supset (\nat N)^{**} \supset \exists R.~ \nat R \land
\ack M N R
\end{equation*}
and the goal
\begin{equation*}
\forall M, N.~ (\nat M)^@ \supset (\nat N)^{@@} \supset \exists R.~ \nat R \land
\ack M N R.
\end{equation*}
The treatment of annotations is the same as described before. The
annotations $*$ and $**$ as well as $@$ and $@@$ are considered
distinct and unrelated. Thus the outer inductive hypothesis applies as
before, while the inner inductive hypothesis can only be applied to
$(\nat M)^@$ from the goal and something with the $**$ annotation
which can only come from case analysis on $(\nat N)^{@@}$.

We will use this treatment to finish the proof of totality for the
Ackermann function. Let $IH$ and $IH'$ be the outer and inner
induction hypotheses, respectively. Then the interesting part of the
proof comes after we have done case analysis on both $(\nat M)^@$ and
$(\nat N)^{@@}$. In particular, in the case where $M = s\ M'$ and $N =
s\ N'$ we need to prove the following sequent.
\begin{equation*}
IH, IH', (\nat (s\ M'))^@, (\nat M')^*, (\nat N')^* \lra \exists R.~
\nat R \land \ack {(s\ M')} {(s\ N')} R
\end{equation*}
Note that we must have performed contraction on $(\nat M)^@$ prior to
case analysis in order to keep a copy of it. Then we can apply the
inner induction hypothesis to $(\nat (s\ M'))^@$ and $(\nat N')^*$ to
get the hypotheses $\nat R'$ and $\ack {(s\ M')} N {R'}$ for some new
variable $R'$. Applying the outer inductive hypothesis to $(\nat M')^*$
and $\nat R'$ produces the hypotheses $\nat R''$ and $\ack {M'} {R'}
{R''}$. Then we can apply $\existsR$ with $R = R''$, and the rest of
the proof is trivial.

We now justify the annotation based treatment of nested induction.
As in the original construction, suppose we want to prove
\begin{equation*}
\forall \vec{x}.~ H_1 \supset \ldots \supset H_n \supset C.
\end{equation*}
And suppose the proof is by an outer induction on $H_i = p\ \vec{t}$
where $p$ is defined by $\forall \vec{y}.p\ \vec{y} \mueq B\ p\
\vec{y}$ and an inner induction on $H_j = q\ \vec{s}$ where $q$ is
defined by $\forall \vec{z}.q\ \vec{z} \mueq B'\ q\ \vec{z}$. We
proceed with the original construction using the original invariant
$S$ for the outer induction. This leaves us with a need to prove the
following.
\begin{equation*}
\vec{x} : B\ S\ \vec{t}, H_1, \ldots, H_n \lra C
\end{equation*}
Now we apply contraction on $H_j = q\ \vec{s}$ and induct on one of
the copies using the following invariant.
\begin{equation*}
S' = \lambda \vec{z} . \forall \vec{x}.~B\ S\ \vec{t} \supset \vec{z}
= \vec{s} \supset H_1 \supset \ldots \supset H_n \supset C
\end{equation*}
The only non-trivial sequent to prove will be $\vec{z} : B'\ S'\
\vec{z} \lra S'\ \vec{z}$. Applying $\forallR$, $\supsetR$, and
$\unrhdL_{\CSNAS}$, this reduces to showing
\begin{equation*}
\vec{x} : B'\ S'\ \vec{s}, B\ S\ \vec{t}, H_1, \ldots, H_n \lra C
\end{equation*}
Now from $B\ S\ \vec{t}$ we have the outer induction invariant
available for the recursive calls to $p$ which arise from case
analysis on $H_i = p\ \vec{t}$. From $B'\ S'\ \vec{s}$ we have the
inner induction invariant available for the recursive calls to $q$
which arise from case analysis on $H_j = q\ \vec{s}$. The caveat is
that the inner induction invariant $S'$ requires a proof of $B\ S\
\vec{t}$. This constrains the variables $\vec{x}$ in the inner
induction variant based on their occurrences in $\vec{t}$. In the
annotation based treatment, the requirement of a hypothesis with a $@$
annotation enforces exactly this condition for the inner inductive
hypothesis.

\section{An Annotation Based Scheme for Co-induction}
\label{sec:annot-co-induct}

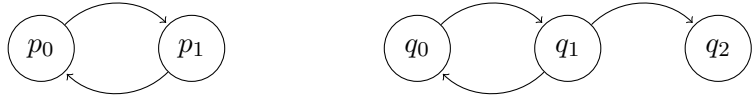
\begin{figure}[t]
\centering
\begin{tikzpicture}[shorten >=1pt,->,bend angle=45,auto]
  \tikzstyle{vertex}=[circle,draw=black,minimum size=25pt,inner
  sep=1pt]

  \node[vertex] (p0) at (0,0) {$p_0$};
  \node[vertex] (p1) at (2,0) {$p_1$};
  \path (p0) edge [bend left] (p1)
        (p1) edge [bend left] (p0) ;

  \node[vertex] (q0) at (5,0) {$q_0$};
  \node[vertex] (q1) at (7,0) {$q_1$};
  \node[vertex] (q2) at (9,0) {$q_2$};
  \path (q0) edge [bend left] (q1)
        (q1) edge [bend left] (q0)
        (q1) edge [bend left] (q2) ;
\end{tikzpicture}
\label{fig:co-step}
\caption{Transition diagrams for two different processes}
\end{figure}

We can also use annotations to treat co-induction. To illustrate how
this works, we will take an example from the domain of process
calculi. Let us consider the two processes depicted in
Figure~\ref{fig:co-step}. Here the circles represent states and the
arrows represent possible transitions between those states. We say
that a $P$ is {\em simulated by} a state $Q$ if for every transition
that $P$ can make to a state $P'$ there exists a state $Q'$ to which
$Q$ can transition and such that $P'$ is simulated by $Q'$. We
consider the notion of simulation as co-inductive so a state can be
simulated by another state even if both have infinite (possibly
cyclic) chains of transitions from them. Suppose then, that we want to
show that the state $p_0$ is simulated by the state $q_0$. We can see
that this is true by considering all possible transitions from these
states and recognizing that $p_1$ is simulated by the state $q_1$.

Let us now think of conducting this example in \logic. We start by
encoding the two processes using the following definition of {\sl
  step}.
\begin{align*}
\step {p_0} {p_1} &\triangleq \top &
\step {p_1} {p_0} &\triangleq \top \\
\step {q_0} {q_1} &\triangleq \top &
\step {q_1} {q_0} &\triangleq \top &
\step {q_1} {q_2} &\triangleq \top
\end{align*}
Then we define simulation as a co-inductive predicate $\simp P Q$ which
holds when the process $P$ is simulated by the process $Q$. The
precise definition is as follows.
\begin{equation*}
\simp P Q \nueq \forall P'.~ \step P P' \supset \exists Q'.~ \step Q
Q' \land \simp {P'} {Q'}
\end{equation*}
Our goal is then to prove $\simp {p_0} {q_0}$ which we generalize
based on the argument sketched above into the following formula to
prove:
\begin{equation*}
\forall P, Q.~ (P = p_0 \land Q = q_0) \lor (P = p_1 \land Q = q_1)
\supset \simp P Q.
\end{equation*}
If we apply annotation based co-induction to this goal we get the
co-inductive hypothesis
\begin{equation*}
\forall P, Q.~ (P = p_0 \land Q = q_0) \lor (P = p_1 \land Q = q_1)
\supset (\simp P Q)^+
\end{equation*}
and the new goal
\begin{equation*}
\forall P, Q.~ (P = p_0 \land Q = q_0) \lor (P = p_1 \land Q = q_1)
\supset (\simp P Q)^{\#}.
\end{equation*}
Note that the annotations for co-induction apply to the consequent of
an implication rather than one of the hypotheses. The rules for these
new annotations are as follows. If we unfold (\ie, use $\defR$ on) a
co-inductive definition with a $\#$ annotation then all of its
recursive calls have the $+$ annotation. Hypotheses with a $+$
annotation are obtained from the co-inductive hypothesis and can {\em
  only} be used to match a goal with the $+$ annotation. For all other
purposes, the annotations can be ignored. The proof of the above
simulation eventually reduces to the following two sequents where $CH$
is the co-inductive hypothesis.
\begin{align*}
CH \lra (\simp {p_0} {q_0})^\# &&
CH \lra (\simp {p_1} {q_1})^\#
\end{align*}
The proofs of these two sequents are similar, so we will consider only
the first one. Here if we apply $\defR$ we will eventually end up with
the sequent
\begin{align*}
CH \lra (\simp {p_1} {q_1})^+.
\end{align*}
At this point we can apply the co-inductive hypothesis to get a
hypothesis which will match the goal.

We can justify the annotation based treatment of co-induction by
translating it into appropriate rules from \logic. Suppose we want to
prove the following where $p$ is defined by $\forall \vec{y}. p\
\vec{y} \nueq B\ p\ \vec{y}$.
\begin{equation*}
\forall \vec{x}.~ H_1 \supset \ldots \supset H_n \supset p\ \vec{t}
\end{equation*}
We proceed as in the construction for induction to get the sequent
\begin{equation*}
\vec{x} : H_1, \ldots, H_n \lra p\ \vec{t}.
\end{equation*}
We then apply co-induction with the invariant $S$ as follows.
\begin{equation*}
S = \lambda\vec{y}. \exists \vec{x} .~ \vec{y} = \vec{t} \land H_1
\land \ldots \land H_n
\end{equation*}
The $\CIR$ rule applied to the earlier sequent requires us to show
$\vec{x} : H_1, \ldots, H_n \lra S\ \vec{t}$ which is trivial and
$\vec{y} : S\ \vec{y} \lra B\ S\ \vec{y}$ which contains the real
content of the co-inductive proof. A derivation of this later sequent
can be constructed as follows.
\begin{equation*}
\infer=[\existsL]
 {\vec{y} : \exists \vec{x} .~ \vec{y} = \vec{t} \land H_1
  \land \ldots \land H_n \lra B\ S\ \vec{y}}
 {\infer=[\landL]
  {\vec{y}, \vec{x} : \vec{y} = \vec{t} \land H_1
   \land \ldots \land H_n \lra B\ S\ \vec{y}}
  {\infer[\unrhdL_{\CSNAS}]
   {\vec{y}, \vec{x} : \vec{y} = \vec{t}, H_1
    \ldots, H_n \lra B\ S\ \vec{y}}
   {\vec{x} : H_1, \ldots, H_n \lra B\ S\ \vec{t}}
  }
 }
\end{equation*}
The derivation for the upper-most sequent here can be constructed
based on the argument carried out in the the annotation based
treatment. Within that argument, when the goal $(p\ \vec{t})^\#$ is
unfolded, the recursive calls will be annotated with $+$ and will be
provable using the co-inductive hypothesis. This is what is given in
the formal derivation by the goal $B\ S\ \vec{t}$.


This annotation based treatment of co-induction can be extended in ways
similar to the inductive treatment. For example, we can allow
co-induction within a context of other hypotheses, or we can allow the
goal to be unfolded multiple times before applying the co-inductive
hypotheses. The soundness arguments for these extensions are similar
to the inductive case.




\chapter{A Two-level Logic Approach to Reasoning}
\label{ch:two-level-reasoning}

One approach to reasoning about object systems is to encode their
descriptions directly into definitions in \logic and to then use the
inference rules of \logic with these definitions. In this chapter we
explore an alternative approach. In particular, we show how the
meta-logic \logic can be used to encode the specification logic \hh
and to then reason about \hh specifications through this encoding.
This is the two-level logic approach to reasoning that was enunciated
by McDowell and Miller earlier in the context of the meta-logic \FOLDN
\cite{mcdowell02tocl}.

An important part of assessing the value of the two-level logic
approach to reasoning is understanding both its benefits and its
costs. One benefit is that the specification logic carves out a useful
subset of the specifications that are possible in the meta-logic while
at the same time possessing a complete proof search procedure which
make it possible to execute the specifications. A second benefit is
that by encoding an entire specification logic in the meta-logic, we
can formalize properties of the specification logic and make them
available during reasoning. An auxiliary observation in this context
is that because of the way the specification logic can be used to
encode object systems, the properties of this logic that are used in
meta-logic reasoning often turn out to be based on intuitions about
the properties of the object systems themselves. From a cost
perspective, one issue with the two-level logic approach to reasoning
is that there is an additional overhead to reasoning about
specifications through the encoded semantics of the specification
logic rather than directly. Another cost to be considered is that
because the specification logic is only a subset of the full range of
specifications allowed by the meta-logic, this approach in some ways
limits what we are able to say within a specification.

After all aspects are taken into account, we believe that the
combination of the \hh specification logic and the meta-logic \logic
seems to provide a nice balance between the benefits and costs of the
two-level logic approach to reasoning. The specification logic \hh
elegantly encodes many systems of interest, and there are efficient
implementations of this specification logic. Moreover, as we saw in
Section~\ref{sec:spec-example}, the properties of \hh provide useful
results during reasoning. Finally, as we shall see in this chapter,
the encoding of \hh into \logic is lightweight and therefore imposes
little overhead on the reasoning process.

The rest of this chapter is laid out as follows.
Section~\ref{sec:encod-spec-logic} describes the encoding of \hh into
\logic. Section~\ref{sec:form-meta-theory} formalizes some properties
of \hh as theorems in \logic; these theorems can then be used as
lemmas to simplify subsequent reasoning.
Section~\ref{sec:example-two-level} illustrates our specific
realization of the two-level logic approach to reasoning and
demonstrates its power by using it to formalize the informal proof
that we have presented in Chapter~\ref{ch:introduction} of the fact
that types are preserved by evaluation in the simply-typed
$\lambda$-calculus. Finally, Section~\ref{sec:adequacy-g} discusses
the issue of adequacy relative to the two-level logic approach to
reasoning.

\section{Encoding the Specification Logic}
\label{sec:encod-spec-logic}

There are two components to our encoding of the specification logic
\hh into the meta-logic \logic. First, we encode the syntax by
defining a mapping $\psi$ from specification logic types and terms to
meta-logic types and terms. Since both logics are constructed from
Church's simple theory of types and hence contain subsets of
expressions that are isomorphic, this encoding can be very shallow.
Second, we encode the semantics of \hh (\ie, the provability relation)
via the definition of a suitably chosen atomic judgment in \logic.
This encoding is lightweight which makes later reasoning fairly
transparent. To aid in that reasoning we observe some formulas that
can be proved in \logic involving the judgment that encodes
specification logic provability. These theorems of \logic can be used
as lemmas to shorten other proofs that we would want to construct in
\logic.

\subsection{Encoding the Syntax of the Specification Logic}

The types of our specification logic are mapped to isomorphic types in
the meta-logic. We define the mapping $\psi$ on types as follows.
\begin{align*}
\psi(\tau) = \tau \;\;\mbox{ if $\tau$ is a base type} &&
\psi(\tau_1 \to \tau_2) = \psi(\tau_1) \to \psi(\tau_2)
\end{align*}
For each specification type, we assume a bijective mapping between
eigenvariables of that type (in the specification logic) and nominal
constants of that type (in the meta-logic). We denote this mapping
using subscripts: the eigenvariable $h$ maps to the nominal constant
$a_h$ and the nominal constant $a$ maps to the eigenvariable $h_a$.
Using this, we define the encoding of specification terms as follows.
\begin{equation*}
\psi(c) = c \;\;\mbox{ if $c$ is a constant} \hspace{1cm}
\psi(h) = a_h \;\;\mbox{ if $h$ is an eigenvariable}
\end{equation*}
\begin{equation*}
\psi(x) = x \;\;\mbox{ if $x$ is a variable} \hspace{1cm}
\psi(\lambda x. t) = \lambda x. \psi(t) \hspace{1cm}
\psi(t_1\ t_2) = \psi(t_1)\ \psi(t_2)
\end{equation*}

Now for clarity and correctness of the encoding, we make two
adjustments to this mapping. First, the specification logic type $o$
for formulas is mapped to a distinguished type $frm$ to avoid
conflicting with the type $o$ for meta-logic formulas. Second, we
introduce a distinguished type $atm$ for atomic specification logic
formulas and a constructor $\langle \cdot \rangle : atm \to frm$ to
inject such atoms into formulas. We then modify the type of the
specification logic $\supset$ connective to $atm \to frm \to frm$ to
enforce the restriction that the left-hand side of an implication is
atomic.

Note that we map specification logic constants to constants of the
same name in the meta-logic. This means, for example, that the
meta-logic will have two constants called $\land$. One will be the
logical connective of \logic with type $o \to o \to o$, and the other
will be a term constructor for representations of specification logic
formulas with type $frm \to frm \to frm$. We will always be able to
distinguish between such constants based on the context in which they
are used.

Our encoding is clearly bijective. Furthermore, typing judgments are
preserved by the bijection in the following sense. Let $\mathcal{K}$
denote the set of meta-logic constants which represent the constants
of the specification logic, then $\Sigma \vdash t : \tau$ is a valid
specification logic typing if and only if $\psi(\Sigma), \mathcal{K}
\vdash \psi(t) : \psi(\tau)$ is a valid meta-logic typing where
$\psi(\Sigma) = \{\psi(h) \mid h \in \Sigma\}$. Since our mapping
$\psi$ is bijective we will use the mapping $\psi^{-1}$ freely.

\subsection{Encoding the Semantics of the Specification Logic}

In the encoding of the semantics of our specification logic, we shall
use two auxiliary notions. First, we introduce a type $nt$ for natural
numbers with the constructors $z : nt$ and $s : nt \to nt$ and the
predicate $\hsl{nat} : nt \to o$ defined by
\begin{align*}
\nat z &\mueq \top &
\nat {(s\ N)} &\mueq \nat N
\end{align*}
As we see below, these numbers will be used to capture the idea of the
height of a derivation in our encoding of the provability relation of
the specification logic. Second, we introduce a type $atmlist$ with
constructors $nil : atmlist$ and the infix $::\ :\ atm \to atmlist \to
atmlist$ and the predicate $\hsl{member} : atm \to atmlist \to o$
defined by
\begin{align*}
\member A (A::L) \mueq \top &&
\member A (B::L) \mueq \member B L
\end{align*}
We shall use lists of this kind and the corresponding membership
predicate to encode the addition to premise sets when trying to prove
implicational formulas in \hh.

We encode \hh provability in \logic through the predicate $\hsl{seq} :
nt \to atmlist \to frm \to o$ that is defined by the clauses in
Figure~\ref{fig:seq}. This encoding of \hh provability derives from
McDowell and Miller \cite{mcdowell02tocl}. As described in
Chapter~\ref{ch:specification-logic}, proofs in \hh contain sequents
of the form $\Sigma : \Delta,{\cal L}\vdash G$ where $\Delta$ is a
fixed set of closed $D$-formulas and $\cal L$ is a varying set of
atomic formulas. The eigenvariables in $\Sigma$ are encoded as nominal
constants in \logic. The meta-logic predicate $\hsl{prog} : atm \to
frm \to o$ is used to represent the $D$-formulas in $\Delta$: the $D$
formula $\forall\vec x.[G_1\supset\cdots\supset G_n\supset A]$ is
encoded as the clause $\forall \vec x. \prog A (G_1\land\cdots\land
G_n) \triangleq \top$ and $\forall \vec{x}. A$ is encoded by the
clause $\forall \vec x. \prog A \top \triangleq \top$. We denote these
{\sl prog} clauses by $\Psi(\Delta)$, and we note that such clauses do
not contain any nominal constants since the formulas of $\Delta$ are
closed. Finally, the \hh sequent is encoded as $\seq N
{\psi(\mathcal{L})} {\psi(G)}$ where we define $\psi$ on lists of
atomic formulas as $\psi(A_n, \ldots, A_1) = A_1 :: \ldots :: A_n ::
nil$. The argument $N$, written as a subscript, roughly corresponds to
the height of the proof tree and is used in inductive arguments. To
simplify notation, we write $L \tridotn n G$ for $\seq n L G$ and $L
\tridot G$ for $\exists n . \nat n \land \seq n L G$. When $L$ is
$nil$ we write simply $\,\tridotn n G$ or $\,\tridot G$.

\begin{figure}[t]
\begin{align*}
& \seq {(s\ N)} L \top \mueq \top \\
& \seq {(s\ N)} L (B \lor C) \mueq \seq N L B \\
& \seq {(s\ N)} L (B \lor C) \mueq \seq N L C \\
& \seq {(s\ N)} L (B \land C) \mueq \seq N L B \land \seq N L C \\
& \seq {(s\ N)} L (A \supset B) \mueq \seq N {(A
  :: L)} B \\
& \seq {(s\ N)} L (\forall B) \mueq \nabla x. \seq N L (B\  x) \\
& \seq {(s\ N)} L (\exists B) \mueq \exists x. \seq N L (B\  x) \\
& \seq {(s\ N)} L \langle A \rangle \mueq \member A L \\
& \seq {(s\ N)} L \langle A \rangle \mueq \exists b. \prog A b
\land \seq N L b
\end{align*}
\caption{Second-order hereditary Harrop logic in \logic}
\label{fig:seq}
\end{figure}

Proofs of universally quantified $G$ formulas in \hh are generic in
nature. A natural encoding of this (object-level) quantifier in the
definition of {\sl seq} uses a (meta-level) $\nabla$-quantifier. In
the case of proving an implication, the atomic assumption is
maintained in a list (the second argument of {\sl seq}). The last
clause for {\sl seq} implements backchaining over a fixed \hh
specification (stored as {\sl prog} atomic formulas). The matching of
atomic judgments to heads of clauses is handled by the treatment of
definitions in the logic \logic, thus the last rule for {\sl seq}
simply performs this matching and makes a recursive call on the
corresponding clause body.

Note that for each specification type $\tau$ we have the constants
$\forall_\tau : (\tau \to frm) \to frm$ and $\exists_\tau : (\tau \to
frm) \to frm$, thus we should have {\sl seq} clauses for each of
these. However, here and going forward, we present only general rules
for $\forall$ and $\exists$, knowing that the actual rules are easily
derived from these.

With this kind of an encoding, we can now formulate and prove in
\logic statements about what is or is not provable in \hh. In
constructing such proofs, we shall sometimes need induction over the
height of derivations. Such arguments can be realized via induction on
the predicate $\nat n$ in a formula of the form $\exists n . \nat n
\land \seq n L G$ occurring on the left of a sequent. We may sometimes
also want to use strong induction in our arguments. Towards this end,
we introduce the auxiliary predicate $\lt : nt \to nt \to o$ defined
as follows.
\begin{align*}
\lt z (s\ N) &\mueq \top \\
\lt {(s\ M)} {(s\ N)} &\mueq \lt M N
\end{align*}
Now, a formula such as $\forall n. (\nat n) \supset P$ can be proven
using strong induction by proving $\forall n, m. (\nat n \land \lt n m
\land \nat m) \supset P$ and using induction on $\nat m$.
Section~\ref{sec:example-two-level} contains an example that uses this
approach. Finally, the $\defL$ rule can be used to realize case
analysis based reasoning in the derivation of an atomic goal. Using
this rule leading eventually to a consideration of the different ways
in which an atomic judgment may have been inferred in the
specification logic.

In the rest of this chapter, we shall conduct all of our reasoning by
constructing derivations in \logic, with the exception of adequacy
arguments where we will need to reason over \logic derivations. Thus,
when we say that ``a formula $F$ is provable'' or that ``a formula $F$
is provable in \logic'', we shall mean that the sequent $\lra F$ is
provable in \logic. Moreover, when we talk about the ``proof of a
formula F'' we shall mean the derivation in \logic of the sequent
$\lra F$. When we say that such proofs are constructed ``by
induction'' we shall mean that we use the $\IL$ rule of \logic with an
induction invariant derived from the entire sequent being considered.
We shall also talk about proving a formula by induction on one of its
hypotheses (\ie, one of its subformulas to the left of a $\supset$) by
which we mean following the constructions for induction described in
Chapter~\ref{ch:architecture}. The construction of the derivations in
\logic is often straightforward, with only a few sequents which may be
interesting, and so we shall frequently skip directly to such
sequents. Finally, we shall often use running text to describe the
construction of a derivation in \logic; this is possible since the
rules of \logic often mimic traditional mathematical reasoning, but it
must be remembered that the proof is still being carried out within
\logic.

Several of the results that we present below concern the provability
of formulas in \logic. While our proofs of these results here involve
arguing about derivations in \logic, it is important to note that
these arguments sketch a scheme for actually carrying out the proof
{\it within} a system such as Abella. Thus, the justification for
using such formulas in subsequent arguments is completely formalized
through actual mechanical proofs and the lemma mechanism of Abella; in
particular, the resulting style of (mechanized) argument does not rely
on the informal proofs we present to justify the approach.

\subsection{Some Provable Properties of the Specification Logic}
\label{sec:deriv-rules-infer}

It is often convenient to reason directly with formulas of the form $L
\tridot G$ rather than expanding them into $\exists n . \nat n \land
\seq n L G$. In this section, we show that certain schematic formulas
corresponding to $\tridot$ judgments are provable in \logic. Using
these as lemmas allows us to encode certain direct forms of reasoning
about $\tridot$ in \logic proofs. The particular formulas that we show
to be provable in \logic closely mirror the clauses which define the
{\sl seq} predicate.

\begin{lemma}
\label{lem:deriv-forward}
The following formulas are provable in \logic.
\begin{enumerate}
\item\label{item:top} $\forall \ell. (\ell \tridot \top)$
\item\label{item:or1} $\forall \ell, g_1, g_2. (\ell \tridot g_1)
\supset (\ell \tridot g_1 \lor g_2)$
\item\label{item:or2} $\forall \ell, g_1, g_2. (\ell \tridot g_2)
\supset (\ell \tridot g_1 \lor g_2)$
\item\label{item:and} $\forall \ell, g_1, g_2. (\ell \tridot g_1)
\land (\ell \tridot g_2) \supset (\ell \tridot g_1 \land g_2)$
\item\label{item:imp} $\forall \ell, a, g. (a::\ell \tridot g) \supset
(\ell \tridot a \supset g)$
\item\label{item:all} $\forall \ell, g. (\nabla x . (\ell \tridot (g\
x)) \supset (\ell \tridot \forall g)$
\item\label{item:ex} $\forall \ell, g, t. (\ell \tridot (g\ t))
\supset (\ell \tridot \exists g)$
\end{enumerate}
\end{lemma}
\begin{proof}
It is easy to see that the formulas~\ref{item:top}, \ref{item:or1},
\ref{item:or2}, \ref{item:imp}, and \ref{item:ex} are provable in
\logic by unfolding (\ie, using $\defR$ on) the goal formulas.

In the straightforward construction of a proof of
formula~\ref{item:and}, we shall need to construct a proof of the
following sequent.
\begin{equation*}
\nat n, \seq n \ell g_1, \nat m, \seq m \ell g_2 \lra \exists p. \nat
p \land \seq p \ell (g_1 \land g_2).
\end{equation*}
To prove this we must reconcile the measures $n$ and $m$. Towards this
end, we might first show that the following formula that relates $n$
and $m$ is provable in \logic:
\begin{equation*}
\forall m, n. (\nat m) \land (\nat n) \supset (\lt m n) \lor (m = n)
\lor (\lt n m).
\end{equation*}
This can be proved by induction on one of the {\sl nat} hypotheses.
Then we can also prove the following formula which allows us to
increase the measure of a derivation:
\begin{equation*}
\forall m, n, \ell, g . (\lt m n) \land (\ell \tridotn m g) \supset
(\ell \tridotn n g).
\end{equation*}
This is proved by induction on $\lt m n$. Using these two lemmas the
rest of the proof is straightforward.

In constructing a proof of Formula~\ref{item:all} we will find it
necessary to construct a proof of the sequent
\begin{equation*}
\exists n. \nat n \land \seq n \ell (g\ a) \lra \exists m. \nat m \land
\seq m \ell (\forall g).
\end{equation*}
where $a$ is a nominal constant. Now when we apply $\existsL$, we have
the sequent
\begin{equation*}
\nat (n'\ a) \land \seq {(n'\ a)} \ell (g\ a) \lra \exists m. \nat m \land
\seq m \ell (\forall g).
\end{equation*}
The raising of $n'$ over $a$ here prevents this proof from going
through immediately, thus we need the following lemma.
\begin{equation*}
\forall n. (\nabla x. \nat (n\ x)) \supset \exists p. n = \lambda y. p
\end{equation*}
This is proved by induction on {\sl nat}. Once we apply this lemma we
have $n' = \lambda y. p$ for some $p$ and rest of the proof is
straightforward.
\end{proof}

\section[Formalizing Properties of the Specification Logic]{Formalizing Meta-Theoretic Properties of the Specification Logic}
\label{sec:form-meta-theory}

In Section~\ref{sec:prop-spec-logic} we observed certain
meta-theoretic properties of \hh which are useful in reasoning about
\hh specifications. Since we have encoded the entire
specification logic into \logic, we can formalize such properties of
the specification logic within \logic. In particular, we can consider
particular formulas in \logic that encode these properties and then we
can show that these formulas are provable in \logic. Doing this will
allow us to later bring these properties to bear on particular
reasoning tasks that are carried out using \logic. The particular
properties of \hh that we consider in this way in this section are
monotonicity, instantiation, and cut admissibility. With one
exception, the proofs of these properties never use a {\sl prog}
formula except in the initial rule and thus the proofs are independent
of any particular specification encoded in {\sl prog}. The one
exception is specifically noted, and even here the proof is
independent of the specification.

\paragraph{Monotonicity}

The statement of monotonicity for \hh, expressed as a formula of
\logic, is
\begin{equation*}
\forall n, \ell_1, \ell_2, g. (\ell_1 \tridotn n g) \land (\forall e .
\member e \ell_1 \supset \member e \ell_2) \supset (\ell_2 \tridotn n
g).
\end{equation*}
The proof is by straightforward induction on the hypothesis $\nat n$
in $\ell_1 \tridotn n g$.

\paragraph{Instantiation}

The instantiation property recovers the notion of universal
quantification from our representation of the specification logic
$\forall$ using $\nabla$. This property is expressed in \logic through
the formula
\begin{equation*}
\forall \ell, g. (\nabla x. (\ell\  x) \tridotn n (g\  x))
\supset \forall t. (\ell\  t) \tridotn n (g\  t).
\end{equation*}
Stated another way, although $\nabla$ quantification cannot be
replaced by $\forall$ quantification in general, it can be replaced in
this way when dealing with specification judgments. The proof of this
formula is by induction on the hypothesis $\nat n$ in $(\ell\ x)
\tridotn n (g\ x)$, and the following two auxiliary results are useful
in constructing this proof.
\begin{equation*}
\forall \ell, a. (\nabla x. \member {(a\ x)} {(\ell\ x)}) \supset
\forall t. \member {(a\ t)} {(\ell\ t)}
\end{equation*}
\begin{equation*}
\forall a, b. (\nabla x. \prog {(a\; x)} {(b\; x)}) \supset \forall t.
\prog {(a\ t)} {(b\ t)}
\end{equation*}
The first is proved by induction on the {\sl member} hypothesis. The
second depends on the particular specification encoded in {\sl prog},
but the core of the proof is always applying $\defL$ to $\prog {(a\;
  x)} {(b\; x)}$ followed by $\defR$ on $\prog {(a\ t)} {(b\ t)}$.
This will succeed for any specification since {\sl prog} only performs
pattern matching and contains no ``logic.''

\paragraph{Cut admissibility}
The cut admissibility property of \hh is expressed in \logic through
the formula
\begin{tabbing} \qquad
$\forall \ell, a, g. (\ell \tridot \langle a\rangle) \land
(a :: \ell \tridot g) \supset (\ell \tridot g).$
\end{tabbing}
The proof is by induction on the $\nat n$ assumption in $\exists n .
\nat n \land \seq n {(a::\ell)} g$. If $n = z$ then the {\sl seq}
judgment is impossible, thus we know $n = s\ m$ for some $m$. The
proof proceeds by case analysis on the {\sl seq} judgment.
\begin{enumerate}
\item One case is when $g = \langle a' \rangle$ and $\member {a'}
(a::\ell)$. Applying $\defL$ to this {\sl member} hypothesis
results in two additional cases: either
$a = a'$ so that $\ell \tridot \langle a \rangle$
holds by assumption, or we know $\member {a'} \ell$ and thus $\ell
\tridot \langle a' \rangle$ holds by applying $\defR^p$ and {\sl init}.

\item Another case is when $g = a' \supset g'$ so that we have $a' ::
a :: \ell \tridotn m g'$. We then apply the monotonicity property once
to get $a :: a' :: \ell \tridotn m g'$ and another time to get $a' ::
\ell \tridot \langle a \rangle$. Then we can apply the inductive
hypothesis to get $a' :: \ell \tridot g'$ and therefore $\ell \tridot
a' \supset g'$.

\item The remaining cases follow directly from the inductive hypothesis
and the results in Lemma~\ref{lem:deriv-forward}.
\end{enumerate}

\section{An Example of the Two-level Logic Reasoning Approach}
\label{sec:example-two-level}

\begin{figure}[t]
\begin{align*}
&\prog {(\eval {(\abs A R)} {(\abs A R)})} \top \triangleq \top \\
&\prog {(\eval {(\app M N)} V)}
       {(\langle \eval M (\abs A R) \rangle \land
         \langle \eval {(R\ N)} V \rangle)} \triangleq \top \\
&\prog {(\of {(\app M N)} B)}
       {(\langle \of M (\arr A B) \rangle \land
         \langle \of N A \rangle)} \triangleq \top \\
&\prog {(\of {(\abs A R)} (\arr A B))}
       {(\forall x. \of x A \supset \langle \of {(R\ x)} B\rangle)}
       \triangleq \top
\end{align*}
\caption{{\sl prog} clauses for simply-typed $\lambda$-calculus}
\label{fig:example-prog}
\end{figure}

Within this framework of the two-level logic approach to reasoning, we
come back to 
the example of evaluation and typing for the simply-typed
$\lambda$-calculus. We use the \hh specification of these notions
given in Section~\ref{sec:spec-example} which yields the {\sl prog}
clauses shown in Figure~\ref{fig:example-prog}. We can now formalize
the type preservation theorem completely in the meta-logic:
\begin{theorem}\label{thm:example-two-level}
The following formula is derivable in \logic.
\begin{equation*}
\forall e, t, v. (\,\tridot \langle \eval e v \rangle) \land
(\,\tridot \langle \of e t \rangle) \supset (\,\tridot \langle \of v t
\rangle)
\end{equation*}
\end{theorem}
\begin{proof}
The informal argument for the proof of type preservation presented in
Section~\ref{sec:spec-example} is based on strong induction over the
height of \hh derivations. We will now show how we can mimic that same
style of induction in \logic. We first generalize the formula we want
to prove to the following.
\begin{equation*}
\forall e, t, v, i, j. (\nat j) \land (\lt i j) \land (\seq i {nil}
\langle \eval e v \rangle) \land (\,\tridot \langle \of e t \rangle)
\supset (\,\tridot \langle \of v t \rangle)
\end{equation*}
If we prove this generalization, then we can use the $\cut$ rule to
bring it in as a hypothesis in a proof of the original formula. The
resulting sequent will then be easily provable. To prove the
generalization, we use induction on $\nat j$. In the case where $j =
z$, the proof is trivial since $\lt i z$ is unsatisfiable. In the
other case we have $j = s\ j'$ and we know the result holds for any
$i$ such that $\lt i j'$. In this way, we can completely handle the
strong induction within our logic.

The rest of proof of the generalization closely follows the informal
argument with only the following points worthy of note.

Case analysis on specification judgments in the informal argument is
realized in the construction of a derivation in \logic by using
$\defL$ twice. Specifically, if we want to do case analysis on a
derivation such as $\seq i {nil} \langle \eval e v \rangle$ then we
apply $\defL$ which results in two cases. The first is that $\member
{(\eval e v)} nil$ holds which is impossible. The second is that
$\exists b. \prog {(\eval e v)} b \land \seq {i'} {nil} b$ holds for
some $i'$ such that $i = s\ i'$. Then we can apply $\defL$ on $\prog
{(\eval e v)} b$ which gives us the two cases corresponding to the
clauses for forming {\sl eval} judgments.

The instantiation and cut admissibility properties of our
specification logic which are used the informal argument are now
formal lemmas which are applied in this proof. Thus the entire proof
is formally constructed within \logic while still using meta-theoretic
properties of \hh.
\end{proof}

\section{Architecture of a Two-level Logic Based Theorem Prover}
\label{sec:arch-two-level}

The architecture of the Abella theorem prover for \logic presented in
Chapter~\ref{ch:architecture} can be naturally extended to support the
two-level logic approach to reasoning that is the topic of discussion
in this current chapter. In fact, the Abella system already
incorporates such an extension \cite{gacek-abella-website}. In this
section we briefly describe the architectural changes which facilitate
this support. Most of these changes can be motivated from the type
preservation example shown in the previous section which we will refer
to as simply ``the example.''

The first step in the two-level logic approach to reasoning is
encoding a specification into the proper {\sl prog} statement. Abella
facilitates this by reading specifications written in the subset of
$\lambda$Prolog which corresponds to \hh. In this way, the
specifications used by Abella are directly executable by
$\lambda$Prolog implementations such as Teyjus without the potentially
error-prone need to translate between different input languages.

To reduce syntactic overhead associated with the two-level logic
approach to 
reasoning, Abella has specialized syntax for representing judgments of
the form $\ell \tridot g$. Direct reasoning on these judgments is
enabled by incorporating the derived rules of inference from
Section~\ref{sec:deriv-rules-infer}. Case analysis on judgments of the
form $\ell \tridot g$ in Abella corresponds to applying $\defL$ to
underlying the {\sl seq} judgment followed by applying $\defL$ to the
resulting {\sl prog} judgment. Trivial cases such as $\member E nil$
are handled automatically. Thus much of the overhead which is shown in
the example is hidden when working with Abella.

The monotonicity, instantiation, and cut-admissibility properties of
the specification logic (Section~\ref{sec:form-meta-theory}) are
incorporated into Abella in the form of tactics. Moreover, the
monotonicity property is incorporated into some other existing tactics
since it seems to be used most often. For example, when determining if
$\ell \tridot g$ implies $\ell' \tridot g$ the system checks if $\ell$
is an obvious subset of $\ell'$. Such checks arise often, for example,
when applying a lemma to hypotheses.

Abella simulates strong induction on \hh derivations using the
technique shown in the example. In general, the induction tactic
applied to a judgment of the form $\ell \tridot g$ is treated as
strong induction on the underlying measure. This is approximated using
the annotation based treatment of induction from
Section~\ref{sec:induct-co-induct} applied directly to specification
judgments. This has the benefit of removing much of the tedious
reasoning about natural numbers which would otherwise clutter a proof.
As an example of this annotation based treatment, suppose we want to
prove a formula of the form
\begin{equation*}
\forall \vec{x}.~ (\ell \tridot g) \supset F.
\end{equation*}
Then the induction scheme creates the following inductive hypothesis
and goal, respectively:
\begin{align*}
\forall \vec{x}.~ (\ell \tridot g)^* \supset F &&
\forall \vec{x}.~ (\ell \tridot g)^@ \supset F.
\end{align*}
Eventual case analysis on $(\ell \tridot g)^@$ results in recursive
judgments of the form $(\ell' \tridot g')^*$ which are subject to the
inductive hypothesis. The monotonicity and instantiation properties of
the specification logic preserve the height of \hh derivations, and
thus tactics which implement them preserve induction annotations as
well (since induction is being carried out on the underlying height
measure). Finally, suppose we want to deal with mutual induction on
specification judgments. For example, suppose we have a goal of the
form
\begin{equation*}
(\forall \vec{x}_1.~ (\ell_1 \tridot g_1) \supset F_1) \land
(\forall \vec{x}_2.~ (\ell_2 \tridot g_2) \supset F_2).
\end{equation*}
We can perform induction on both of the specification judgments
simultaneously by instead considering the following goal
\begin{equation*}
\forall n. \nat n \supset (\forall \vec{x}_1.~ (\ell_1 \tridotn n g_1)
\supset F_1) \land (\forall \vec{x}_2.~ (\ell_2 \tridotn n g_2)
\supset F_2),
\end{equation*}
and performing induction on $\nat n$. Once this new goal is proven,
the original is an easy consequence. We extend the annotation based
treatment of induction to treat this kind of mutual induction
directly. Specifically, it creates the following two inductive
hypotheses
\begin{align*}
(\forall \vec{x}_1.~ (\ell_1 \tridot g_1)^* \supset F_1) &&
(\forall \vec{x}_2.~ (\ell_2 \tridot g_2)^* \supset F_2),
\end{align*}
and the goal becomes
\begin{equation*}
(\forall \vec{x}_1.~ (\ell_1 \tridot g_1)^@ \supset F_1) \land
(\forall \vec{x}_2.~ (\ell_2 \tridot g_2)^@ \supset F_2).
\end{equation*}
The proof then proceeds as normal. When case analysis is performed on
a judgment with a $@$ annotation, the recursive calls will have the
$*$ annotation and thus be candidates for either of the inductive
hypotheses.

\section{Adequacy for the Two-level Logic Approach to Reasoning}
\label{sec:adequacy-g}

Adequacy within the framework based on the two-level logic approach to
reasoning has three components:
\begin{enumerate}
\item Our encoding of the object system into \hh must be adequate.
\item Our encoding of \hh into \logic must be adequate.
\item We must show that information about object system
properties can be extract from theorems in \logic via the two
encodings.
\end{enumerate}
The first component is particular to the object system of interest.
For example, adequacy for the \hh encoding of evaluation and typing
for the simply-typed $\lambda$-calculus was shown in
Section~\ref{sec:spec-adequacy}. In the current section we are primary
concerned with latter two components which deal with adequacy relative
to \logic. The second component is a general result about \hh and its
encoding in the predicate {\sl seq} (we shall often call this simply
``the adequacy of {\sl seq}''). The proof of this result is carried
out in the next subsection, and it never needs to be changed since \hh
and {\sl seq} are fixed. The last component of adequacy is particular
to the theorems of interest, and in Section~\ref{sec:adequacy-seq-sr}
we show this adequacy for the example of type preservation for the
simply-typed $\lambda$-calculus.

There is some difficulty in establishing adequacy relative to \logic.
When we represent objects in \logic we usually denote bound variables
using $\lambda$-terms and free variables using nominal constants.
Then, when we quantify over such objects, we are usually interested
only in objects whose free variables are restricted to a particular
set (\eg, we may care only about closed objects). The $\forall$ and
$\exists$ quantifiers of \logic, however, allow nominal constants to
appear freely in the terms that instantiate them. There are two ways
to address this mismatch (without modifying the logic \logic). The
first is to define an explicit typing of objects (\eg, through a
predicate $\hsl{typeof}\ L\ T\ A$ where $L$ is a context of nominal
constants), and to attach this typing judgment to all quantified
variables. This is a very heavy approach and requires explicitly
maintaining a context of which nominal constants are allowed to appear
in objects. An alternative approach, and the one we use to establish
the adequacy of {\sl seq} in the next subsection, is to restrict the
use of nominal constants in such a way that adequacy can still be
established. How exactly this is done depends on the particular system
of interest and how nominal constants are treated by it. In the case
of {\sl seq} we know that nominal constants can always be
instantiated, thus the only restriction we need is that nominal
constants are allowed only at inhabited types.

\subsection{Adequacy of Encoding of the Specification Logic}
\label{sec:adequacy-seq}

We now show that our encoding of the specification logic \hh in the
definition of {\sl seq} and {\sl prog} is adequate. The critical
aspect of this result is showing that theoremhood in the two systems
is preserved under an appropriate mapping.

\begin{theorem}\label{thm:seq-adequacy}
Let $\Delta$ be a list of closed $D$-formulas, $\mathcal{L}$ a list of
atoms, $G$ a $G$-formula, and $\Sigma$ a set of eigenvariables
containing at least the free variables of $\Delta$, $\mathcal{L}$, and
$\mathcal{G}$. Suppose that all non-logical specification logic
constants and types are represented by equivalent constants and types
in \logic. Suppose also that specification logic
$\forall$-quantification (eigenvariables) and meta-logic
$\nabla$-quantification (nominal constants) are allowed only at
inhabited types. Then $\Sigma : \Delta, \mathcal{L} \vdash G$ has a
derivation in \hh if and only if $\psi(\mathcal{L}) \tridot \psi(G)$
is provable in \logic with the clauses for {\sl nat}, {\sl member},
and {\sl seq} as stated before and the clauses for {\sl prog} as given
by $\Psi(\Delta)$.
\end{theorem}
\begin{proof}
Note that in this proof we will desugar the representation of
quantification and substitution in the specification logic.

{\bf Forward direction.} Given a derivation of $\Sigma :
\Delta, \mathcal{L} \vdash G$ in \hh, we will construct a proof of
$\psi(\mathcal{L}) \tridot \psi(G)$ in \logic. The construction uses
structural induction on the \hh derivation and proceeds by cases on
the last rule used in the derivation.
\begin{enumerate}
\item Suppose the derivation ends with $\OR_1$:
\begin{equation*}
\infer[\OR_1]
      {\Sigma : \Delta, \mathcal{L} \vdash G_1 \lor G_2}
      {\Sigma : \Delta, \mathcal{L} \vdash G_1}
\end{equation*}
By the inductive hypothesis we know $\psi(\mathcal{L}) \tridot
\psi(G_1)$ is provable in \logic. Then we know $\psi(\mathcal{L})
\tridot \psi(G_1 \lor G_2)$ using the appropriate formula from
Lemma~\ref{lem:deriv-forward}.

\item Suppose the derivation ends with $\TRUE$, $\OR_2$, $\AND$, or
$\AUGMENT$: these cases are similar to the previous one.

\item Suppose the derivation ends with $\GENERIC$:
\begin{equation*}
\infer[\GENERIC]
      {\Sigma : \Delta, \mathcal{L} \vdash \forall G'}
      {\Sigma, c : \Delta, \mathcal{L} \vdash G'\ c}
\end{equation*}
By the inductive hypothesis we know $\psi(\mathcal{L}) \tridot
\psi(G'\ c)$ is provable in \logic. We also know $\psi(G'\ c) =
\psi(G')\ a_c$ where $a_c$ is a nominal constant not in $\psi(\Sigma)$
(and therefore not occurring in $\psi(\mathcal{L})$ or $\psi(G')$).
Thus we know there is a proof of $\nabla x. (\psi(\mathcal{L}) \tridot
(\psi(G')\ x))$. Using the appropriate formula from
Lemma~\ref{lem:deriv-forward}, there must be a proof of
$\psi(\mathcal{L}) \tridot \forall \psi(G')$.

\item Suppose the derivation ends with $\INSTANCE$:
\begin{equation*}
\infer[\INSTANCE]
      {\Sigma : \Delta, \mathcal{L} \vdash \exists_\tau G'}
      {\Sigma : \Delta, \mathcal{L} \vdash G'\ t}
\end{equation*}
By the inductive hypothesis we know $\psi(\mathcal{L}) \tridot
\psi(G'\ t)$ is provable in \logic. We also know $\psi(G'\ t) =
\psi(G')\ \psi(t)$. Using the appropriate formula from
Lemma~\ref{lem:deriv-forward}, there must be a proof of
$\psi(\mathcal{L}) \tridot \exists \psi(G')$.

\item Suppose the derivation ends with $\BACKCHAIN$:
\begin{equation*}
\infer[\BACKCHAIN]
      {\Sigma : \Delta, \mathcal{L} \vdash A}
      {\Sigma : \Delta, \mathcal{L} \vdash G_1\ \vec{t} &
       \cdots &
       \Sigma : \Delta, \mathcal{L} \vdash G_m\ \vec{t}}
\end{equation*}
where $\forall \vec{x} . (G_1\ \vec{x} \supset \cdots \supset G_m\
\vec{x} \supset A'\ \vec{x}) \in \Delta, \mathcal{L}$ and $A'\ \vec{t}
= A$. We distinguish two cases based on whether the formula is in
$\Delta$ or in $\mathcal{L}$.
\begin{enumerate}
\item Suppose $\forall \vec{x} . (G_1\ \vec{x} \supset \cdots \supset
G_m\ \vec{x} \supset A'\ \vec{x}) \in \Delta$. Then we must have the
following clause.
\begin{equation*}
\forall \vec{x}. \prog {(\psi(A')\ \vec{x})} (\psi(G_1)\ \vec{x} \land
\cdots \land \psi(G_m)\ \vec{x}) \triangleq \top
\end{equation*}
By the inductive hypothesis we have a proof of $\psi(\mathcal{L})
\tridot \psi(G_i\ \vec{t})$ for each $i \in \{1,\ldots,m\}$. By
repeatedly using the appropriate formula from
Lemma~\ref{lem:deriv-forward} we can construct a proof of
$\psi(\mathcal{L}) \tridot (\psi(G_1\ \vec{t}) \land \cdots \land
\psi(G_m\ \vec{t}))$, which we can write as $\psi(\mathcal{L}) \tridot
(\psi(G_1)\ \overrightarrow{\psi(t)}) \land \cdots \land \psi(G_m)\
\overrightarrow{\psi(t)})$. Finally we know $\psi(A) = \psi(A'\
\vec{t}) = \psi(A')\ \overrightarrow{\psi(t)}$. Thus we know $\exists
b. \prog {\psi(A)} b \land (\psi(\mathcal{L}) \tridot b)$ and we can
construct a proof of $\psi(\mathcal{L}) \tridot \langle
\psi(\mathcal{A}) \rangle$.

\item Suppose $\forall \vec{x} . (G_1\ \vec{x} \supset \cdots \supset
G_m\ \vec{x} \supset A'\ \vec{x}) \in \mathcal{L}$. Since $\mathcal{L}$
contains only atoms we must have $A = A'$ and thus $A \in \mathcal{L}$.
Then $\member {\psi(A)} {\psi(\mathcal{L})}$ is provable and thus so
is $\psi(\mathcal{L}) \tridot \langle \psi(\mathcal{A}) \rangle$.
\end{enumerate}
\end{enumerate}

{\bf Backward direction.} It suffices to show if $\nat (s\ n)$ and
$\seq {(s\ n)} {\psi(\mathcal{L})} {\psi(G)}$ have cut-free proofs in
\logic, then we can construct a derivation of $\Sigma : \Delta,
\mathcal{L} \vdash G$ in \hh for any $\Sigma$ which contains at least
the eigenvariables of $\mathcal{L}$ and $G$. The proof is by induction
on the natural number denoted by $(s\ n)$ (which we know is a natural
number since $\nat (s\ n)$ has a proof). This proof will always end
with $\defR^p$ (or can be seen to) and we will consider cases based on
the definitional clause used in this rule.
\begin{enumerate}
\item The cases for the first five clauses of {\sl seq} are all
similar and thus we will consider just one instance. Suppose the
cut-free proof ends with,
\begin{equation*}
\infer[\defR^p]
      {\lra \seq {(s\ n)} {\psi(\mathcal{L})} {(\psi(G_1) \lor
          \psi(G_2))}}
      {\lra \seq n {\psi(\mathcal{L})} {\psi(G_1)}}
\end{equation*}
By the inductive hypothesis we know there is a derivation of $\Sigma :
\Delta, \mathcal{L} \vdash G_1$ and we can construct the following.
\begin{equation*}
\infer[\OR_1]
      {\Sigma : \Delta, \mathcal{L} \vdash G_1 \lor G_2}
      {\Sigma : \Delta, \mathcal{L} \vdash G_1}
\end{equation*}

\item Suppose the cut-free proof ends with,
\begin{equation*}
\infer[\defR^p]
      {\lra \seq {(s\ n)} {\psi(\mathcal{L})} {(\forall \psi(G'))}}
      {\infer[\nablaR]
             {\lra \nabla x. \seq n {\psi(\mathcal{L})} {(\psi(G')\ x)}}
             {\lra \seq n {\psi(\mathcal{L})} {(\psi(G')\ a)}}}
\end{equation*}
Since $\psi(G')\ a = \psi(G'\ h_a)$ we know from the inductive
hypothesis that there is a derivation of $\Sigma, h_a : \Delta,
\mathcal{L} \vdash G'\ h_a$. Thus we can construct the following.
\begin{equation*}
\infer[\GENERIC]
      {\Sigma : \Delta, \mathcal{L} \vdash \forall G'}
      {\Sigma, h_a : \Delta, \mathcal{L} \vdash G'\ h_a}
\end{equation*}

\item Suppose the cut-free proof ends with,
\begin{equation*}
\infer[\defR^p]
      {\lra \seq {(s\ n)} {\psi(\mathcal{L})} {(\exists_\tau \psi(G'))}}
      {\infer[\existsR]
             {\lra \exists_\tau x. \seq n {\psi(\mathcal{L})} {(\psi(G')\ x)}}
             {\mathcal{C}, \mathcal{K} \vdash t : \tau &
              \lra \seq n {\psi(\mathcal{L})} {(\psi(G')\ t)}}}
\end{equation*}
Now $t$ may contain any nominal constants and therefore $t' =
\psi^{-1}(t)$ may contain eigenvariables not in $\Sigma$. Thus when we
apply the inductive hypothesis to $\seq n {\psi(\mathcal{L})}
{\psi(G'\ t')}$ we get a derivation of $\Sigma' : \Delta, \mathcal{L}
\vdash G'\ t'$ where $\Sigma'$ may contain additional eigenvariables.
To reconcile this, we use the restriction that eigenvariables
are allowed only at inhabited types. For each eigenvariable in $t'$ and
not in $\Sigma$, we select an inhabitant of the corresponding type and
substitute it for the eigenvariable using the instantiation property
of \hh. Since these eigenvariables do not occur in $\Sigma$, they also
do not occur in $\mathcal{L}$ or $G$ and therefore the instantiations
affect only $t'$. Thus the result of all these instantiations is a
derivation of $\Sigma : \Delta, \mathcal{L} \vdash G'\ t''$ for some
$t''$. Then we can construct the following.
\begin{equation*}
\infer[\INSTANCE]
      {\Sigma : \Delta, \mathcal{L} \vdash \exists G'}
      {\Sigma : \Delta, \mathcal{L} \vdash G'\ t''}
\end{equation*}

\item Suppose the cut-free proof ends with,
\begin{equation*}
\infer[\defR^p]
      {\lra \seq {(s\ n)} {\psi(\mathcal{L})} {\langle \psi(A) \rangle}}
      {\lra \member {\psi(A)} \psi(\mathcal{L})}
\end{equation*}
Then it must be that $A \in \mathcal{L}$, and so we can construct the
following.
\begin{equation*}
\infer[\BACKCHAIN]
      {\Sigma : \Delta, \mathcal{L} \vdash A}
      {}
\end{equation*}

\item Suppose the cut-free proofs ends with,
\begin{equation*}
\infer
 [\defR^p]
 {\lra \seq {(s\ n)} {\psi(\mathcal{L})} {\langle \psi(A) \rangle}}
 {\infer
   [\existsR]
   {\lra \exists b. \prog {\psi(A)} b \land \seq n {\psi(\mathcal{L})}
     b}
   {\infer
     [\landR]
     {\lra \prog {\psi(A)} b \land \seq n {\psi(\mathcal{L})} b}
     {\infer[\defR^p]{\lra \prog {\psi(A)} b}{} &
      \lra \seq n {\psi(\mathcal{L})} b}}}
\end{equation*}
for some instantiation of $b$. Suppose also that $\prog {\psi(A)} b$
holds by matching with some clause,
\begin{equation*}
\forall \vec{x}. \prog {(\psi(A')\ \vec{x})} (\psi(G_1)\ \vec{x} \land
\cdots \land \psi(G_m)\ \vec{x}) \triangleq \top.
\end{equation*}
Then we know $\forall \vec{x} . (G_1\ \vec{x} \supset \cdots \supset
G_m\ \vec{x} \supset A'\ \vec{x}) \in \Delta$. From matching with the
{\sl prog} clause we know there exists $\vec{t}$ such that $\psi(A) =
\psi(A')\ \vec{t}$, so let $\vec{s} = \psi^{-1}(\vec{t})$. Then $b$ is
$\psi(G_1\ \vec{s}) \land \cdots \land \psi(G_m\ \vec{s})$ and we have
proofs of $\seq n {\psi(\mathcal{L})} \psi(G_i\ \vec{s})$ for each $i
\in \{1,\ldots,m\}$. By the inductive hypothesis we have derivations
of $\Sigma' : \Delta, \mathcal{L} \vdash G_i\ \vec{s}$ where $\Sigma'$
contains the eigenvariables of $\mathcal{L}, G_1, \ldots, G_m$, and
$\vec{s}$. Note that as was the case for the {\sl seq} rule governing
the existential quantifier, $\Sigma'$ may contain some eigenvariables
from $\vec{s}$ which do not occur in $\Sigma$. As with that case, we
can use the restriction on specification logic eigenvariables to
instantiate all such eigenvariables with inhabitants therefore
yielding derivations $\Sigma : \Delta, \mathcal{L} \vdash G_i\
\vec{r}$ where $\vec{r}$ is the result of the instantiations on
$\vec{s}$. Finally, we know $A = A'\ \vec{s}$ but we need to know $A =
A'\ \vec{r}$. Note that $A'$ contains no eigenvariables and the
eigenvariables of $A$ are a subset of $\Sigma$, thus the
eigenvariables in $\vec{s}$ but not in $\Sigma$ play no role in the
equality $A = A'\ \vec{s}$. Therefore instantiating those
eigenvariables does not change the equality and we have $A = A'\
\vec{r}$. Thus we can construct the following.
\begin{equation*}
\infer[\BACKCHAIN]
      {\Sigma : \Delta, \mathcal{L} \vdash A}
      {\Sigma : \Delta, \mathcal{L} \vdash G_1[\vec{r}/\vec{x}] &
       \cdots &
       \Sigma : \Delta, \mathcal{L} \vdash G_m[\vec{r}/\vec{x}]} \qedhere
\end{equation*}
\end{enumerate}

\end{proof}

Note that this theorem restricts the definitions of the predicates
{\sl nat}, {\sl member}, {\sl seq}, and {\sl prog}, but makes no
explicit reference to other predicates. Indeed, the definitions of
other predicates have no affect on the adequacy of the encoding of the
specification logic. Additionally, \logic may make use of additional
constants and types which are unconnected to the constants and types
used to represent the specification logic without affecting the
adequacy of the encoding.

Another point of interest is the following condition of the previous
theorem: specification logic $\forall$-quantification and meta-logic
$\nabla$-quantification are allowed only at inhabited types. This
condition arises because we have chosen to do a shallow encoding of
the typing judgment of the specification logic. That is, rather than
encode an explicit typing judgment for specification logic terms, we
have instead relied on the typing judgment of \logic to enforce the
well-formedness of terms. Due to the lack of restrictions on the
occurrences of nominal constants, the typing judgment in \logic is more
permissive than the specification logic typing. As the previous
theorem shows, however, this difference only manifests itself for
uninhabited types. A deeper encoding involving an explicit typing
judgment would avoid this condition, but would also impose some
overhead additional costs in terms of reasoning about and through the
encoding. We find the shallow encoding to be a good balance in
practice.

\subsection{Adequacy of Type Preservation Example}
\label{sec:adequacy-seq-sr}

We can now use our adequacy results to extract a proof of type
preservation for the simply-typed $\lambda$-calculus from the proof of
its encoding in \logic.

\begin{theorem}
\label{thm:sr-roundtrip}
If $t \Downarrow v$ and $\vdash t : a$ then $\vdash v : a$.
\end{theorem}
\begin{proof}
Suppose $t \Downarrow v$ and $\vdash t : a$, then by the adequacy
results in Section~\ref{sec:spec-adequacy}, we know that $\Delta
\vdash \eval {\phi(t)} {\phi(v)}$ and $\Delta \vdash \of {\phi(t)}
{\phi(a)}$ have derivations in \hh where $\phi$ is the bijection
between the object language and its specification logic representation
and $\Delta$ is the specification of {\sl eval} and {\sl of}. By
Theorem~\ref{thm:seq-adequacy}, we know $\ \tridot \langle \eval
{\psi(\phi(t))} {\psi(\phi(v))} \rangle$ and $\ \tridot \langle \of
{\psi(\phi(t))} {\psi(\phi(v))} \rangle$ have proofs in \logic. Using
these proofs and the proof of the formula in
Theorem~\ref{thm:example-two-level} together with various rules of
\logic (notably the {\sl cut} rule), we can construct a proof of $\
\tridot \langle \of {\psi(\phi(v))} {\psi(\phi(a))} \rangle$ in
\logic. Then using the backwards direction of
Theorem~\ref{thm:seq-adequacy} we know $\Delta \vdash \of {\phi(v)}
{\phi(a)}$ has a derivation in \hh, and using adequacy results from
Section~\ref{sec:spec-adequacy} we find that $\vdash v : a$ must hold.
\end{proof}



\chapter{Applications of The Framework}
\label{ch:applications}

In this chapter we consider various applications of the proposed
framework, focusing mainly on the reasoning component. The purpose of
these applications is illustrate both the strengths and the weaknesses
of the framework. From this perspective, we are interested in the {\em
  quality} of the encodings and associated reasoning, \eg, properties
such as naturalness, expressiveness, complexity, and overhead. We will
try to expose and highlight these traits in this chapter.

We begin in Section~\ref{sec:type-uniq-simply} with a proof of type
uniqueness for the simply-typed $\lambda$-calculus which provides a
simple example of how judgment contexts and the related variable
freshness information is handled in the framework. In
Section~\ref{sec:poplmark-challenge} we present a solution to part of
the POPLmark challenge \cite{aydemir05tphols} which demonstrates the
more sophisticated inductive reasoning that is possible within \logic.
Section~\ref{sec:path-equiv-lambda} contains an example of proving the
equivalence of $\lambda$-terms based on the set of paths they contain,
and shows how easily the framework handles formulas with a more
sophisticated quantification structure. In
Section~\ref{sec:conv-de-bruijn} we describe a translation between
higher-order abstract syntax and de Bruijn notation for
$\lambda$-terms, and we show that this translation is deterministic in
both directions. This example highlights a more expressive use of
definitions to describe the structure of judgment contexts. Finally,
in Section~\ref{sec:girards-strong-norm} we show how Girard's proof of
strong normalization for the simply-typed $\lambda$-calculus can be
encoded. This is by far the largest application in this chapter, and
it uses many of the features highlighted by previous examples as well
as introducing new ones such as a way of dealing with an arbitrary
number of substitutions applied to a term.

There have been many other applications of the reasoning component of
our framework that we do not discuss explicitly in this thesis. These
include the following.
\begin{itemize}
\item Properties of big and small step evaluation and typing in the
simply-typed $\lambda$-calculus
\item Translation among combinatory logic, natural deduction, and
sequent calculus
\item Soundness and completeness for a focused sequent calculus
\item Cut-admissibility for LJ
\item Takahashi's proof of the Church-Rosser theorem
\item Properties of bi-simulation in CCS and the $\pi$-calculus
\item Tait's argument for weak normalization of the simply-typed
$\lambda$-calculus \cite{gacek08lfmtp}.
\item The substitution theorem for Canonical LF.
\end{itemize}
All of the applications mentioned above and the ones presented in this
chapter are available on the Abella website
\cite{gacek-abella-website}. We note that some of these examples have
been developed by other researchers. Randy Pollack contributed the
formalization of the Church-Rosser result. The formalization of the
substitution theorem for Canonical LF was contributed by Todd Wilson
and is the largest development done in Abella to date. This
development includes two sophisticated results: one which uses a
triply nested induction where the innermost induction is an eight-way
mutual induction and another which uses a doubly nested induction with
an outer strong induction and an inner three-way mutual induction. The
richness and elegance of this development serves as a powerful example
of the expressivity of Abella.

Finally, before we proceed to the examples we establish a few common
items and conventions which simplify the presentation. First, in
specification formulas we elide the outermost universal quantifiers
and assume that tokens given by capital letters denote variables that
are implicitly universally quantified over the entire formula.
Second, for judgments of the form $(L\tridot \langle A \rangle)$ we
write simply $(L\tridot A)$ since we will only ever display this with
atomic formulas on the right of the judgment. We assume the following
definition of {\sl name} (with appropriate type based on the
application):
\begin{equation*}
(\nabla x. \name x) \triangleq \top.
\end{equation*}
We will use the following result about the (non)occurrences of nominal
constants in lists:
\begin{equation*}
\forall L, E.\nabla x.~\member {(E\ x)} L \supset \exists E'.~ (E =
\lambda y.E').
\end{equation*}
This says that if an element of a list depends on a nominal constant
and the list itself does not, then the element's dependency must be
vacuous. The proof is by induction on the {\sl member} hypothesis. We
will leave out the details of most proofs except to note the uses of
induction or the particularly interesting cases. Also, we will freely
and implicitly make use of the properties of the specification logic.

\section{Type-uniqueness for the Simply-typed $\lambda$-calculus}
\label{sec:type-uniq-simply}

The type of a $\lambda$-term in the simply-typed $\lambda$-calculus is
unique. Proving this type uniqueness property requires reasoning
inductively about typing judgments which, in turn, requires
generalizing the context in which typing judgments are made. We can
encode such arguments directly in our framework so long as we can
describe the structure of the judgment contexts. Such descriptions can
be naturally expressed using nominal abstraction and, in fact, this is
the most common use of nominal abstraction. Thus, we use the present
example to demonstrate how nominal abstraction can be used in this way
and to point out the related lemmas that often go along with such
descriptions.

\begin{figure}[t]
\begin{tabbing}
\hspace{3cm}\= $\ctx nil \mueq \top$ \\
\> $\ctx (\of X A :: L) \mueq\null$\=
$(\forall M, N.~ X = \app M N \supset \bot) \land\null$ \\
\>\> $(\forall R, B .~ X = \abs B R \supset \bot) \land\null$ \\
\>\> $(\forall B.~ \member {(\of X B)} L \supset \bot) \land \null$ \\
\>\> $\ctx L$
\end{tabbing}
\caption{Potential {\sl ctx} definition without nominal abstraction}
\label{fig:ctx-lg}
\end{figure}

We will use the specification of the simply-typed $\lambda$-calculus
developed thus far in the thesis (Section~\ref{sec:spec-example}).
Relative to this, we can formally state type uniqueness as
\begin{equation*}
\forall E, T_1, T_2.~ (\,\tridot \of E T_1) \supset (\,\tridot \of E
T_2) \supset (T_1 = T_2).
\end{equation*}
Suppose we try to prove this directly by induction on one of the
typing judgments. Then, when we consider the case where $E$ is an
abstraction, the typing context will grow which means the inductive
hypothesis will not be able to apply. Instead, we need to generalize
the statement of type uniqueness to the following.
\begin{equation*}
\forall L, E, T_1, T_2.~ \ctx L \supset (L \tridot \of E T_1) \supset
(L \tridot \of E T_2) \supset (T_1 = T_2).
\end{equation*}
Where {\sl ctx} is a definition which restricts $L$ so that the
formula is provable. In particular, $\ctx L$ should enforce that $L$
has the structure $(x_1, A_1)::\ldots::(x_n, A_n)::nil$ where each
$x_i$ is atomic and unique. In the logics which preceded \logic, these
atomicity and uniqueness properties could not be directly described
and instead one needed to encode them by explicitly excluding the other
possibilities as shown in Figure~\ref{fig:ctx-lg}. However, using
nominal abstraction we define {\sl ctx} as
\begin{align*}
\ctx nil \mueq \top && (\nabla x. \ctx (\of x A :: L)) \mueq
\ctx L.
\end{align*}
Note that in $(\of x A :: L)$, the atomicity of $x$ is enforced by it
being $\nabla$ quantified while the uniqueness is enforced by $L$
being quantified outside the scope of $x$. Had we wanted to allow $x$
to occur later in the context we could have written $(L\ x)$ in place
of $L$.

The definition of {\sl ctx} enforces atomicity and uniqueness
properties for the first element of the context and then calls itself
recursively on the remaining portion of the context. Thus, to know
that an arbitrary element of the context has the atomicity and
uniqueness properties requires inductive reasoning. We state these
properties in the following two lemmas.
\begin{align*}
& \forall L, X, A.~ \ctx L \supset \member {(\of X A)} L \supset \name
X \\
& \forall L, X, A_1, A_2.~ \ctx L \supset \member {(\of X A_1)} L
\supset \member {(\of X A_2)} \supset (A_1 = A_2)
\end{align*}
Both of these lemmas have direct proofs using induction on one
of the {\sl member} hypotheses.

With the above lemmas in place, the rest of the type uniqueness proof
is straightforward. There is an interesting point to be noted here,
though, concerning the treatment of abstractions, \ie, when
considering the typing in the context $L$ of a $\lambda$-term of the
form $\abs A R$. The use of a universal
quantifier in the specification of typing in this case and the
interpretation in the meta-logic of such universal quantifiers via
$\nabla$-quantifiers ensures that the typing of $R$ will be done in a
context given by $\of x A :: L$ where $x$ is a nominal constant not
appearing in $L$.
In the type uniqueness proof, we will need to show that this extended
typing context is well-formed. This is done by showing that $\ctx (\of
x A :: L)$ follows from $\ctx L$ which is clear based on the
definition of {\sl ctx} and the way
$x$ was
introduced in the typing process. If a definition such as in
Figure~\ref{fig:ctx-lg} were used, this argument would be more
complicated.

\section{The POPLmark Challenge}
\label{sec:poplmark-challenge}

The POPLmark challenge is a call to researchers to develop tools and
methodologies for animating and for reasoning about systems with
binding \cite{aydemir05tphols}. The particular challenge proposed
focuses on System $F_{\fsub}$, a polymorphic $\lambda$-calculus with
subtyping \cite{cardelli94ic,curien94taoop}. This challenge is of
interest to us primarily because it provides a common benchmark on
which various frameworks may be compared. In addition, some of the
reasoning required for this problem illustrates the sophistication and
naturalness of the reasoning tools available in our framework.

The POPLmark challenge consists of three challenge problems which
focus on 1) the type system, 2) evaluation, type preservation, and
progress, and 3) animation. In this section we explain the solution to
the first challenge problem which requires sophisticated induction
schemes and some reasoning about binding structure. The second
challenge problem requires a significant amount of reasoning about
binding structure, but since we take binding as fundamental in our
framework, this challenge problem is straightforward and fairly
mundane in our framework (the development is available on the Abella
website). Finally, the last challenge problem could be
addressed through an animation system for $\lambda$Prolog, but we do
not explore this in this section. The first and second challenge
problems also have an additional component that asks for proofs to be
repeated for System $F_{\fsub}$ extended with records and patterns.
This extension requires a significant amount of additional work
without providing much additional insight in the framework, and thus
we do not pursue this extension.

The first POPLmark challenge problem focuses on the type system of
System $F_{\fsub}$.  In particular, given an algorithmic presentation of
the subtyping rules for System $F_{\fsub}$, the challenge asks one to
show that the subtyping relation is reflexive and transitive, the key
results needed to show equivalence between the algorithmic and
declarative descriptions of subtyping.
Reflexivity turns out to be straightforward, while transitivity
requires sophisticated inductive reasoning. In the rest of this
section we focus on the proof of transitivity.

\begin{figure}[t]
\begin{center}
\begin{tabular}{@{\hspace{2cm}}c@{\hspace{3cm}}l}
$\Gamma \vdash S ~\fsub~ {\tt Top}$ &
(SA-Top) \\[16pt]
$\Gamma \vdash X ~\fsub~ X$ &
(SA-Refl-TVar) \\[16pt]
\raisebox{-6pt}{
$\infer{\Gamma \vdash X ~\fsub~ T}
{X\fsub U \in \Gamma & \Gamma \vdash U ~\fsub~ T}$
} &
(SA-Trans-TVar) \\[16pt]
\raisebox{-6pt}{
$\infer{\Gamma \vdash S_1 \to S_2 ~\fsub~ T_1 \to T_2} {\Gamma \vdash
  T_1 ~\fsub~ S_1 & \Gamma \vdash S_2 ~\fsub~ T_2}$
} &
(SA-Arrow) \\[16pt]
\raisebox{-6pt}{
$\infer{\Gamma \vdash (\forall X \fsub S_1.~ S_2) ~\fsub~ (\forall X
  \fsub T_1.~ T_2)} {\Gamma \vdash T_1 ~\fsub~ S_1 & \Gamma, X \fsub
  T_1 \vdash S_2 ~\fsub~ T_2}$
} &
(SA-All)
\end{tabular}
\end{center}
\caption{Algorithmic subtyping rules for System $F_{\fsub}$}
\label{fig:fsub}
\end{figure}

Types and typing contexts in System $F_{\fsub}$ are described by the
following grammars.
\begin{align*}
T &::= X \mid {\tt Top} \mid T \to T \mid \forall X \fsub T.~T \\
\Gamma &::= \emptyset \mid \Gamma, X \fsub T
\end{align*}
Here $X$ denotes a variable occurrence, and $\forall X \fsub T_1.~ T_2$
denotes that the variable $X$ is bound within the scope of $T_2$ (but
not in the scope of $T_1$). In $\Gamma, X \fsub T$ it is assumed that $X$
does not occur in $\Gamma$. The algorithmic subtyping relation of
System $F_{\fsub}$ is denoted by $\Gamma \vdash S ~\fsub~ T$, and is
defined by the rules in Figure~\ref{fig:fsub}.

The challenge problem is to prove that the subtyping relation is
transitive: if $\Gamma \vdash S ~\fsub~ Q$ and $\Gamma \vdash Q
~\fsub~ T$ then $\Gamma \vdash S ~\fsub~ T$. The proof of this
property requires another result called narrowing to be proved
simultaneously: if $\Gamma, X\fsub Q, \Delta \vdash M ~\fsub~ N$ and
$\Gamma \vdash P ~\fsub~ Q$ then $\Gamma, X \fsub P, \Delta \vdash M
~\fsub~ N$. The proof of these two properties requires a mutual
induction on the structure of the type $Q$. Within this induction the
transitivity property is proved by induction on the structure of
$\Gamma \vdash S ~\fsub~ Q$ and it uses the narrowing property for
structurally smaller types $Q$. The narrowing property is proved by an
inner induction on the structure of $\Gamma, X\fsub Q, \Delta \vdash M
~\fsub~ N$ and uses the transitivity property for the type $Q$. With
the proper induction schemes as described, the details of the proof
are straightforward.

To formalize System $F_{\fsub}$ types we introduce the type $ty$ and the
following constants.
\begin{align*}
\hsl{top} : ty && \hsl{arrow} : ty \to ty \to ty && \hsl{all} : ty
\to (ty \to ty) \to ty
\end{align*}
Typing contexts will be represented using the context of specification
logic judgments. We introduce the constant $bound : ty \to ty \to o$
for representing individual type bindings within that context.

\begin{figure}[t]
\begin{align*}
&\sub S \ftop \\
&\bound X U \supset \sub X X \\
&\bound X U \supset \sub U T \supset \sub X T \\
&\sub {T_1} {S_1} \supset \sub {S_2} {T_2} \supset
\sub {(\arrow {S_1} {S_2})} {(\arrow {T_1} {T_2})} \\
&\sub {T_1} {S_1} \supset (\forall x. \bound x T_1 \supset \sub {(S_2\
  x)} {(T_2\ x)}) \supset \sub {(\all {S_1} {S_2})} {(\all {T_1}
  {T_2})}
\end{align*}
\caption{Specification of algorithmic subtyping for System $F_{\fsub}$}
\label{fig:spec-fsub}
\end{figure}

We encode subtyping rules of System $F_{\fsub}$ as specification logic
formulas concerning the constant $sub : ty \to ty \to o$ as presented in
Figure~\ref{fig:spec-fsub}. Note that we do not explicitly represent the
typing context, but instead make assumptions of the form $\bound X T$
to denote a typing assumption of $X \fsub T$. Also, in the formal
rules {\sc SA-Refl-TVar} and {\sc SA-Trans-TVar} the variable $X$
represents only type variables while our translation of these rules
do not directly enforce this constraint. Instead, our translations
require that any such $X$ satisfy a $\bound X U$ judgment for some
$U$. Since we only make such judgments for $X$ which denotes a type
variable, our encoding remains adequate.

To reason about subtyping we first formalize the notion that a typing
context is well-formed. Strictly speaking, a context is well-formed if
it is either $\emptyset$ or $\Gamma, X\fsub T$ where $X$ is a variable
which does not occur in $\Gamma$. For reasons we discuss later, we
deliberately weaken this notion and require only that $X$ is a
variable. We recognized such well-formed contexts with the following
definition.
\begin{align*}
\ctx nil \mueq \top && \ctx (\bound X U :: L) \mueq \name X
\land \ctx L
\end{align*}
We also prove the following associated lemma.
\begin{align*}
&\forall E, L.~ \ctx L \supset \member E L \supset \exists X, U.~ (E =
\bound X U) \land \name X
\end{align*}
This is proved by a simple induction on the {\sl member} hypothesis.

The logic \logic allows for induction only on definitions and not on
terms. Thus to induct on the structure of a System $F_{\fsub}$ type we
must create a definition which recognizes such types. We define a
predicate $\hsl{wfty} : ty \to o$ as follows.
\begin{align*}
\wfty \ftop &\mueq \top \\
(\nabla x. \wfty x) &\mueq \top \\
\wfty (\arrow {T_1} {T_2} )&\mueq \wfty T_1 \land \wfty T_2 \\
\wfty (\all {T_1} {T_2} )&\mueq \wfty T_1 \land \nabla x. \wfty
(T_2\ x)
\end{align*}
Induction on $\wfty Q$ will correspond to structural induction on the
type $Q$ as needed. Note that we could impose additional
well-formedness constraints which restrict variable occurrences
relative to some context of type variables, but such restrictions are
unnecessary for the proof at hand.

We can state the combined transitivity and narrowing property as
follows.
\begin{tabbing}
\hspace{1cm}\=$\forall Q.~ \wfty Q \supset$ \\
\>\hspace{1cm}\=
$(\forall L, S, T.~ \ctx L \supset (L \tridot \sub S Q) \supset (L
\tridot \sub Q T) \supset (L \tridot \sub S T)) \land~$ \\
\>\> $(\forall L, P, X, M, N.~ \ctx (\bound X Q :: L) \supset (L \tridot
\sub P Q) \supset$ \\
\>\>\hspace{2cm} $(\bound X Q :: L \tridot \sub M N) \supset
(\bound X P :: L \tridot \sub M N))$
\end{tabbing}
The proof is by an outer induction on $\wfty Q$. To prove the inner
conjunction we use the following derived rule of \logic.
\begin{equation*}
\infer[\landR^*]
      {\Gamma \lra B \land C}
      {\Gamma \lra B & \Gamma, B \lra C}
\end{equation*}
This rule is clearly admissible using $\cut$ and $\landR$. We use this
rule with $B$ as the transitivity result for the type $Q$ and $C$ as
the narrowing result for the type $Q$. Thus this rule allows us to use
the transitivity result for the type $Q$ while proving the
corresponding narrowing result. Once this is applied we can prove
transitivity using a further induction on $(L \tridot \sub S Q)$ and
narrowing using a further induction on $(\bound X Q :: L \tridot \sub
M N)$. The reasoning which remains is straightforward.

Notice that in the original statement of narrowing, the distinguished
typing assumption $X\fsub Q$ is taken from the middle of the typing
context, while in our formalized statement we consider the assumption
$\bound X Q$ only at the front. By formalizing narrowing in this way,
we greatly simplify the associated reasoning (\eg, we do not need to
talk about appending contexts as we would with a direct statement).
The cost is that when we add other elements to the context, we must
show that the distinguished binding can always be moved to the front.
This is possible since we have weakened the {\sl ctx} judgment to not
contain any freshness information, and therefore no ordering
information. Since freshness information is not relevant to the
transitivity and narrowing results, there is no cost to leaving this
information out. To establish adequacy, we can use a more precise
description of typing contexts and still make use of these results
proved for the looser description.

\section{Path Equivalence for $\lambda$-terms}
\label{sec:path-equiv-lambda}

\begin{figure}[t]
\centering
\begin{tikzpicture}[sibling distance=2cm]
\tikzstyle{edge from parent}=[draw,solid,thick] \node {$\lambda x$}
child {node {@} child {node {$x$}} child {node {$\lambda y$} child
    {node {$y$}}}} ;
\end{tikzpicture}
\label{fig:tree}
\caption{Tree form of $\lambda x.(x(\lambda y.y))$}
\end{figure}
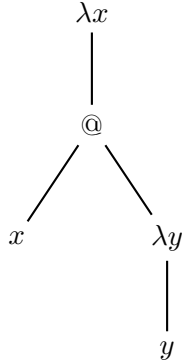

We can characterize $\lambda$-terms by means of their paths, where a
path formalizes the idea of descending through the abstract syntax
tree of a term. For example, the tree for the
$\lambda$-term $\lambda x. (x(\lambda y.y))$ is shown in
Figure~\ref{fig:tree} has has two paths:
\begin{enumerate}
\item descend through the binder for $x$, go left at the application,
stop at $x$, and
\item descend through the binder for $x$, go right at the application,
descend through the binder for $y$, stop at $y$
\end{enumerate}
Our goal is section is to show that if two $\lambda$-terms share all
the same paths, then the terms must be equal. We call this the {\em
  path equivalence} property.

We are interested in the path equivalence property since it expresses
a model checking-like property over terms with binding structure. This
type of property is difficult or impossible to formalize in
competing frameworks like Twelf \cite{pfenning99cade} since expressing
the hypothetical property that
two $\lambda$-terms have all the same paths requires a sufficiently
rich logic. However, in our framework, we find that this property can
be stated and reasoned about directly. Also, this application
illustrates how we can use definitions to describe the structure of
multiple judgment contexts which have related structure. Finally, a
complication in this application demonstrates the need for occasional
vacuity properties to be established regarding the occurrences of
nominal constants in terms.

\begin{figure}[t]
\begin{align*}
&\term M \supset \term N \supset \term (\app M N) \\
&(\forall x. \term x \supset \term (R\ x)) \supset \term (\uabs R) \\[10pt]
&\pathp M P \supset \pathp {(\app M N)} {(\pleft P)} \\
&\pathp N P \supset \pathp {(\app M N)} {(\pright P)} \\
&(\forall x.\forall p. \pathp x p \supset \pathp {(R\ x)} {(S\ p)})
\supset \pathp {(\uabs R)} {(\bnd S)}
\end{align*}
\caption{Specification of paths through $\lambda$-terms}
\label{fig:spec-path}
\end{figure}

We introduce a type $tm$ for untyped $\lambda$-terms and $pt$ for
paths together with the following constructors.
\begin{align*}
\hsl{app} &: tm \to tm \to tm & \hsl{abs} &: (tm \to tm) \to tm
\end{align*}
\vspace{-0.9cm}
\begin{align*}
\hsl{left} &: pt \to pt & \hsl{right} &: pt \to pt &
\hsl{bnd} &: (pt \to pt) \to pt
\end{align*}
We then introduce the predicates $\hsl{term} : tm \to o$ and ${\hsl
  path} : tm \to pt \to o$ defined by the specification logic formulas
in Figure~\ref{fig:spec-path}.

Given this description of paths through $\lambda$-terms we can state
the path equivalence property as follows.
\begin{align*}
&\forall M, N.~ (\,\tridot \term M) \supset (\forall P.~ (\,\tridot
\pathp M P) \supset (\,\tridot \pathp N P)) \supset (M = N)
\end{align*}
Note that we have added the explicit assumption $(\,\tridot \term M)$
so that we can induct on the structure of $M$. Also, we have stated
only that the paths in $M$ are also in $N$, but not vice-versa. It
turns out that this weaker property is sufficient to prove the result.

Before we can proceed with the proof of the above statement, we need
to strengthen it. In particular, when $M$ is an abstraction we need to
consider how the contexts for the {\sl term} and {\sl path} judgments
will grow. This is done with the following definition of {\sl ctxs}
which describes not only how each context grows, but how the two
contexts are related.
\begin{align*}
\ctxs {nil} {nil} \mueq \top && (\nabla x. \nabla p. \ctxs {(\term x
  :: L)} {(\pathp x p :: K)}) \mueq \ctxs L K
\end{align*}
Along with this definition, we need the following lemmas which allow
us to extract information about a term based on its membership in one
of the contexts described by {\sl ctxs}.
\begin{align*}
&\forall X, L, K.~ \ctxs L K \supset \member {(\term X)} L \supset \\
&\hspace{5cm}
\name X \land \exists P.~\member {(\pathp X P)} K \\
& \forall X, P, L, K.~ \ctxs L K \supset \member {(\pathp X P)} K
\supset \name X \land \name P
\end{align*}
The proofs of both lemma are by straightforward induction on the {\sl
  member} hypotheses.

We can state the strengthened equivalence property as follows.
\begin{align*}
&\forall L, K, M, N.~ \ctxs L K \supset (L \tridot \term M) \supset \\
&\hspace{3cm} (\forall P.~ (K \tridot \pathp M P) \supset (K \tridot
\pathp N P)) \supset (M = N)
\end{align*}
The proof of this statement is by induction on $(L \tridot \term M)$.
In the base case we need the following lemma which is proved by
induction one of the {\sl member} hypotheses.
\begin{align*}
&\forall L, K, X_1, X_2, P.~ \ctxs L K \supset \\
&\hspace{2cm} \member {(\pathp {X_1} P)} K \supset \member {(\pathp
  {X_2} P)} K \supset (X_1 = X_2)
\end{align*}
In the other cases of the proof, we need to show that the top-level
constructor of $M$ is also the top-level constructor for $N$. We do by
finding a path through $M$ and using the hypothesis that $M$ and $N$
share the same paths to find the same path in $N$. The top-level
constructor of that path will determine the top-level constructors of
$M$ and $N$. However, this requires that we can always find a path
through a term which we formalize this as the following lemma.
\begin{align*}
&\forall L, K, M, P.~ \ctxs L K \supset (L \tridot \term M) \supset
\exists P.~ (K\tridot \pathp M P)
\end{align*}
The proof of this lemma is by induction on $(L\tridot \term M)$.

There is one last complication in the proof of path equivalence which
comes from the inductive case concerning abstractions. Suppose $M =
\uabs R$ and $N = \uabs R'$. Here we know
\begin{equation*}
\forall P.~ (K \tridot \pathp {(\uabs R)} P) \supset
(K \tridot \pathp {(\uabs R')} P)
\end{equation*}
but in order to use the inductive hypothesis we must show
\begin{equation*}
\forall P.~ (\pathp x p :: K \tridot \pathp {(R\ x)} P) \supset
(\pathp x p :: K \tridot \pathp {(R'\ x)} P)
\end{equation*}
where $x$ and $p$ are nominal constants. Now the problem is that when
we go to prove this latter formula, the $\forallR$ rule says that we
must replace $P$ by $P'\ x\ p$ for some new eigenvariable $P'$. Note
that $P'$ is raised over both $x$ and $p$ even though the dependency
on $x$ must be vacuous. We must prove this vacuity to finish this case
of the proof, and thus we need the following lemma.
\begin{align*}
&\forall K, M, P. \nabla x, p.~ (\pathp x p :: K \tridot \pathp {(M\ x)}
{(P\ x\ p)}) \supset \exists P'. (P = \lambda z.P')
\end{align*}
This is proved by induction on the {\sl path} judgment. With this
issue resolved, the rest of the path equivalence proof is
straightforward.

As we have seen, the path equivalence property is expressed naturally
in our framework through the use of a formula with a nested universal
quantifier and implication. We briefly discuss the adequacy
considerations regarding such a formula. The goal is to use the path
equivalence property proven in \logic in order to prove the path
equivalence property for the object system. To do this, we need to
show that the hypotheses we have about the object system imply that
there are proofs in \logic of the corresponding hypotheses for the
formalization of the path equivalence problem; if we can show this,
then we will obtain the desired result by using the bijectivity of the
mappings for terms. Looking more carefully at the hypothesis, we see
that the main
concern is showing that if every path in a $\lambda$-term $m$ is a
path in another $\lambda$-term $n$ then the following is provable in
\logic:
\begin{equation}
\label{eq:path}
\forall P.~ (\,\tridot \pathp {\psi(\phi(m))} P) \supset
(\,\tridot \pathp {\psi(\phi(n))} P)
\end{equation}
Here $\phi$ is the bijection between object terms and their specification
logic representations, and $\psi$ is the bijection between
specification logic terms and their meta-logic representations. 

To complete this discussion, we provide a sketch of how a proof
of (\ref{eq:path}) might be constructed. We start with the knowledge
that every path 
in $m$ is a path in $n$. Then, assuming that the specification of {\sl
  path} is adequate, we know that whenever $\Delta \vdash \pathp
{\phi(m)} {\phi(p)}$ has an \hh derivation, it must be that $\Delta
\vdash \pathp {\phi(n)} {\phi(p)}$ also has an \hh derivation where
$\Delta$ is the specification of {\sl path} and {\sl term}. By the
adequacy of {\sl seq} established in Theorem~\ref{thm:seq-adequacy},
we know that whenever $\,\tridot \pathp {\psi(\phi(m))}
{\psi(\phi(p))}$ is provable in \logic, it must be that $\,\tridot
\pathp {\psi(\phi(n))} {\psi(\phi(p))}$ is also provable in \logic. We
will use this knowledge shortly. Now to prove (\ref{eq:path}) in
\logic we start by applying the $\forallR$ and $\supsetR$ rules. Then
we repeatedly apply appropriate left rules starting with the
assumption $\,\tridot \pathp {\psi(\phi(m))} P$. Since $\psi(\phi(m))$
has no eigenvariables and {\sl path} always deconstructs its first
argument, this repeated application of left rules can be made to
result only in sequents with no formulas on the left and where $P$ is
instantiated with a term such that $\,\tridot \pathp {\psi(\phi(m))}
P$ is provable in \logic. Now using our knowledge from before and the
assumption that $\phi$ and $\psi$ are bijections, it must be that
$\,\tridot \pathp {\psi(\phi(n))} P$ is provable in \logic. This is
exactly the form of the right side of each of the sequents which
results from the repeated application of left rules.
Thus each such sequent must be provable, and therefore
(\ref{eq:path}) must also be provable in \logic.

\section{Conversion between HOAS and de Bruijn Notation}
\label{sec:conv-de-bruijn}

De Bruijn notation is a first-order representation of binding which
uses numeric indices to associate variable occurrences with their
binders. More precisely, the index denoting a variable occurrence
corresponds the number of abstractions between the occurrence and its
binder. In this section we describe a translation between higher-order
abstract syntax representation and de Bruijn notation for untyped
$\lambda$-terms, and we prove that this translation is deterministic
in both directions. This example highlights the use of a definition
for describing a context which carries more than just variable
freshness information.

We start by introducing the type $tm$ for the higher-order abstract
syntax representation of untyped $\lambda$-terms with the constructors
$\hsl{app} : tm \to tm \to tm$ and $\hsl{abs} : (tm \to tm) \to tm$.
For natural numbers we use the type $nt$ with constructors $z : nt$
and $s : nt \to nt$. Finally, for de Bruijn notation terms we
introduce the type $db$ with the following constructors.
\begin{align*}
& \hsl{dabs}: db \to db &
& \hsl{dapp} : db \to db \to db &
& \hsl{dvar} : nt \to db
\end{align*}

\begin{figure}[t]
\begin{align*}
&\add z C C. \\
&\add A B C \supset \add {(s\ A)} B {(s\ C)} \\[10pt]
&\hodb M D {M'} \supset \hodb N D {N'} \supset \hodb {(\app M N)} D
{(\dapp {M'} {N'})} \\
&\depth X D_X \supset \add {D_X} {X'} D \supset \hodb X D (\dvar {X'}) \\
&(\forall x. \depth x D \supset \hodb {(R\ x)} {(s\ D)} {R'})
\supset \hodb {(\uabs R)} D {(\uabs {R'})}
\end{align*}
\caption{Specification of translation between HOAS and de Bruijn notation}
\label{fig:spec-ho2db}
\end{figure}

We translate from higher-order abstract syntax to de Bruijn notation
as follows. We walk over the structure of the term keeping track of
the number of abstractions we have descended through. Whenever we come
to an abstraction we use the context to record a new variable for that
abstraction and the abstraction depth at which it was encountered.
When we encounter a variable occurrence, we subtract the current
abstraction depth from the corresponding depth in the context to
determine the index for that variable occurrence. Using the predicates
$\hsl{add} : nt \to nt \to nt \to o$, $\hsl{depth} : tm \to nt \to o$,
and $\hsl{ho2db} : tm \to nt \to db \to o$, the specification of the
translation is presented in Figure~\ref{fig:spec-ho2db}.

Now there is a derivation of $\hodb M z {M'}$ if and only if $M$ is a
higher-order abstract syntax representation of the de Bruijn notation
term $M'$. Moreover, note that the translation is symmetric: we could
start with either $M$ or $M'$ and construct a derivation of $\hodb M z
{M'}$ to determine a value for the other.

Now we want to show that the above translation is deterministic in
both directions. In doing this, we will need to make certain
properties of natural numbers explicit. For this we make use of the
following two definitions.
\begin{align*}
\nat z &\mueq \top & \mle A A &\mueq \top \\
\nat (s\ A) &\mueq \nat A & \mle A (s\ B) &\mueq \mle A B
\end{align*}
Along with these we prove the following arithmetic properties by
straightforward induction.
\begin{align*}
& \forall A, B.~ \mle {(s\ A)} B \supset \mle A B \\
& \forall A.~ \nat A \supset \mle {(s\ A)} A \supset \bot \\
& \forall A, B, C.~ (\,\tridot \add A B C) \supset \mle B C \\
& \forall A_1, A_2, B, C.~ \nat C \supset (\,\tridot \add {A_1} B C)
\supset (\,\tridot \add {A_2} B C) \supset (A_1 = A_2) \\
& \forall A, B_1, B_2, C.~ (\,\tridot \add A {B_1} C)
\supset (\,\tridot \add A {B_2} C) \supset (B_1 = B_2)
\end{align*}
Note that we have made the assumption {\sl nat} explicit in some of
these to provide a target for induction.

Derivations of {\sl ho2db} will construct contexts of the form
\begin{equation*}
\depth {x_n} (s^n\ z)::\ldots::\depth {x_2} {(s\ (s\ z))}::\depth {x_1}
{(s\ z)}::\depth {x_0} z::nil
\end{equation*}
where each $x_i$ is unique. Moreover, the numbers associated with each
$x_i$ will also be unique since they are sequential. Each of these
uniqueness properties will be needed to show determinacy for one or
the other direction of the translation. We can describe these contexts
with the following definition.
\begin{align*}
\dctx {nil} z \mueq \top && (\nabla x. \dctx {(\depth x D :: L)}
{(s\ D)}) \mueq \dctx L D
\end{align*}
The corresponding lemma for {\sl dctx} is as follows
\begin{equation*}
\forall E, L, D.~ \dctx L D \supset \member E L \supset \exists X,
D_X.~ (E = \depth X D_X) \land \name X \\
\end{equation*}
The proof is by induction on the {\sl member} judgment. One
complication related to contexts arises when we call {\sl add} from
within {\sl ho2db}: the {\sl add} judgments inherits the context from
{\sl ho2db}. This is a problem since all of our lemmas about {\sl add}
assume that it has an empty context. We can fix this by proving the
following lemma.
\begin{equation*}
\forall L, D, A, B, C.~ \dctx L D \supset (L \tridot \add A B C)
\supset (\,\tridot \add A B C)
\end{equation*}
This is proved by a simple induction on the {\sl add} judgment.

Now let us consider the determinacy proof going from higher-order
abstract syntax to de Bruijn notation. For this, we need the following
lemma which says that each variable in the context has a unique index
associated with it.
\begin{align*}
&\forall L, D, X, D_1, D_2.~ \dctx L D \supset \\
&\hspace{2cm} \member {(\depth X D_1)} L \supset \member {(\depth X
  D_2)} L \supset (D_1 = D_2)
\end{align*}
This is proved by a straightforward induction on one of the {\sl
  member} hypotheses. Then we can prove the generalized determinacy
result:
\begin{align*}
& \forall L, M, M'_1, M'_2, D.~ \dctx L D \supset \\
&\hspace{2cm} (L \tridot \hodb M D
{M'_1}) \supset (L\tridot \hodb M D {M'_2}) \supset (M'_1 = M'_2).
\end{align*}
This is proved by induction on one of the {\sl ho2db} judgments. We
then apply this generalization with $L = nil$ and $D = z$ to get
the specific determinacy result we care about.

To prove determinacy in the other direction we need a lemma which
says that each index in the context has a unique variable associated
with it. We can state this as
\begin{align*}
&\forall L, D, X_1, X_2, D_X.~ \dctx L D \supset\\
&\hspace{2cm} \member {(\depth {X_1} D_X)} L \supset \member {(\depth
  {X_2} D_X)} L \supset (X_1 = X_2).
\end{align*}
This is proved by induction on one of the {\sl member} hypotheses,
however we need an additional result about the restrictions on
indices in the context for the proof to go through. Specifically, the
following lemma is required.
\begin{align*}
&\forall L, D, D_X, X.~ \dctx L D \supset \member {(\depth X D_X)} L
\supset \mle D D_X \supset \bot
\end{align*}
This is proved by induction on the {\sl member} hypothesis and in turn
requires the following result which follows by a simple induction.
\begin{align*}
&\forall L, D.~ \dctx L D \supset \nat D
\end{align*}
With these lemmas in place, the generalized determinacy result is as
follows.
\begin{align*}
&\forall L, M_1, M_2, D, M'.~ \dctx L D \supset \\
&\hspace{2cm} (L\tridot \hodb {M_1} D M') \supset (L\tridot \hodb
{M_2} D M') \supset (M_1 = M_2)
\end{align*}
This is now proved by straightforward induction on one of the {\sl ho2db}
hypotheses, and again we can substitution $L = nil$ and $D = z$ to
obtain the specialized result.

\section{Formalizing Tait-Style Proofs for Strong Normalization}
\label{sec:girards-strong-norm}

Tait introduced the idea of a logical relation and showed how this
could be used to provide an elegant proof of the strong normalization
property for the typed $\lambda$-calculus \cite{tait67jsl}. Girard
subsequently generalized this idea to obtain a strong normalization
result for the computationally much richer second-order
$\lambda$-calculus or System F \cite{girard72phd}. This style of
argument has both an elegance and a sophistication that would be
interesting to see captured in formalizations. We show in this section
that our framework is up to the task by considering an encoding of the
argument for the simply typed $\lambda$-calculus drawn from
\cite{girard89book}. One note, however, is that the strong
normalization argument requires a definition for a logical relation
which does not satisfy our current stratification restriction. We
strongly believe that the stratification condition on definitions in
\logic could be weakened to allow this definition while preserving
cut-elimination, but at present we have no corresponding
cut-elimination proof.

\begin{figure}[t]
\begin{align*}
&\type i \\
&\type A \supset \type B \supset \type (\arrow A B) \\[10pt]
&\of M (\arrow A B) \supset \of N A \supset \of {(\app M N)} B \\
&\type A \supset (\forall x. \of x A \supset \of {(R\ x)} B) \supset
\of {(\abs A R)} {(\arrow A B)} \\
&\type A \supset \of c A \\[10pt]
&\step M {M'} \supset \step {(\app M N)} {(\app {M'} N)} \\
&\step N {N'} \supset \step {(\app M N)} {(\app M {N'})} \\
&\step {(\app {(\abs A R)} M)} {(R\ M)} \\
&(\forall x. \step {(R\ x)} {(R'\ x)}) \supset \step {(\abs A R)}
{(\abs A {R'})}
\end{align*}
\caption{Specification of typing and one-step reduction}
\label{fig:spec-girard}
\end{figure}

To encode the simply-typed $\lambda$-calculus we use the familiar
types $ty$ and $tm$ along with their constructors {\sl i}, {\sl
  arrow}, {\sl app}, and {\sl abs}. In Girard's argument he assumes
that we are always working with open terms and can therefore always
select a free variable at any type. Rather than explicitly
representing this style of reasoning, we opt to introduce a constant
$c : tm$ which we allow to take on any type. This does not impair the
adequacy of our final result: if a term does not contain $c$ then none
of the terms it reduces to will contain it, and therefore $c$ has no
effect on normalization. The specification of typing ($\hsl{of} : tm
\to ty \to o$) and one-step reduction ($\hsl{step} : tm \to tm \to o$)
is given in Figure~\ref{fig:spec-girard}. The specification includes a
predicate a predicate $\hsl{type} : ty \to o$ to recognize types,
which we use in the abstraction typing rule since this will be needed
for later arguments. Also, we add a typing clause for $c$ to allow it
to take on any type.

Strong normalization says that all reduction paths eventually
terminate. We can succinctly encode this property in the following
definition.
\begin{equation*}
\sn M \mueq \forall M'.~ (\,\tridot \step M {M'}) \supset \sn {M'}
\end{equation*}
Note that there is no explicit base case for {\sl sn}, but if $M$ has
no reductions then $(\,\tridot\step M {M'})$ will be impossible and
therefore $\sn M$ will hold. Also, we will see that structural
induction on the definition of {\sl sn} corresponds to induction on
the structure of the possible reductions from a term. The adequacy of
{\sl sn} can be established in the same manner as adequacy for the
path equivalence application (Section~\ref{sec:path-equiv-lambda}).
We can now state the goal of this section:
\begin{equation*}
\forall M, A.~ (\,\tridot \of M A) \supset \sn M
\end{equation*}
The rest of this section describes definitions and lemmas necessary to
prove this formula.

\subsection{Typing and One-step Reduction}

In order to reason about typing judgments, we need to make explicit
the structure of the contexts of such judgments. They are described by
the following definition.
\begin{align*}
&\ctx nil \mueq \top &
&(\nabla x.\ctx (\of x A :: L)) \mueq (\,\tridot \type A) \land
\ctx L
\end{align*}
We then prove the corresponding lemma about context
membership:
\begin{equation*}
\forall E, L.~ \ctx L \supset \member E L \supset \exists X, A.~ (E =
\of X A) \land \name X \land (\,\tridot \type A)
\end{equation*}
The proof is by induction the the {\sl member} hypothesis. Another
auxiliary lemma we need about typing says that we can extract {\sl
  type} judgments from {\sl of} judgments.
\begin{equation*}
\forall L, M, A.~ \ctx L \supset (L \tridot \of M A) \supset (\,\tridot
\type A)
\end{equation*}
This is proved by induction on the {\sl of} judgment and requires the
following lemma which says that {\sl type} judgments ignore typing
contexts.
\begin{equation*}
\forall L, A.~ \ctx L \supset (L \tridot \type A) \supset (\, \tridot
\type A)
\end{equation*}
This is proved by induction on the {\sl type} judgment.

Now, the first real result we need is that one-step reduction
preserves typing:
\begin{equation*}
\forall L, M, M', A.~ \ctx L \supset (L \tridot \of M A) \supset (\,
\tridot \step M {M'}) \supset (L \tridot \of {M'} A).
\end{equation*}
The proof is by induction on the {\sl step} judgment. Note that we
have to generalize the typing context since one-step reduction can
take place underneath abstractions. Another useful lemma is the
following.
\begin{equation*}
\forall M.~ \sn (\app M c) \supset \sn M
\end{equation*}
The proof is by induction on {\sl sn}.

\subsection{The Logical Relation}

The difficulty with proving strong normalization directly is that it
is not closed under application, \ie, $\sn M$ and $\sn N$ does not
imply $\sn (\app M N)$. Instead, we must strengthen the normalization
property to one which includes a notion of closure under application.
This strengthened condition is called {\em reducibility} and is
originally due to Tait \cite{tait67jsl}. We say that a term $M$
reduces at type $A$ if $\reduce M A$ holds where {\sl reduce} is
defined as follows:
\begin{align*}
\reduce M i \mueq\null & (\,\tridot \of M i) \land \sn M \\
\reduce M (\arrow A B) \mueq\null & (\,\tridot \of M (\arrow A
B)) \land \null \\
& (\forall U.~ \reduce U A \supset \reduce {(\app M U)} B)
\end{align*}
Note that {\sl reduce} is defined with a negative use of itself and
therefore does not satisfy the current stratification condition on
definition. However, the second argument to {\sl reduce} is smaller in
the negative occurrence, and thus there are no logical loops
introduced by this definition. Intuitively, we can think of $(\lambda
x. \reduce x A)$ as defining a separate fixed-point for each type $A$,
and that these fixed-points are constructed based on induction on $A$.

An auxiliary notion used when discussing reducibility is called {\em
  neutrality}\/: a term is called {\em neutral} if it is not an
abstraction. We can define this directly as follows.
\begin{equation*}
\neutral M \triangleq \forall A, R.~ (M = \abs A R) \supset \bot
\end{equation*}
Now Girard lays out three properties of reducibility which we can
formalize as follows.
\begin{align*}
\mbox{(CR 1) } &\forall M, A.~ (\,\tridot \type A) \supset \reduce M A
\supset \sn M \\
\mbox{(CR 2) } &\forall M, M', A.~ (\,\tridot \type A) \supset \reduce M A
\supset (\,\tridot \step M {M'}) \supset \reduce {M'} A \\
\mbox{(CR 3) } &\forall M, A.~ (\,\tridot \type A) \supset \neutral M
\supset (\,\tridot \of M A) \supset\null \\
&\hspace{1cm} (\forall M'.~ (\,\tridot \step M {M'})\supset \reduce
{M'} A) \supset \reduce M A
\end{align*}
Each of these follows by induction on the {\sl type} judgment. The
proof of (CR 2) is straightforward, but the proofs (CR 1) and (CR 3)
are more complicated. In particular, (CR 1) depends on (CR 3) at
types structurally smaller than $A$ while (CR 3) depends on
(CR 1) at the same type $A$. As in the POPLmark application
(Section~\ref{sec:poplmark-challenge}) we can handle this by stating
a combined lemma and using $\landR^*$ within the induction:
\begin{align*}
&\forall A.~ (\,\tridot \type A) \supset \\
&\hspace{1cm} (\forall M.~ \reduce M A \supset \sn M) \land \null \\
&\hspace{1cm} (\forall M.~ \neutral M
\supset (\,\tridot \of M A) \supset\null \\
&\hspace{3cm} (\forall M'.~ (\,\tridot \step M {M'})\supset \reduce
{M'} A) \supset \reduce M A)
\end{align*}
The proof is by induction on the {\sl type} judgment, and the (CR 1)
portion of the proof is relatively straightforward. In the (CR 3)
portion, when $A$ is an arrow type, say $\arrow {A_1} {A_2}$, we need
to show
\begin{equation*}
\forall U.~ \reduce U {A_1} \supset \reduce {(\app M U)} A_2.
\end{equation*}
From the (CR 1) inductive hypothesis on type $A_1$ we can determine
that $\sn A_1$ holds, and then proof is by an inner induction on $\sn
A_1$.

The last reducibility lemma we need says that if for all reducible $U$
of type $A$, $M[U/x]$ is reducible, then so is $\lambdat x A M$. For
$\lambdat x A M$ to be reducible requires showing that for all
reducible $V$ that $M\ V$ is reducible. Girard proves this by
induction on the sum of the lengths of the longest reduction paths
from $M$ and $V$. We can state this unfolded reducibility lemma as
follows.
\begin{align*}
\forall V, M, A, B.~&  (\,\tridot \of {(\abs A M)} (\arrow A B))
\supset \null \\
& \sn V \supset \sn (M\ c) \supset \reduce V A \supset \null \\
&(\forall U.~ \reduce U A \supset \reduce {(M\ U)} B) \supset \null \\
&\hspace{2cm} \reduce {(\app {(\abs A M)} V)} B
\end{align*}
The proof of this formula is by induction on $\sn V$ with a nested
induction on $\sn (M\ c)$.

Clearly {\sl reduce} is closed under application and by (CR 1) it
implies strong normalization, thus we strengthen our desired
normalization result to the following:
\begin{equation*}
\forall M, A.~ (\,\tridot \of M A) \supset \reduce M A.
\end{equation*}
In order to prove this formula we will have to induct on the height of
the proof of the typing judgment. However, when we consider the case
that $M$ is an abstraction, we will not be able to use the inductive
hypothesis since {\sl reduce} is defined only on closed terms, \ie,
those typeable in the empty context. The standard way to deal with
this issue is to generalize the desired formula to say that if $M$, a
possibly open term, has type $A$ then each closed instantiation for
all the free variables in $M$, say $N$, satisfies $\reduce N A$. This
requires a formal description of simultaneous substitutions that can
``close'' a term.

\subsection{Arbitrary Cascading Substitutions and Freshness Results}

Given $(L\tridot \of M A)$, \ie, an open term and its typing context, we
define a process of substituting each free variable in $M$ with a
value $V$ which satisfies the logical relation for the appropriate
type. We define this {\sl subst} relation as follows:
\begin{align*}
& \subst {nil} M M \mueq \top \\
(\nabla x.&\subst {((\of x A) :: L)} {(R\ x)} M) \mueq \exists
U.~ \reduce U A \land \subst L {(R\ U)} M
\end{align*}
By employing nominal abstraction in the second clause, we are able to
use the notion of substitution in the meta-logic to directly and
succinctly encode substitution in the object language. Also note that
we are, in fact, defining a process of cascading substitutions rather
than simultaneous substitutions. Since the substitutions we define
(using closed terms) do not affect each other, these two notions of
substitution are equivalent. We will have to prove some part of this
formally, of course, which in turn requires proving results about the
(non)occurrences of nominal constants in our judgments.

One consequence of defining cascading substitutions via the notion of
substitution in the meta-logic is that we do not get to specify where
substitutions are applied in a term. In particular, given an
abstraction $\abs A R$ we cannot preclude the possibility that a
substitution for a nominal constant in this term will affect the type
$A$. Instead, we must show that well-formed types cannot contain free
variables which we formalize as
\begin{equation*}
\forall A. \nabla x.~ (\,\tridot \type (A\ x)) \supset \exists A'.~ (A =
\lambda y. A').
\end{equation*}
This formula essentially states any dependencies a type has nominal
constants must be vacuous. A related result is that in any provable
judgment of the form $(L\tridot \of M A)$, any nominal constant
(denoting a free variable) in $M$ must also occur in $L$, \ie,
\begin{equation*}
\forall L, M, A. \nabla x.~ \ctx L \supset (L \tridot \of {(M\ x)} (A\
x)) \supset \exists M'.~ (M = \lambda y. M')
\end{equation*}
This is proved by induction on the {\sl of}\/ judgment.

Given these results about the (non)occurrences of nominal constants in
judgments, we can now prove fundamental properties of arbitrary
cascading substitutions. The first property states that closed terms,
those typeable in the empty context, are not affected by
substitutions, \ie,
\begin{equation*}
\forall L, M, N, A.~ (\,\tridot \of M A) \supset \subst L M N \supset (M
= N).
\end{equation*}
The proof here is by induction on {\sl subst} which corresponds to
induction on the length of the list $L$. The key step within the proof
is using the lemma that any nominal constant in the judgment
$(\,\tridot \of M A)$ must also be contained in the context of that
judgment. Since the context is empty in this case, there are no
nominal constants in $M$ and thus the substitutions from $L$ do not
affect it.

We must show that our cascading substitutions act compositionally on
terms in the simply-typed $\lambda$-calculus. For the term $c$ this is
almost trivial,
\begin{equation*}
\forall L, M.~ \subst L c M \supset (M = c).
\end{equation*}
The proof is by induction on {\sl subst}. For application we have the following.
\begin{align*}
& \forall L, M, N, U.~
\ctx L \supset \subst L {(\app M N)} U \supset \null \\
&\hspace{2cm} \exists M_U, N_U.~ (U = \app {M_U} {N_U}) \land
\subst L M M_U \land \subst L N N_U
\end{align*}
This is proved by induction on {\sl subst}. Finally, for abstractions
we prove the following, also by induction on {\sl subst}:
\begin{align*}
&\forall L, A, R, U.~
\ctx L \supset \subst L {(\abs A R)} U \supset (\,\tridot \type A)
\supset\null\\
&\hspace{2cm} \exists R_U.~ (U = \abs A R_U) \land \null \\
&\hspace{3cm} (\forall V.~
\reduce V A  \supset \nabla x.~ \subst {((\of x A) :: L)} {(R\ x)}
(R_U\ V))
\end{align*}
Here we have the additional hypothesis of $(\,\tridot \type A)$
to ensure that the substitutions created from $L$ do not affect $A$.
At one point in this proof we have to show that the order in which
cascading substitutions are applied is irrelevant. The key to showing
this is realizing that all substitutions are for closed terms. Since
closed terms cannot contain any nominal constants, substitutions do
not affect each other.

Finally, we must show that cascading substitutions preserve typing.
Moreover, after applying a full cascading substitution for all the
free variables in a term, that term should now be typeable in the
empty context:
\begin{align*}
&\forall L, M, N, A.~ \ctx L \supset \subst L M N \supset (L\tridot
\of M A) \supset (\,\tridot \of N A).
\end{align*}
This formula is proved by induction on {\sl subst}.

\subsection{The Final Result}

Using cascading substitutions we can now formalize the generalization
of strong normalization that we described earlier: given a (possibly
open) well-typed term, every closed instantiation for it satisfies the
logical relation {\sl reduce}\/:
\begin{equation*}
\forall L, M, N, A .~ \ctx L \supset (L\tridot \of M A) \supset \subst L M N
\supset \reduce N A
\end{equation*}
The proof of this formula is by induction on the typing judgment. The
inductive cases are fairly straightforward using the compositional
properties of cascading substitutions and various results about
reducibility. In the base case, we must prove
\begin{equation*}
\forall L, M, N, A.~ \ctx L \supset \member {(\of M A)} L \supset \subst L
M N \supset \reduce N A,
\end{equation*}
which is done by induction on {\sl member}. Strong normalization is
now a simple corollary where we take $L$ to be $nil$. Thus we have
proved
\begin{equation*}
\forall M, A.~ (\,\tridot \of M A) \supset \sn M.
\end{equation*}



\chapter{Related Work}
\label{ch:related-work}

There are many frameworks which can be used to specify, to prototype,
and to reason about computational systems. Some of these are designed
specifically for this purpose while others have a different
motivation, but can achieve a similar result. In this chapter we
present a selection of these frameworks and contrast their
capabilities with the framework put forth in this thesis. As the
contributions of this thesis are primarily in the reasoning part of
the framework, we shall give extra attention to this component in the
comparisons.

Our framework is based on a two-level logic approach to reasoning. We
have found this to be very effective in practice, but one could use
the logic \logic in a single-level logic fashion as well. The
frameworks in this chapter come in both varieties: some use a
two-level logic approach to which we can compare directly, while
others use a single-level logic approach. In either case, the
differences due to the reasoning approach used are often overshadowed
by the differences in the treatment of binding. Thus we shall often
say very little about the reasoning approach except when comparing
against another two-level logic framework.

We organize our comparison of frameworks around the techniques used to
represent the binding structure of objects. This is by far the most
salient characteristic of the frameworks, and has the largest effect
on the succinctness and the quality of the corresponding reasoning.
Thus we will focus on issues such as the representation of binding,
determining equality modulo renaming of bound variables,
capture-avoiding substitution, and representing judgments with
side-conditions related to binding. We will use the example of the
simply-typed $\lambda$-calculus from Section~\ref{sec:example} to
illustrate these issues. We will order our comparisons based on the
kind of support for binding provided by the framework. Specifically,
we will look at frameworks based on first-order, nominal, and
higher-order representations.

\section{First-order Representations}
\label{sec:related-first-order}

First-order representations provide no special treatment for binders.
As a result, variables must be encoded using strings or integers and
binding aspects must be captured through constructors. Further,
mechanisms for manipulating and reasoning about binders must be
developed by interpreting the constructors representing them on a
case-by-case basis by by users of the framework. On the other hand,
the benefit of first-order representations is that many mature
frameworks exist which support this type of representation. For
example, languages like SML and Prolog can effectively prototype
specifications written using a first-order representation, while in
the reasoning phase, theorem provers like Coq \cite{bertot04book},
ACL2 \cite{kaufmann00book}, and HOL \cite{harrison96fmcad} can operate
directly on first-order representations. Our discussion in this
section will focus not on any particular framework but rather on the
benefits and costs of various first-order representations. In
particular, we look at the three most common first-order
representations: named, nameless, and locally nameless.

\subsection{Named Representation}
\label{sec:named}

The most direct and naive approach to encoding binders is to assign
each variable a fixed name. For instance, the term $(\lambdat x i x)$
might be encoded as $(\abst {\mbox{``$x$''}} i (\var
\mbox{``$x$''}))$. Here we have picked a particular name, $x$, to
denote the otherwise arbitrary variable in the function. This
representation is very natural, but it creates at least three major
problems for users.

First, equality modulo the renaming of bound variables is not
reflected in the representation. For example, the terms $(\lambdat x i
x)$ and $(\lambdat y i y)$ have two different representations, $(\abst
{\mbox{``$x$''}} i (\var \mbox{``$x$''}))$ and $(\abst
{\mbox{``$y$''}} i (\var \mbox{``$y$''}))$. Thus users of a named
representation must explicitly define a notion of equivalence for each
syntactic class with binding. This becomes particularly painful in
reasoning where the user must establish many equivalence lemmas.

Second, no support is provided for capture-avoiding substitution over
binding, and instead users must define this substitution on their own.
Naive capture-avoiding substitution is not structurally recursive, and
thus one must resort to well-founded recursion or instead use
simultaneous capture-avoiding substitution. Either choice results in
additional overhead during reasoning when the user must prove various
substitution lemmas. Moreover, substitution must be defined for each
class of syntactic objects with binding, and the proofs of related
lemmas must be repeated.

Third, no logical support is provided for treating side-conditions
related to variable binding structure. An example of such a
side-condition is manifest in the following rule for typing
abstractions in the $\lambda$-calculus:
\begin{equation*}
\infer[x\notin\dom(\Gamma).]
 {\Gamma \vdash (\lambdat x a r) : a \to b}
 {\Gamma, x:a \vdash r : b}
\end{equation*}
With the named representation, users must devise their own mechanisms
for treating such side-conditions. A naive approach in the case of the
rule above is to select any fresh variable name, but this can lead to
structural induction principles which are too weak to be usable in
practice. Moreover, one must still prove that the choice for a
variable name is truly arbitrary.

Large-scale developments have been constructed using the named
representation, and the result is often that the binding issues
overwhelm the development. For instance, VanInwegen used a named
representation to encode and reason about SML in the HOL theorem
prover \cite{vaninwegen96phd}. She noted:
\begin{quote}
Proving theorems about substitutions (and related operations such as
alpha-conversion) required far more time and HOL code than any other
variety of theorems.
\end{quote}

\subsection{Nameless Representation}
\label{sec:nameless}

A more sophisticated first-order representation encodes each variable
occurrence with an integer denoting the location of its binder
relative to the binding structure around it. Commonly, one uses the
distance from the variable occurrence to its binder, measured in terms
of other binders above it in the abstract syntax tree. For example,
the term $(\lambdat x i (\lambdat y i x))$ would be encoded as $(\abs
i (\abs i (\var 2)))$. Here the 2 denotes that the binder for this
variable occurrence is two binders away. This kind of representation
originates from de Bruijn \cite{debruijn72} and hence is often
referred to as the de Bruijn representation.

The benefit of a nameless representation over a named representation
is that $\alpha$-equivalent terms, \ie, those that differ only in the
names of bound variables, are syntactically identical. Thus in the
reasoning phase the user does not need to prove additional properties
about $\alpha$-equivalence.

The nameless representation shares many problems with the named
representation and has some additional ones as well. The nameless
representation still requires users to define capture-avoiding
substitution themselves, and now this makes it necessary to reason
about the correctness of the arithmetical operations that have to be
carried out for maintaining the consistency of the representation when
effecting substitutions. A new difficulty introduced by the nameless
treatment of variables is that representations become hard for humans
to read, since different occurrences of the same variable in them may
be rendered into different integers depending on the contexts in which
they appear. This also has an impact on the statements of lemmas and
theorems that often need to explicitly talk about re-numberings and
other arithmetical operations over terms, thereby diminishing clarity.

The nameless representation has been used in large-scale developments.
Hirschkoff, for instance, used it to formalize the $\pi$-calculus in
the Coq theorem prover \cite{hirschkoff97tphol}. He found that the
nameless representation simplified much of the work with bound
variables versus the named representation, but the treatment of
binding within it still overwhelmed the development. He concluded:
\begin{quote}
Technical work, however, still represents the biggest part of our
implementation, mainly due to the managing of De Bruijn indexes [...]
Of our 800 proved lemmas, about 600 are concerned with operators on
free names.
\end{quote}

\subsection{Locally Nameless Representation}
\label{sec:locally-nameless}

The most promising first-order representation is a hybrid approach
which uses the nameless representation for bound variables and the
named representation for free variables. This is called the
locally nameless representation \cite{aydemir08popl, chargueraud09ln}.

The locally nameless representation has advantages over both the named
and nameless representations. First, $\alpha$-equivalent terms are
syntactically equal, as in the nameless representation. Second, the
statement of lemmas and theorems rarely need to talk about
arithmetical operations over terms. Third, since free and bound
variables are syntactically distinguished, capture-avoiding
substitution can be defined in a straightforward and structurally
recursive way.

Like other first-order approaches, the locally nameless representation
still requires users to define capture-avoiding substitution and
prove various lemmas about it. A drawback specific to this
representation is that users must provide functions which bind and
unbind variables (\ie, implementing the interface between the named
and nameless representations). Constructing or deconstructing a term
with binding requires going through these functions in order to ensure
that certain invariants regarding free and bound variables are
maintained. Finally, users must show that these binding and unbinding
functions interact with substitution in appropriate ways. Recent
progress has been made in automatically generating this type of
infrastructure \cite{aydemir09lngen}.

The locally nameless representation has some analogs to our own
representation in the following sense: we represent bound variables
using $\lambda$-terms and free variables using nominal constants.
However, we provide capture-avoiding substitution for free to the
user. Unbinding and binding of terms (\eg, switching between
$\lambda$-binders and nominal constants) is handled using application
and nominal abstraction, respectively. In the locally nameless
approach one occasionally needs to prove that free variables can be
renamed while preserving provability, while that is an innate property
of our framework due to our treatment of nominal constants. The
fundamental contrast is that the locally nameless representation
allows one to use an existing theorem prover, but requires significant
binding infrastructure to be constructed, while our representation
requires a new theorem prover, but incorporates binding infrastructure
into the theory underlying the prover.

\section{Nominal Representations}
\label{sec:related-nominal}

The nominal representation of binding is a mild extension of
first-order abstract syntax with support for $\alpha$-equivalence
classes. The basis of the nominal representation is an infinite
collection of names called atoms together with a freshness
predicate---denoted by the infix operator $\#$---between atoms and
other objects and a swapping operation involving a pair of atoms and a
term. Binding is represented by means of a term constructor $\langle
\cdot \rangle \cdot$ which takes an atom and a term. The nominal
representation then assumes certain properties of swapping and
freshness with respect to this constructor so that
$\alpha$-equivalence classes are respected. This representation is
also referred to as nominal abstract syntax.

Nominal representations were first introduce through the nominal logic
of Pitts \cite{Pitts03ic}, which is an extension of first-order logic.
When working with nominal abstract syntax in a logical setting it is
often desirable to quantify over fresh atoms. In this regard, a useful
consequence of the properties assumed for freshness and swapping is
that the following equivalence holds for any formula $\phi$ whose free
variables are $a, x_1,\ldots, x_n$ where $a$ is of atom type:
\begin{equation*}
\exists a. (a \# x_1 \land \ldots \land a \# x_n \land \phi)
\quad \equiv\quad
\forall a. (a \# x_1 \land \ldots \land a \# x_n \supset \phi)
\end{equation*}
Nominal logic introduces the \new-quantifier by defining $\new a.\phi$
as one of the above formulas. This is very reminiscent of the
properties shown for the $\nabla$-quantifier in
Section~\ref{sec:nabla-freshness}, and in general, the
$\nabla$-quantifier and the \new-quantifier behave very similarly.

The most prominent specification and prototyping language based on
nominal representations is $\alpha$Prolog, an extension of Prolog that
accords a proof search interpretation of a version of Horn clauses in
nominal logic \cite{cheney03unif}. In particular, $\alpha$Prolog
allows the \new-quantifier to appear in the heads of clauses. This
allows $\alpha$Prolog to describe specifications which involve a finer
treatment of names than what is possible in our specification logic of
\hh. However, it seems that $\alpha$Prolog clauses bear a close
resemblance to the patterned form of definitions in \logic which allow
the $\nabla$-quantifier in the head (see
Section~\ref{sec:pattern-form}). While a formal encoding of
$\alpha$Prolog clauses as definitions in \logic is left to future
work, we note that such definitions can be animated using a system
similar to Bedwyr \cite{baelde07cade}, a specification tool based on a
simple proof search procedure for the Linc logic (one of the
precursors to \logic).

Nominal logic does not have a parallel to the fixed-point
interpretation of definitions in \logic, and thus nominal logic cannot
be used directly to reason about specifications written within it.
Instead, such reasoning must be carried out indirectly by first
formalizing the relevant nominal logic specification in a richer logic
such as that underlying a system like Coq or Isabelle/HOL and then
using the capabilities of that logic
\cite{aydemir06lfmtp,urban05cade}. The most prominent development in
this area is the Nominal package for Isabelle/HOL. This package allows
for an easy definition of syntactic objects with $\alpha$-equivalence
classes. This construction is conducted completely within the HOL
logic and can thus be trusted. Moreover, the construction of these
$\alpha$-equivalence classes and some boilerplate results about them
are provided automatically via the macro-like features of Isabelle.
This includes a strong induction principle which matches the one used
in typical ``pencil and paper'' proofs, and it includes a recursion
combinator which allows capture-avoiding substitution to be defined
structurally.

The nominal approach has a number of drawbacks. First, binding is only
simulated by means of a distinguished constructor and thus
substitution is not automatically provided. Instead, users must define
it on their own for both specification and reasoning, and
consequently, must prove substitution lemmas relative to their
definition of substitution. Second, in order to use functions and
predicates in the reasoning phase, one must prove properties which
state that name swapping does not change the results of a function or
the provability of a predicate---a property which is enforceable
statically for definitions of predicates in \logic. Third, to
effectively use the nominal representation in reasoning, one really
needs an existing package which automates the construction of
$\alpha$-equivalence classes and proves the related lemmas. Although
such a mature package exists for Isabelle/HOL, other theorem provers
may not have the automation capabilities necessary to effectively
construct such a package. Finally, an often trumpeted benefit of
nominal representations is that they allow a first-class treatment of
names, but the analyses enabled by that treatment seem no more
powerful than what is now provided by nominal abstraction. A
formal validation of this observation is left to future work.

\section{Higher-order Representations}
\label{sec:related-higher-order}

Higher-order representations use the meta-level function space to
encode binding in object languages, \eg, by using data constructors
such as $\hsl{abs} : (tm \to tm) \to tm$. This allows the object
representation to inherit all the properties of binding from the
meta-level. However, traditional tools often have a very strong notion
of equality (\eg, incorporating case analysis or fixed-point
combinators) which makes them ill-suited to encoding higher-order
representations. For this reason, we choose to focus here on
frameworks based on the $\lambda$-tree syntax representation of
binding which assumes only $\alpha\beta\eta$-conversion in determining
equality \cite{miller00cl}. This allows an adequate representation of
object languages with binding, and provides free $\alpha$-conversion
and capture-avoiding substitution for those languages. The cost is
that usually new frameworks must be developed which support the
$\lambda$-tree syntax representation. In this section we discuss such
frameworks which have been implemented.

\subsection{Hybrid}

Hybrid is a system which aims to support reasoning over higher-order
abstract syntax specifications using traditional theorem provers such
as Coq and Isabelle/HOL \cite{felty09tr}. The basic idea of the system
is translate higher-order abstract syntax descriptions into an
underlying de Bruijn representation. The logic of the theorem prover
then serves as the meta-logic in which reasoning is conducted. This
approach necessarily produces more overhead during reasoning due to
the need occasionally to reason about the effects of the translation.
However, there is good reason to believe that most of this can be
automated in the future. Also, Hybrid is often used in a two-level
logic approach using a specification logic which is essentially
identical to our own \hh specification language.

The Hybrid system, by design, lacks a meta-logic with the tools to
elegantly reason over higher-order abstract syntax descriptions. Most
notably, the meta-logics used by Hybrid lack a device like the
$\nabla$-quantifier for reasoning about open terms and generic
judgments. Recent work has suggested that such a device is not
necessary for simple reasoning tasks such as type uniqueness arguments
\cite{felty09ppdp}. Yet, it is unclear how the naive approach used in
this work will scale to problems such as those proposed by the
POPLmark Challenge \cite{aydemir05tphols}. In such problems one needs
to recognize as equivalent those judgments which differ only in the
renaming of free variables. Such a property is built into our
meta-logic by representing such free variables by nominal constants,
while in Hybrid one will have to manually develop and prove properties
about notions of variable permutations.

\subsection{Twelf}
\label{sec:twelf}

Twelf \cite{pfenning99cade} is a system for specifying and reasoning
with $\lambda$-tree syntax using LF, a dependently typed lambda
calculus \cite{harper93jacm}. In the LF methodology, object language
judgments are encoded as LF types, and rules for making judgments are
encoded as LF constructors for the corresponding types. The LF terms
inhabiting these types are then derivations of judgments. Thus LF
constitutes a specification language. Twelf implements an operational
semantics for constructing LF terms which provides a means of
animating LF specifications.

Since dependent types can be exploited in LF specifications, these can
often be more elegant than those described in our simply-typed
setting. For example, one can provide a definition of simply-typed
$\lambda$-terms where the type of a $\lambda$-term is reflected in the
type of its LF representation. When it is done in this way, one does
not need to talk about pre-terms and provide a separate typing
judgment for selecting well-typed terms. Moreover, this allows some
properties to be obtained for free. For example, we can define
evaluation over this representation of simply-typed $\lambda$-calculus
so that type preservation is a direct consequence of the type of the
evaluation judgment (\ie, evaluation is defined to take a
$\lambda$-term with a particular type and return another
$\lambda$-term with the same type). However, in terms of expressive
power, the simply-typed and dependently-typed specification languages
are equivalent \cite{felty91lf}. Thus when referring to the example of
the simply-typed $\lambda$-calculus we will assume that it is encoded
in LF in the same style as in our framework.

Since derivations of judgments are LF terms, we can think of defining
further judgments over such terms. For example, suppose that we encode
the simply-typed $\lambda$-calculus in LF including the type
constructors {\sl of} and {\sl eval} corresponding to typing and
evaluation judgments and the corresponding term constructors for
forming those judgments. Then we could define a judgment named {\sl
  preserve} which holds of a derivation of $(\of t a)$, a derivation
of $(\eval t v)$, and a derivation of $(\of v a)$. Viewing this
judgment as one which takes the first two arguments and produces the
third, we could provide term constructors for {\sl preserve} which
describe how derivations of $(\of t a)$ and $(\eval t v)$ are used to
reconstruct a derivation of $(\of v a)$. Twelf can then check that
this judgment is total in its first two arguments, \ie, it is defined
and terminates for all inputs. If so, we can think of {\sl preserve}
as a proof of the meta-property that evaluation preserves typing in
the simply-typed $\lambda$-calculus. This style of encoding is known
as a Twelf meta-theorem.

The Twelf approach of encoding meta-theorems as LF judgments has some
serious limitations. For example, consider the following statement of
the type preservation theorem: ``{\it forall} derivations of $(\of t
a)$ and {\it forall} derivations of $(\eval t v)$ there {\it exists} a
derivation of $(\of v a)$.'' This theorem was encoded in an LF
judgment which took the first two derivations as input and produced
the last one as output. In general, a judgment representing a Twelf
meta-theorem has inputs corresponding to $\forall$ quantifiers and
outputs corresponding to $\exists$ quantifiers. Therefore,
meta-theorems are restricted to a $\forall\exists$ quantification
structure.

A related issue with the Twelf approach is that Twelf does not have a
definition mechanism. Instead one has to use LF judgments to describe
the properties of a specification. This is severely limiting since LF
judgments can only describe behaviors that {\em may} happen and cannot
describe those which {\em must} happen. For example, to state the
strong normalization property for the simply-typed $\lambda$-calculus
in Section~\ref{sec:girards-strong-norm}, we used the following
definition:
\begin{equation*}
\sn M \mueq \forall M'.~ (\,\tridot \step M {M'}) \supset \sn {M'}
\end{equation*}
This says that in order for $\sn M$ to hold, every term to which $M$
can convert {\em must} also satisfy {\sl sn}. Such a definition is not
possible with Twelf. A similar issue arises if one tries to encode the
path equivalence property for $\lambda$-terms from
Section~\ref{sec:path-equiv-lambda}. The hypothesis in this case is
that every path in one $\lambda$-term {\em must} occur in the other
$\lambda$-term.

There is also a practical issue of relying on Twelf's totality checks
in order to ensure that a meta-theorem is correct. It is possible, for
example, for one to fill out the details of a meta-theorem so that
totality holds, but for Twelf's checker to be unable to determine
totality. In such a case, one must confront various options: 1) try to
rewrite the meta-theorem so that totality is more evident, 2) wait for
a new version of Twelf's totality checker that may be more powerful,
or 3) do a careful hand proof of totality. The first option is not
always possible, and the latter two are fairly undesirable.

An interesting comparison between the Twelf approach and our own is in
the treatment of judgment contexts. In our approach, the definition of
{\sl seq} includes a list argument which keeps track of the context of
a judgment and makes it explicit during reasoning. We then define a
predicate like {\sl ctx} which will recognize the structure of such a
context, and we prove various inversion lemmas about membership in
that context. In Twelf, such contexts are called regular worlds, and
although they are declared explicitly, they are kept implicit during
reasoning. The Twelf machinery automatically provides the associated
inversion properties of regular worlds. Like most automation, this is
very useful when it works and rather bothersome when it does not. For
instance, in the conversion between higher-order abstract syntax and
de Bruijn representations from Section~\ref{sec:conv-de-bruijn}, we
work with a context which has an arithmetical property which depends
on the judgment being made. Specifically, the context must not contain
de Bruijn indices which are greater than the depth at which the
conversion judgment is being made. This is needed to ensure uniqueness
of de Bruijn indices when descending underneath abstractions. The
regular worlds mechanism of Twelf does not allow the description of a
context to the depend on the arguments of the judgments made in that
context. Thus one cannot express this property directly and must
instead find a way to work around this limitation, \eg, by making the
context explicit \cite{crary08lfmtp}.


\subsection{Delphin}
\label{sec:delphin}

Delphin is a higher-order functional programming language which
operates over LF terms and can serve as a meta-logic for LF
specifications \cite{poswolsky08phd}. Delphin makes a distinction
between LF functions which are purely representational (\ie, that must
be parametric in their argument) and Delphin functions which are
computational (\ie, that may perform case analysis on their argument).
A Delphin meta-theorem is a Delphin function which is total. For
example, the property of type preservation for the simply-typed
$\lambda$-calculus is encoded as a function which takes LF terms
denoting derivations of $(\of t a)$ and $(\eval t v)$ and returns an
LF term denoting a derivation of $(\of v a)$. Like Twelf, it is
possible for Delphin not to be able to automatically determine
totality of a meta-theorem, and then one must either rewrite the
meta-theorem, wait for a stronger totality checker, or perform the
totality check by hand.

The central way in which Delphin improves on Twelf is that it treats
Delphin functions as first-class, and thus more sophisticated
properties can be encoded during reasoning. For example, the path
equivalence of $\lambda$-terms from
Section~\ref{sec:path-equiv-lambda} can be encoded fairly directly in
Delphin. The property that all the paths in the $\lambda$-term $s$
must also exist in the $\lambda$-term $t$ can be represented in
Delphin by a function which takes a judgment like $(\pathp s p)$ and
returns a judgment like $(\pathp t p)$, and such a function can be an
input (\ie, hypothesis) to a Delphin meta-theorem stating the path
equivalence property.

Delphin also uses first-class functions to treat the contexts of
specification judgments. When a Delphin meta-theorem is written, it
may make a recursive call to itself underneath some additional
abstractions. These abstractions create new variables for which the
Delphin meta-theorem must be defined. To achieve this, the Delphin
meta-theorem carries around an argument which is a function mapping
such variables to an appropriate invariant. This approach to
representing contexts is more flexible than the regular worlds
approach of Twelf. Specifically, in the example of conversion between
higher-order abstract syntax and de Bruijn representations from
Section~\ref{sec:conv-de-bruijn}, the dependency between the judgment
and the context in the judgment can be made explicit in Delphin. Thus
one can prove that the conversion is deterministic in a fairly
straightforward way in Delphin.

Despite the additional flexibility that Delphin provides in working
with the contexts of judgments, it still does not make those contexts
explicit as in our approach. Thus, some operations over contexts which
we can perform easily in our framework are difficult or impossible in
the Delphin approach. For example, in our formalization of Girard's
proof of strong normalization for the simply-typed $\lambda$-calculus
in Section~\ref{sec:girards-strong-norm}, we defined a process of
closing a term by instantiating all free variables with closed terms
of the appropriate types. This definition was based on walking over
the context of the typing judgment of such a term, something that is
not possible to do in Delphin.

\subsection{Tac}
\label{sec:Tac}

Tac is a general framework for implementing logics. For the purposes
of our present discussion, we will focus on the particular logic
$\mu$LJ which is the most popular logic implemented in Tac
\cite{tac-website, baelde08phd}. The logic $\mu$LJ comes from the same
line of logics as \logic and differs primarily in the semantics
attributed to the $\nabla$-quantifier. We recall that the
interpretation of $\nabla$ in \logic is derived from adding to \FOLDN
the exchange and strengthening properties related to this quantifier
that are embodied in the following equivalences:
\begin{align*}
\nabla x.\nabla y.F \equiv \nabla y.\nabla x.F && \nabla x.F \equiv F
\mbox{, if $x$ does not occur in $F$}
\end{align*}
The $\mu$LJ logic eschews these additions, strengthening the
interpretation of the $\nabla$-quantifier instead through a capability
to lift its predicative effect over types. At a practical, proof
construction level, whereas the $\nabla$-quantifier can be treated in
\logic using nominal constants, in $\mu$LJ it must be treated by using
explicit local contexts for each formula in a sequent. The size and
ordering of the local context is always respected and instantiations
for existentially or universally quantified variables may only use
those generic variables which appear in the local context.

The $\mu$LJ logic does not have an operation like nominal abstraction
and instead treats only equality. The issue with extending $\mu$LJ to
treat nominal abstraction is that the process of nominal
capture-avoiding substitution (through which the nominal abstraction
rules are defined) is based on carrying substitution information from
one formula into all other formulas in a sequent. In the minimal
setting, however, such information may be invalid in other formulas
because the local signatures do not match. For example, a substitution
which replaces $M$ by a variable $x$ from the local context does not
make any sense in a formula which contains $M$ but has an empty local
context. As a result of this lack of nominal abstraction, the
descriptions of properties such as the binding structure of
specification judgment contexts in $\mu$LJ is less direct and thus
harder to work with (see Figure~\ref{fig:ctx-lg} for an example).
Furthermore, without nominal abstraction, one cannot directly
formulate the invariants necessary to perform induction underneath
$\nabla$ (see Section~\ref{sec:induct-with-nabla}). An ability of
equivalent power is obtained in $\mu$LJ instead through the lifting
capability mentioned earlier \cite{baelde08lfmtp}. From a practical
perspective, however, we find that reasoning based on lifting is often
much more complicated than reasoning based on traditional induction
combined with nominal abstraction.

The benefit of minimal treatment of the $\nabla$-quantifier is that
the local context of a formula can be used to provide an adequate
encoding for certain types of similar contexts in an encoding. This
allows certain encodings to be shallower or to have fewer adequacy
side-conditions than their counterparts in our setting. For example,
in the statement of adequacy for our encoding of the specification
logic into the predicate {\sl seq} in Section~\ref{sec:adequacy-seq}
we have the requirement that $\nabla$-quantification is allowed only
at inhabited types. This is necessary since if $\tau$ were an
un-inhabited type then $\exists_\tau x . \top$ should not be provable
in the specification logic, and yet its encoding as a {\sl seq}
judgment is provable if $\nabla$-quantification is allowed at type
$\tau$. The issue is that the specification logic existential
quantifier is mapped to the meta-logic existential quantifier and the
latter allows instantiations containing any nominal constants even if
there are no other inhabitants at that type. If we take the definition
of {\sl seq} as being in $\mu$LJ then it should be an adequate
encoding of the specification logic without any conditions. Thus the
local context in the minimal approach provides an adequate
representation of the variable signature of an \hh sequent. To achieve
the same condition-less adequacy for \logic would require explicitly
carrying around a representation of the specification logic signature
and using this to restrict the type of instantiations for meta-logic
universal and existential quantifiers. This approach would require
more work due to the need to establish properties about the signature,
but this is the same work which is already required in the minimal
approach. Moreover, this explicit encoding of the signature would
allow one to directly analyze and interact with the signature (\eg,
quantifying over all signatures of a certain type) which is not
possible in the minimal approach.



\chapter{Conclusion and Future Work}
\label{ch:future-work}

This thesis has concerned the development of a framework for
specifying, prototyping, and reasoning about formal systems. The
specific framework that has been of interest has two defining
characteristics. First, it has been based on an intertwining of two
distinct logics for specification and for reasoning about specifications.
The specification logic has the property of also being executable,
thereby rendering descriptions written in it transparently into
prototypes of the formal systems that are encoded. The reasoning logic
has the capability of directly embedding the specification logic;
specifications themselves are represented indirectly through this
medium. This is, in fact, the style of encoding that is developed
here. The benefits of this approach are that the same specifications
can be used for prototyping and reasoning and generic properties of
the specification logic can be proved and used to advantage in
reasoning. The second important characteristic of our framework is
that uses a higher-order treatment of binding constructs, supporting
this approach in both the specification and the reasoning levels
through targeted logical devices.

The focus in this thesis has been on the reasoning component of the
above framework. In this context, we have developed the logic \logic
that provides the mechanism of fixed-point definitions that can also
be interpreted inductively or co-inductively and that has
sophisticated devices for dealing with higher-order representations of
syntactic constructs. An important component of this logic is the
notion of nominal abstraction that allows for the reflection into
definitions of properties of objects introduced into proofs in the
course of treating binding constructs. We have used \logic as the
basis of an interactive theorem prover called Abella and have explored
a two-level logic approach to reasoning about formal systems in its
context. This system has been applied to several interesting reasoning
examples and has yielded appealing solutions in most of these
situations.

While several promising results have been obtained in this thesis,
there remain many more interesting things still to be done. We sketch
below some possible ways in which the framework for specification,
prototyping, and reasoning that has been considered can be further
enriched. The kind of work involved in realizing these
possibilities ranges from foundational considerations for increasing
the expressive power of the meta-logic to more implementation oriented
efforts to better facilitate the reasoning process.

\section{More Permissive Stratification Conditions for Definitions}

The stratification condition for definitions in \logic is fairly
simplistic, and it rules out seemingly well-behaved definitions such
as the reducibility relation used in logical relations arguments (see
Section~\ref{sec:girards-strong-norm}). One could imagine a more
sophisticated condition which would allow definitions to be stratified
based on an ordering relation over the arguments of the predicate
being defined. The proof theoretic arguments needed to prove
cut-elimination for a logic with such definitions seem rather
delicate, particularly since we allow substitutions which may
interfere with any ordering based on term structure. From the
perspective of developing the theory for such an extension, a first
step might be to realize the addition to the Linc$^-$ logic
\cite{tiu.momigliano}. Given the way the cut-elimination proof for
\logic has been obtained from cut-elimination for Linc$^-$, if we can
successfully carry out such an extension to Linc$^-$, the desired
result relative to \logic might then follow easily.

There is also an interaction of this line of research with the
development of induction and co-induction. The strict notion of
stratification that \logic uses ensures that each definition describes
a single fixed-point and the induction and co-induction rules operate
on this structure. However, if we weaken the stratification condition,
then each definition can be viewed as a possibly infinite collection
of fixed-points. The rules for induction and co-induction must be
carefully adapted in light of this fact.

\section{Context Inversion Properties}

When reasoning about specification judgments we often need to describe
and utilize properties of the contexts in which those judgments are
formed. This takes the form of stating a definition describing those
contexts, proving various inversion lemmas about membership in those
contexts, and then applying these lemmas at the appropriate times.
Manually stating, proving, and using these lemmas introduces a fair
amount of overhead which seems mundane enough that we might want to
avoid it.

One option is to attack this problem with automation. One could
imagine automatically generating and proving inversion properties for
those definitions which can be seen as describing contexts. The
inversion properties follow directly from the definitions, and the
proofs are by simple inductive arguments. These lemmas could then be
automatically applied anytime we have a member of such a context.
However, it is unlikely that such automation of these properties would
be able to cope with more complicated properties of contexts such as
those used in the conversion between higher-order abstract syntax and
the de Bruijn representation (see Section~\ref{sec:conv-de-bruijn}).

\begin{figure}[t]
\begin{align*}
(\nabla x. \typeof {(L\ x)} x A) &\triangleq \nabla x.~\member {(\assm
  x A)} (L\ x) \\
\typeof L {(\app M N)} B &\triangleq \exists A.~\typeof L M (\arr A B)
\land \typeof L N A \\
\typeof L {(\abs A R)} (\arr A B) &\triangleq \nabla x.~ \typeof
{((\assm x A)::L)} {(R\ x)} B
\end{align*}
\caption{Typing judgment directly within \logic}
\label{fig:typeof}
\end{figure}

Another option would be to devise an alternate version of the
specification logic or of its encoding in the meta-logic so that such
context inversion properties are not needed as often. It is unclear
how such alternatives would be developed, but as an analogy, consider
the following. Typing for the simply-typed $\lambda$-calculus can be
defined directly within \logic via a definition of $(\typeof L M A)$
which holds when $M$ has type $A$ in the typing context $L$. The
clauses for this definition are presented in Figure~\ref{fig:typeof}.
Using nominal abstraction, this definition of typing directly
precludes the possibility of looking anything up in the context which
is not of the form $(\assm x A)$ for some nominal constant $x$. Thus
one does not need to deal with superfluous cases when performing case
analysis on a typing judgment. Note, however, that uniqueness
properties regarding the typing context would still need to handled
manually.

\section{Types and Explicit Typing}

The types in \logic play no role in reasoning except to restrict the
valid instantiations of quantifiers. Thus, for example, one cannot
directly perform induction or case analysis on a term based on its
type. Instead, one must create a definition which recognizes terms of
that type, and then use induction or case analysis on that definition.
This requires that one knows that the definition holds on the term,
which in turn may require carrying around more explicit typing
information in the specification or reasoning. All of this creates
overhead just to work effectively with types. For example, in
formalizing Girard's
proof of strong normalization for the simply-typed $\lambda$-calculus
(Section~\ref{sec:girards-strong-norm}) we had to create a specification
logic judgment which recognized well-formed types. This judgment was
then carried around during reasoning, and it even had to be put into
the specification of the object language typing judgment. We then had
to prove a lemma which said that an object language type could not
contain any nominal constants.

One possible solution is to attach explicit typing information to
every variable in the specification and in reasoning. Ideally this
should be done in such a way that the end user would not need to deal
with explicit typing information, but would be able to perform
operations like induction and case analysis based on the type of a
term. A major difficulty in such automation would be dealing with the
contexts needed to recognize terms which use higher-order abstract
syntax. Multiple terms may have different contexts which have a
particular relationship to each other which needs to be maintained. It
is not clear how such information could be succinctly expressed.

\section{Alternate Specification Logics}

One motivation for the two-level logic approach to reasoning is that
it lets us use
general properties of a specification logic in reasoning about
particular specifications. This approach has been successful relative
to the second-order hereditary Harrop formula logic. However,
different problem domains might require different specification
logics. For example, a {\it linear specification logic} that allows
for transient judgments has been found useful in characterizing
properties of hardware \cite{chirimar95phd} and programming languages
with references \cite{mcdowell02tocl}. One can imagine an extension of
the Abella system which allows different specification logics to be
plugged in and used as particular reasoning tasks demand. Given the
way our framework is designed, judgments from these different
specification languages would be able to co-exist during reasoning.

\section{Focusing and Proof Search}

Recent research has been looking at techniques for guiding proof
search in \logic-like logics based on the notion of {\it focusing}
\cite{baelde07lpar,baelde09focused}. These techniques allow the
automation of a significant portion of the reasoning process by
pruning redundant choices. For example, it was proven that if an
atomic judgment is to be inducted on during a proof, then this
induction can be done immediately. These techniques have been
effectively realized in the Tac theorem prover \cite{tac-website}.
The Abella system could also be extended to support this type of
automation. Moreover, one should investigate how this automation
interacts with the two-level logic approach to reasoning.

\section{An Integrated Framework}

The Teyjus system allows for animating descriptions in our
specification logic and the Abella system allows for reasoning about
such descriptions. It would be worthwhile to combine these systems
into an integrated framework which enables a more fluid relationship
between the processes of specification and reasoning. In its simplest
form, such an integration would allow the different aspects of
prototyping and reasoning to be invoked seamlessly from a common
description of a formal system. As an example of a deeper kind of
integration looked at from the perspective of the reasoning component,
uses of the $\defR$ and $\defL$ rules relative to the encodings of
specifications within \logic can draw benefit from computations within
the specification logic. An important issue to be tackled in
implementing such relationships would be that of designing an
interface that allows a smooth transition between the different
functionalities that Teyjus and Abella, the two currently separate
components of our framework, provide.



\bibliographystyle{alpha}
\bibliography{master}

\end{document}